\documentclass[11pt,reqno]{amsart}

\usepackage{amsmath, amsthm, amsopn, amssymb, graphicx, enumerate, geometry, afterpage, caption, subcaption, color, textcomp, stmaryrd, bbm}
\usepackage[english]{babel}

\newcommand{\grayscaleText}[2] {#2} 
\newcommand{\printable}{\grayscaleText{-printable}{}}
\newcommand{\grayscale}{\grayscaleText{grayscale}{}}
\newtheorem{thm}{Theorem}[section]
\newtheorem{lemma}[thm]{Lemma}

\newtheorem{definition}[thm]{Definition}
\newtheorem{cor}[thm]{Corollary}

\theoremstyle{definition}

\renewcommand{\subsubsection}[1] {\smallskip \noindent {\bf #1.}}
\newcommand{\reffig}[1] {\textsc{\ref{#1}}}

\newcommand{\cA}{\mathcal{A}}

\newcommand{\cC}{\mathcal{C}}
\newcommand{\cH}{\mathcal{H}}

\newcommand{\cS}{\mathcal{S}}
\newcommand{\cG}{\mathcal{G}}

\newcommand{\R}{\mathbb{R}}
\newcommand{\Z}{\mathbb{Z}}

\newcommand{\E}{\mathbb{E}}

\renewcommand{\Pr}{\mathbb{P}}
\newcommand{\cE}{\mathcal{E}}
\newcommand{\1}{\mathbbm{1}}

\newcommand{\intB}{\partial_{\bullet}}
\newcommand{\extB}{\partial_{\circ}}
\newcommand{\intextB}{\partial_{\ins \out}}

\newcommand{\ins}{\bullet}
\newcommand{\out}{\circ}
\newcommand{\outtwo}{\minuso} 

\newcommand{\up}{\,\uparrow\,}
\newcommand{\down}{\,\downarrow\,}

\newcommand{\Col}{\cC}
\newcommand{\FS}{\mathcal{K}}
\newcommand{\SO}{\mathcal{S}}
\newcommand\anomaly{*}
\DeclareMathOperator\rev{rev}
\newcommand{\Even}{\mathrm{Even}}
\newcommand{\Odd}{\mathrm{Odd}}
\DeclareMathOperator\MC{MC}
\DeclareMathOperator\diam{diam}

\DeclareMathOperator\dist{dist}
\newcommand{\distTV}{\mathrm{d_{TV}}}

\newcommand{\Unknown}[2] {A_{#1 \parallel #2}}
\newcommand{\DilemmaOdd}[2] {A_{\ins #1 #2}}
\newcommand{\DilemmaEven}[2] {A_{\out #1 #2}}
\newcommand{\DilemmaEvenTwo}[2] {A_{\outtwo #1 #2}}

\newcommand\refFSpartition{(\ref{def:four-section}\ref{it:def-FS-1})}
\newcommand\refFSodd{(\ref{def:four-section}\ref{it:def-FS-3})}
\newcommand\refFSisolated{(\ref{def:four-section}\ref{it:def-FS-5})}

\newcommand\refAFSvalues{(\ref{def:adapted-fs}a)}
\newcommand\refAFSneighbor{(\ref{def:adapted-fs}b)}

\newcommand\refFA[1]{(\ref{def:four-approx}\ref{it:FA#1})}

\newcommand\refLAAseparate{(\ref{def:level-1-approx}\ref{it:def-level-1-approx-separate})}
\newcommand\refLAAmanyRevealed{(\ref{def:level-1-approx}\ref{it:def-level-1-approx-many-revealed})}
\newcommand\refLAAsize{(\ref{def:level-1-approx}\ref{it:def-level-1-approx-size})}

\newcommand\refLBAsubset{(\ref{def:level-2-approx}\ref{it:def-weak-FA-1})}
\newcommand\refLBAanomaly{(\ref{def:level-2-approx}\ref{it:def-weak-FA-anomaly})}
\newcommand\refLBAmanyAnomalies{(\ref{def:level-2-approx}\ref{it:def-weak-FA-many-anomalies})}
\newcommand\refLBAseesManyi{(\ref{def:level-2-approx}\ref{it:def-weak-FA-sees-many-i})}
\newcommand\refLBAsize{(\ref{def:level-2-approx}\ref{it:def-weak-FA-size})}

\newcommand\refLCAextendsLBA{(\ref{def:level-3-approx}\ref{it:def-weak-FA-extends-rough-FA})}

\newcommand\refLCAsqrtdTight{(\ref{def:level-3-approx}\ref{it:def-weak-FA-sqrt-d-approx})}

\newcommand\refLDAextendsLCA{(\ref{def:strong-FA}\ref{it:def-strong-FA-extends-weak-FA})}

\newcommand\refLDAodd{(\ref{def:strong-FA}\ref{it:def-strong-FA-odd})}
\newcommand\refLDAanomalyZeroThree{(\ref{def:strong-FA}\ref{it:def-strong-FA-anomaly03})}
\newcommand\refLDAanomalyOneTwo{(\ref{def:strong-FA}\ref{it:def-strong-FA-anomaly12})}

\makeatletter
\newcommand*\rel@kern[1]{\kern#1\dimexpr\macc@kerna}
\newcommand*\widebar[1]{%
  \begingroup
  \def\mathaccent##1##2{%
    \rel@kern{0.8}%
    \overline{\rel@kern{-0.8}\macc@nucleus\rel@kern{0.2}}%
    \rel@kern{-0.2}%
  }%
  \macc@depth\@ne
  \let\math@bgroup\@empty \let\math@egroup\macc@set@skewchar
  \mathsurround\z@ \frozen@everymath{\mathgroup\macc@group\relax}%
  \macc@set@skewchar\relax
  \let\mathaccentV\macc@nested@a
  \macc@nested@a\relax111{#1}%
  \endgroup
}
\makeatother

\newcommand{\comment}[1] {}

\definecolor{fsc0}{RGB}{227,208,119}
\definecolor{fsc1}{RGB}{222,67,54}
\definecolor{fsc2}{rgb}{0.35,0.5,0.85}
\definecolor{fsc3}{RGB}{152,46,122}
\definecolor{fsc0grayscale}{gray}{0.91}
\definecolor{fsc1grayscale}{gray}{0.75}
\definecolor{fsc2grayscale}{gray}{0.60}
\definecolor{fsc3grayscale}{gray}{0.45}
\definecolor{val0color}{RGB}{240,240,240}
\definecolor{val1color}{RGB}{205,235,190}
\definecolor{val2color}{RGB}{180,200,245}
\definecolor{val0colorgrayscale}{gray}{0.9}
\definecolor{val1colorgrayscale}{gray}{0.75}
\definecolor{val2colorgrayscale}{gray}{0.6}
\definecolor{anomaly-color}{rgb}{1,0.5,0}
\definecolor{out-of-phase-bdry-color}{rgb}{0.28,0.68,0.36}
\definecolor{out-of-phase-inner-color}{rgb}{0,0.4,0.1}
\definecolor{anomaly-colorgrayscale}{gray}{0.3}
\definecolor{out-of-phase-bdry-colorgrayscale}{gray}{0.75}
\definecolor{out-of-phase-inner-colorgrayscale}{gray}{0.55}

\title{Long-range order in the 3-state antiferromagnetic Potts model in high dimensions}
\date{\today}

\author{Ohad N. Feldheim}
\address{Ohad N. Feldheim\hfill\break
    Hebrew University of Jerusalem\\
    Einstein Institute of Mathematics\\
    Givat Ram, Jerusalem, Israel.}
\email{ohad.feldheim@mail.huji.ac.il}
\thanks{Research of O.F. was conducted at Tel-Aviv University and the IMA. Supported in part by the Institute for Mathematics and its Applications with funds provided by the National Science
Foundation}

\author{Yinon Spinka}
\address{Yinon Spinka\hfill\break
    Tel Aviv University\\
    School of Mathematical Sciences\\
    Tel Aviv, 69978, Israel.}
\email{yinonspi@post.tau.ac.il}
\urladdr{http://www.math.tau.ac.il/~yinonspi}
\thanks{Research of Y.S. supported by Israeli Science Foundation grant 1048/11, Marie Sk\l{}odowska-Curie IRG grant SPTRF, and the Adams Fellowship Program of the Israel Academy of Sciences and Humanities}

\begin{document}

\begin{abstract}
We prove the existence of long-range order for the $3$-state Potts antiferromagnet at low temperature on $\Z^d$ for sufficiently large $d$. In particular, we show the existence of six extremal and ergodic infinite-volume Gibbs measures, which exhibit spontaneous magnetization in the sense that vertices in one bipartition class have a much higher probability to be in one state than in either of the other two states. This settles the high-dimensional case of the Koteck{\'y} conjecture.
\end{abstract}

\maketitle

\section{Introduction}

The $q$-state Potts model is a classical model in statistical mechanics, which generalizes the Ising model by allowing more than two states. A special case was first considered by Ashkin and Teller in 1943 \cite{ashkin1943statis}, while the general model was proposed by Domb and published by Potts in 1951 \cite{Potts1952somege}.
Since the late 1970's, the model has drawn substantial attention from mathematicians and physicists alike, largely because it proved to be very rich, displaying a much wider spectrum of phenomena than the simpler Ising model. For an extensive survey of classical results on the Potts model, see Wu~\cite{Wu1982thepot}.

While the ferromagnetic regime of the model is by now relatively well understood, the picture for the antiferromagnetic regime is far from complete.
Here we consider the $3$-state antiferromagnetic (AF) Potts model on the integer lattice $\Z^d$.
The model assigns a random value $f(v) \in\{0,1,2\}$ to each vertex $v$ in some domain $\Lambda \subset \Z^d$, favoring different states on adjacent vertices.
The probability of any given configuration $f$ is proportional to $\exp(-\beta H_f)$, where $\beta \ge 0$ is a real parameter and
\[ H_f:=\sum_{\{u,v\}} \1_{f(u)=f(v)} ,\]
where the sum is taken over all pairs of nearest neighbors.
In other words, $f$ is distributed according to the Boltzmann distribution with the Hamiltonian $H_f$ at inverse temperature $\beta$.

At infinite temperature ($\beta=0$), the values assigned to different vertices are independent, and the model is completely disordered. The Dobrushin uniqueness condition \cite{Dobrushin1968TheDe} guarantees that, in any dimension, disorder persists at sufficiently high temperature (small $\beta$).
A fundamental question is whether at low temperature (large $\beta$) the model remains disordered or, instead, undergoes a phase transition into an ordered phase.
In the latter case, it is also desirable to understand the structure of a typical ordered configuration.
Such a phase transition does not occur in every dimension (e.g., for $d=1$), and it is interesting to determine in which dimensions (if any) it does.


Following an earlier debate in the physical community (see e.g.~\cite{Banavar1980Order} where a continuous transition was conjectured),
Koteck{\'y} conjectured circa 1985 (implied in~\cite{kotecky1985long}, also mentioned in~\cite{galvin2012phase}) that in high dimensions, possibly already in three dimensions, the $3$-state AF Potts model indeed undergoes a phase transition, and that at sufficiently low temperature, a configuration typically follows one of six patterns. To understand the nature of these patterns, note first that the graph $\Z^d$ is bipartite, and that each bipartition class forms a sublattice. In a typical large-volume disordered configuration, each value is assigned to roughly one-third of the vertices in each of the two sublattices.
In contrast, in high dimensions it is conjectured that at sufficiently low temperature the model adheres to one of
six phases, sometimes called broken-sublattice-symmetry (BSS) phases. These phases are characterized by having each state assigned to more than one-third of one sublattice and to less than one-third of the other.
These conjectures were later supported by Monte-Carlo simulations~\cite{wang1990three} and by mean-field arguments~\cite{Itakura1996Mean}.

The Koteck{\'y} conjecture has proven difficult to verify.
The difficulty originates partly from the fact that the model exhibits non-vanishing \emph{residual entropy}, meaning that in the extreme case of $\beta=\infty$ (i.e., at zero-temperature), the number of configurations is exponentially large in the size of the domain.
In fact, even in this case, in which the model consists of a uniformly chosen proper $3$-coloring of the domain, demonstrating the existence of six BSS phases is highly non-trivial.
Already in~\cite{kotecky1985long}, Koteck{\'y} observed that the problem resists standard Peierls arguments, and suggested looking at correspondences with other models as a possible approach for tackling it.
The existence of six BSS phases was recently verified in the zero-temperature case in high dimensions by Peled~\cite{peled2010high} and, independently, by Galvin, Kahn, Randall and Sorkin~\cite{galvin2012phase}. Both groups obtained their results through highly sophisticated contour methods. Their techniques, however, rely heavily on the special topological structure of proper $3$-colorings on $\Z^d$ and cannot be used to show the BSS structure at positive temperature.

In this work, we prove the high-dimensional case of the Koteck{\'y} conjecture, showing that for sufficiently high $d$ there exists a positive temperature below which the model exhibits six BSS phases. Our methods also improve the quantitative sublattice bias estimates obtained in \cite{peled2010high} and \cite{galvin2012phase}
for the zero-temperature case.

\subsection{Main result}\label{sec:mainres}

We begin by defining the finite-volume model with boundary conditions.
For a finite set $\Lambda \subset \Z^d$ and a map $\tau \colon \Z^d \to \{0,1,2\}$, write $\cC_\Lambda^\tau$ for the collection of maps $f\colon \Z^d\to\{0,1,2\}$ which agree with $\tau$ on $\Z^d \setminus \Lambda$.
The $3$-state AF Potts model
in volume $\Lambda$ with boundary conditions $\tau$ at inverse-temperature $\beta \ge 0$ is the probability measure on $\cC_\Lambda^\tau$ defined by
\[ \mu^\tau_{\Lambda,\beta}(f) := \frac{e^{-\beta H_\Lambda(f)}}{Z^\tau_{\Lambda,\beta}}, \quad f \in \cC_\Lambda^\tau, \]
where
\[ H_\Lambda(f) := \sum_{\substack{ \{u,v\} \in E(\Z^d) \\ \{u,v\} \cap \Lambda \neq \emptyset}} \1_{f(u)=f(v)} \]
and $Z^\tau_{\Lambda,\beta}$, \emph{the partition function} of the model, is a normalizing constant.
Thinking of a configuration as a coloring, we also refer to states as colors, and say that an edge is \emph{proper} if its endpoints take different colors, and that it is \emph{improper} otherwise. Thus, the Hamiltonian counts the number of improper edges touching $\Lambda$.

We call a vertex in $\Z^d$ \emph{even} (odd) if its graph-distance from the origin is even (odd). We denote the set of even and odd vertices by $\Even$ and $\Odd$, respectively. We say that a subset $\Lambda \subset \Z^d$ is a \emph{domain} if it is finite, non-empty, connected and its complement is connected (e.g., a box $\{-N,\dots,N\}^d$).
We write $\extB \Lambda$ for the external boundary of $\Lambda$, i.e., the set of vertices outside $\Lambda$ which are adjacent to a vertex in $\Lambda$.
The reader should note that the marginal distribution of $\mu^\tau_{\Lambda,\beta}$ on $\Lambda$ depends on $\tau$ only through $\tau|_{\extB \Lambda}$.

We shall use the term \emph{even-$i$ (odd-$i$) boundary conditions}, for $i \in \{0,1,2\}$, to describe any pair $(\Lambda,\tau)$ such that $\Lambda$ is domain, $\extB \Lambda$ consists only of even (odd) vertices and $\tau^{-1}(i)$ is precisely the set $\Even$ ($\Odd$).
To simplify the statements, all our results are stated for even-$0$ boundary conditions, although, by symmetry, they hold in all six cases.

Our first result is a verification of the existence of the conjectured long-range order at sufficiently low temperature.
A macroscopic analogue concerning the empirical percentage of vertices taking each color is given in Corollary~\ref{cor:bias} below.

\begin{thm}\label{thm:main}
There exist constants $C,c,d_0>0$ such that for any $d \ge d_0$ and $\beta \ge C \log d$ the following holds.
Let $(\Lambda,\tau)$ be even-0 boundary conditions and let $u \in \Z^d$ be even and $v \sim u$. Then
\begin{enumerate}[\qquad(a)]
	\item $\mu^\tau_{\Lambda,\beta}(f(u) \neq 0) \le e^{-cd}$.
	\item $\mu^\tau_{\Lambda,\beta}(f(v) = 0) \le e^{-cd^2} + e^{-\beta d}$.
	\item $\mu^\tau_{\Lambda,\beta}(f(u)=f(v)) \le e^{-cd-\beta}$.
\end{enumerate}
\end{thm}
%
%

The first two parts of Theorem~\ref{thm:main} show that under even-$0$ boundary conditions the random coloring is rather rigid, tending to follow a particular pattern, which we call the even-$0$ pattern.
Specifically, an even vertex is likely to be colored $0$, while an odd vertex is likely to be colored either $1$ or $2$ (see Figure~\reffig{fig:samples}).
Moreover, the probability that any bounded number of vertices
conforms to this pattern tends to one as the dimension tends to infinity.
The third part of the theorem is of a slightly different nature as it is concerned with the unlikeliness of improper edges.
Observe that on the event $\{f(u)=f(v)\}$, either $u$ or $v$ violates the even-$0$ pattern. Thus, the first two parts of the theorem imply that in high dimensions it is unlikely for a given edge to be improper (as a function of $d$). Since every improper edge reduces the weight of a configuration by a factor of $e^{-\beta}$, it is no surprise that this is also unlikely as a function of $\beta$, but we do not know an \emph{elementary} argument for showing this.
We remark that one may show that the bounds in Theorem~\ref{thm:main} are tight up to the constants in the exponents.

In fact, our methods allow us to show that in high enough dimensions, violations of the even-0 pattern do not percolate (see again Figure~\reffig{fig:samples}).
Denote by $T(f)$ the set of vertices which violate the even-$0$ pattern, i.e., even vertices $u$ and odd vertices $v$ such that $f(u) \neq 0$ and $f(v)=0$.
Observe that the \emph{singularities} of $f$, i.e., the endpoints of improper edges, are all contained in $T(f) \cup \extB T(f)$.
Let $B(f,v)$ be the connected component of $v$ in 
$T(f) \cup \extB T(f)$.
The \emph{diameter} of a finite connected set $U \subset \Z^d$, denoted by $\diam U$, is the maximum graph-distance between two vertices in $U$.

\begin{thm}\label{thm:main2}
	There exist constants $C,c,d_0>0$ such that for any $d \ge d_0$ and $\beta \ge C \log d$ the following holds.
	Let $(\Lambda,\tau)$ be even-0 boundary conditions and let $v \in \Z^d$. Then, for any $k \ge 1$,
	\[ \mu^\tau_{\Lambda,\beta}(|B(f,v)| \ge k) \le e^{-c k^{1-1/d}} \quad\text{ and }\quad \mu^\tau_{\Lambda,\beta}(\diam B(f,v) \ge k) \le e^{-c d k} .\]
\end{thm}

\begin{figure}
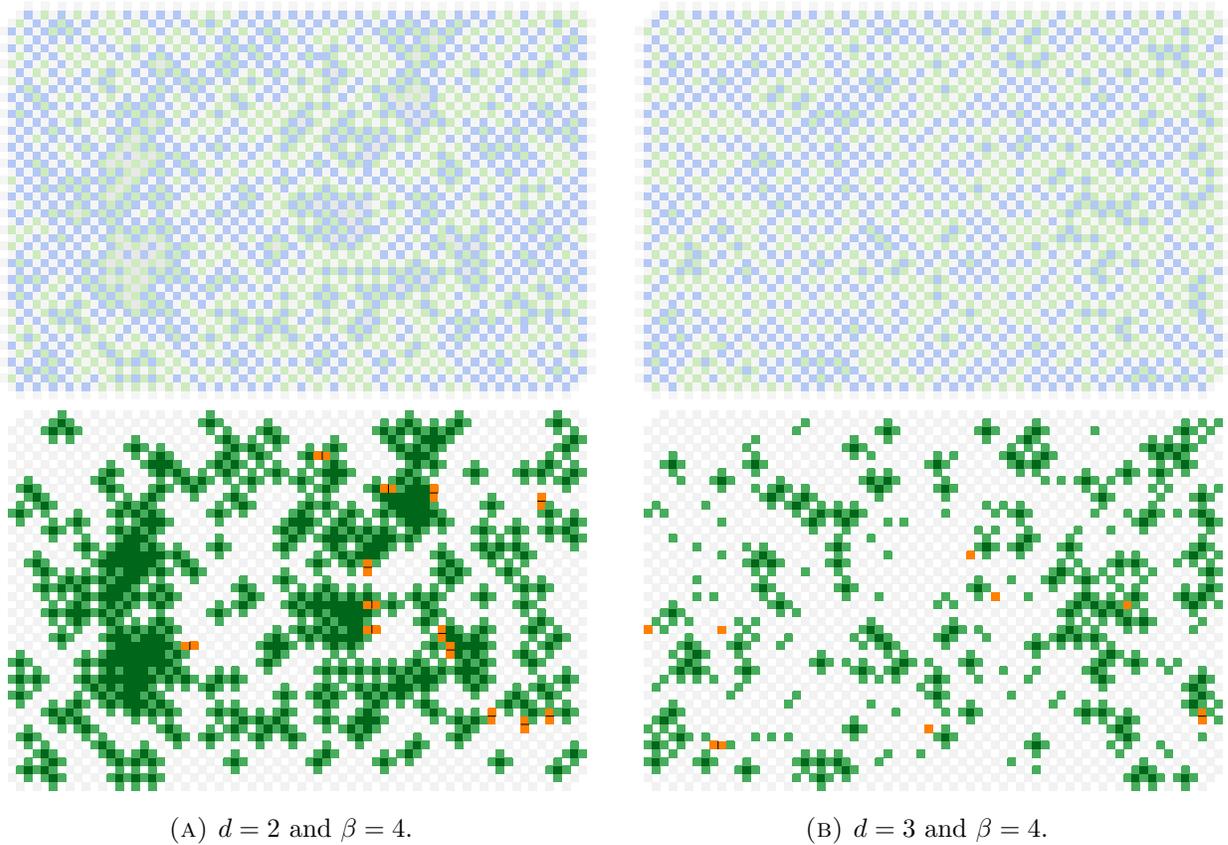

	\centering
	\captionsetup{width=0.86\textwidth}
	\begin{subfigure}[t]{.45\textwidth}
		\includegraphics[scale=0.11]{coloring-sample1-2d-beta=4-values\printable.pdf}
		\includegraphics[scale=0.11]{coloring-sample1-2d-beta=4\printable.pdf}
		\caption{$d=2$ and $\beta=4$.}
		\label{fig:sample-2d}
	\end{subfigure}%
	\begin{subfigure}{20pt}
		\quad
	\end{subfigure}%
	\begin{subfigure}[t]{.45\textwidth}
		\includegraphics[scale=0.11]{coloring-sample1-3d-beta=4-values\printable.pdf}
		\includegraphics[scale=0.11]{coloring-sample1-3d-beta=4\printable.pdf}
		\caption{$d=3$ and $\beta=4$.}
		\label{fig:sample-3d}
	\end{subfigure}
	\caption{Two samples of random $3$-colorings in dimensions two and three. On the left, a sample on a $72 \times 48$ grid. On the right, the middle slab of a sample on a $72 \times 48 \times 72$ grid. Colorings are sampled under even-$0$ boundary conditions via Glauber dynamics for $10$ billion iterations started from a random pure coloring ($0$ at even vertices and $1$ or $2$ at odd vertices). On the top, the sampled colorings are depicted ({\color{val0color\grayscale}$\blacksquare$}=0, {\color{val1color\grayscale}$\blacksquare$}=1, {\color{val2color\grayscale}$\blacksquare$}=2). On the bottom, an encoding of the coloring above it is shown: {\color{anomaly-color\grayscale}$\blacksquare$} represents singularities (vertices having a neighbor with the same color), {\color{out-of-phase-inner-color\grayscale}$\blacksquare$} represents the set $T(f)$ of vertices violating the even-$0$ pattern, and {\color{out-of-phase-bdry-color\grayscale}$\blacksquare$} represents their neighbors $\extB T(f)$. Theorem~\ref{thm:main} shows that in high dimensions it is unlikely for any given vertex to violate the $0$-pattern. Moreover, Theorem~\ref{thm:main2} shows that the connected component of any given vertex in $T(f) \cup \extB T(f)$ is unlikely to be large.}
	\label{fig:samples}
\end{figure}

\subsection{Gibbs states}\label{sec:gibbs}
The $3$-state AF Potts model can be extended to the entire lattice through \emph{infinite-volume Gibbs measures}.
Here, we recall the definition of these measures, state a convergence result for the finite-volume measures, and use ergodicity arguments combined with Theorem~\ref{thm:main} to show the existence of six BSS Gibbs measures.

Gibbs measures are defined through the Dobrushin--Lanford--Ruelle conditions (see, e.g.,~\cite{georgii2011gibbs}) as follows.
Let $\mu$ be a probability measure on $\cC_{\Z^d} := \{0,1,2\}^{\Z^d}$ and let $f$ be sampled according to $\mu$. The measure $\mu$ is said to be an \emph{(infinite-volume) Gibbs measure} for the $3$-state
AF Potts model at inverse temperature $\beta \in (0,\infty]$
if for any domain $\Lambda$ and $\mu$-almost every $\tau$, the distribution of the coloring $f$, conditioned that $f\in\cC_\Lambda^\tau$, is given by $\mu_{\Lambda,\beta}^\tau$.
For a discussion of the $\beta=\infty$ case, see Section~\ref{sec:rem-and-open} below.

 General compactness arguments show that at least one Gibbs measure always exists. A simple recipe for producing such a measure is to pick a sequence of boundary conditions $(\Lambda_n,\tau_n)$ such that $(\Lambda_n)_{n=1}^\infty$ increases to $\Z^d$. Any (weak) subsequential limit of $(\mu^{\tau_n}_{\Lambda_n,\beta})_{n=1}^{\infty}$ would converge to a Gibbs measure. To show that there are multiple Gibbs measures, one may demonstrate that different choices of $(\Lambda_n,\tau_n)$ result in different limiting measures. The fact that this is the case in the low-temperature high-dimensional $3$-state AF Potts model is a direct consequence of Theorem~\ref{thm:main}.

In fact, we show that it is not necessary to take subsequential limits in order to obtain convergence, as long as all the boundary conditions $(\Lambda_n,\tau_n)$ are of some fixed type (e.g., even-0).
Moreover, we are able to show that the limiting Gibbs measures in this case are strongly mixing with respect to every parity-preserving translation of $\Z^d$, i.e., translations preserving the even and odd sublattices.

A \emph{cylinder event} is a set of the form $A \times \{0,1,2\}^{\Z^d \setminus U}$ for some finite $U \subset \Z^d$ and $A \subset \{0,1,2\}^U$.
Let $\mu$ be a probability measure on $\cC_{\Z^d}$ and let $T\colon\Z^d\to\Z^d$ be a \emph{measure-preserving} transformation, i.e., $\mu(T^{-1}A)=\mu(A)$ for all cylinder events $A$.
The measure $\mu$ is said to be \emph{strongly mixing} with respect to $T$ if for every two cylinder events $A$ and $B$, we have
\begin{equation}\label{eq:mix}
 \lim_{n\to\infty}\mu(A\cap T^{-n}B)=\mu(A)\mu(B).
\end{equation}

\begin{thm}\label{thm:convergence-of-finite-volume-measure}
There exist constants $C,d_0>0$ such that for any $d \ge d_0$ and $\beta \ge C \log d$ the following holds.
Let $(\Lambda_n,\tau_n)$ be a sequence of even-0 boundary conditions and assume that $(\Lambda_n)$ increases to $\Z^d$.
Then the measures $\mu^{\tau_n}_{\Lambda_n,\beta}$ converge (weakly) as $n \to \infty$ to an infinite-volume Gibbs measure $\mu^{0,0}_{\Z^d,\beta}$, which is an extreme point of the set of all infinite-volume Gibbs measures and which is strongly mixing with respect to any parity-preserving translation.
\end{thm}


It follows that $\mu^{0,0}_{\Z^d,\beta}$ does not depend on the specific choice of domains $(\Lambda_n)$ as one may interleave two such sequences to obtain another convergent sequence.
By symmetry, for any $i \in \{0,1,2\}$ and $j \in \{0,1\}$, taking even-$i$ or odd-$i$ boundary conditions, according to whether $j=0$ or $j=1$, one obtains a limiting measure $\mu^{i,j}_{\Z^d,\beta}$.
Consequently, it follows that $\mu^{i,j}_{\Z^d,\beta}$ is invariant under parity-preserving automorphisms of $\Z^d$.
However, as Theorem~\ref{thm:main} implies that these measures are distinct for different $(i,j)$, they are not invariant under {\em all} automorphisms.
In the course of the proof, we further obtain that the rate of convergence in~\eqref{eq:mix}
is in fact exponential. In the special case of spin-spin correlations, for example, this means that the covariance between any two events of the form $\{f(u)=a\}$ and $\{f(v)=b\}$ decays exponentially in the distance between $u$ and $v$ (see Lemma~\ref{lem:almost-independence-of-colorings} for more details).



It is well-known that strong mixing implies ergodicity.
Thus, the measures $\mu^{i,j}_{\Z^d,\beta}$ are ergodic with respect to parity-preserving translations.
In physical terminology (see, e.g., \cite{georgii2011gibbs}), such a Gibbs measure is often called a \emph{pure state}.
It would be interesting to determine whether the six measures $\{\mu^{i,j}_{\Z^d,\beta}\}_{i \in \{0,1,2\},j \in \{0,1\}}$ are the only pure states.

For a positive integer $n$, denote $\Lambda_n:=\{-n,\dots,n\}^d$. Birkhoff's pointwise ergodic theorem (see, e.g., \cite[Theorem 2.1.5]{keller1998equilibrium}) implies that if $\mu$ is an ergodic probability measure on $\cC_{\Z^d}$ (with respect to parity-preserving translations) then, $\mu$-almost surely, the percentage of odd (even) vertices in $\Lambda_n$ colored $i$ converges to $\mu(f(v)=i)$ as $n \to \infty$, where $v$ is any odd (even) vertex. In particular, Theorem~\ref{thm:main} and Theorem~\ref{thm:convergence-of-finite-volume-measure} imply the following result.

\begin{cor}\label{cor:bias}
There exist constants $C,d_0>0$ such that for any $d \ge d_0$ and $\beta \ge C \log d$ the following holds.
Let $f$ be sampled according to $\mu^{0,0}_{\Z^d,\beta}$.  Then, almost surely,
\begin{align*}
\lim_{n\to\infty} \frac{|\Even \cap \Lambda_n \cap f^{-1}(0)|}{|\Even \cap \Lambda_n|} &~>~ \frac{1}{3} ~>~ \lim_{n\to\infty} \frac{|\Odd \cap \Lambda_n \cap f^{-1}(0)|}{|\Odd \cap \Lambda_n|},\\
\lim_{n\to\infty} \frac{|\Even \cap \Lambda_n \cap f^{-1}(i)|}{|\Even \cap \Lambda_n|} &~<~ \frac{1}{3} ~<~ \lim_{n\to\infty} \frac{|\Odd \cap \Lambda_n \cap f^{-1}(i)|}{|\Odd \cap \Lambda_n|}, \quad i \in \{1,2\}.
\end{align*}
\end{cor}

Thus, Corollary~\ref{cor:bias} verifies the existence of six BSS Gibbs measures.

\subsection{Remarks and open problems.}\label{sec:rem-and-open}
In this section, we make a couple of remarks which tie our work to previous works, and offer directions for future research.

\medbreak
\noindent{\bf The zero-temperature case.} Formally, the $\beta=\infty$ model is defined by
\begin{equation}\label{eq:beta-infty-def}
 \mu^\tau_{\Lambda,\infty}(f) :=
\lim_{\beta\to\infty} \mu^\tau_{\Lambda,\beta}(f) .
\end{equation}
Thus, $\mu^\tau_{\Lambda,\infty}$ is the uniform distribution on colorings $f \in \cC_\Lambda^\tau$ that minimize $H_\Lambda(f)$.
Note that for some $\tau$ this minimum energy might be strictly positive, i.e.,
there might not exist proper $3$-colorings of $\Lambda$ that are consistent with the boundary conditions $\tau$.

All of our results hold also for $\beta=\infty$.
Theorem~\ref{thm:main} offers a quantitative improvement in this case in comparison to the results
of~\cite{peled2010high} and~\cite{galvin2012phase}, improving by a factor of $\log d$ in the exponent obtained in~\cite[Theorem~2.10]{peled2010high} (or by a factor of $d$ in the exponent compared to the result in~\cite[Theorem~1.1]{galvin2012phase}).
Theorem~\ref{thm:convergence-of-finite-volume-measure} is novel also for $\beta=\infty$.

\medbreak
\noindent{\bf The critical inverse-temperature.}
Let $\beta_c(d)$ be the supremum of those inverse-temperatures $\beta$ for which the 3-state AF Potts model admits a unique Gibbs state.
Our results show that $\beta_c(d) \le C \log d$ for $d$ high enough.
It is plausible that the model admits a unique Gibbs state for all $\beta < \beta_c(d)$. However, in the absence of a monotonicity argument, this remains unknown.
We therefore define $\beta'_c(d) \le \beta_c(d)$ to be the infimum of those inverse-temperatures $\beta$ for which the model admits multiple Gibbs states. The Dobrushin uniqueness condition~\cite{Dobrushin1968TheDe} implies that $\beta'_c(d) \ge c/d$ for all $d \ge 1$.
It very well may be the case that this lower bound is tight and that $\beta_c(d) = \Theta(1/d)$.

%
%

\medbreak
\noindent{\bf Other boundary conditions.}
In \cite{galvin2007torpid},~\cite{feldheim2013rigidity} and~\cite{galvin2012phase}, the authors consider the $\beta=\infty$ case under periodic boundary conditions, i.e., proper $3$-colorings of a bounded even torus.
This is appealing since it is a very symmetric setting in which there is no need to fix boundary conditions, and yet, it is much more accessible than the general setting of free boundary conditions.
We may also consider the $\beta<\infty$ case under periodic boundary conditions.
A superficial inspection of this setting seems to indicate that
our methods could be extended to obtain a result similar to Theorem~\ref{thm:main}, namely
that in high dimensions and at sufficiently low temperature, typically one sublattice of the torus is dominated by one color, while the other sublattice is dominated by the two remaining colors.
However, we do not pursue this direction here and the details remain to be verified.

Other interesting boundary conditions are the so-called Dobrushin boundary conditions (first introduced in~\cite{dobrushin1973gibbs} for the Ising model).
Let $\Pi := \{ x \in \Z^d : x_1 \ge 0 \}$ be a half-space. Consider domains $\Lambda'_n \supset \{-n,\dots,n\}^d$ such that
$\extB\Lambda'_n \cap \Pi \subset \Even$ and $\extB\Lambda'_n \setminus \Pi \subset \Odd$. Fix the boundary conditions $\tau$ to be $0$ on $\Pi$ and $1$ on $\Z^d\setminus \Pi$.
We conjecture that when $d$ and $\beta$ are large enough,
the measures $\mu^\tau_{\Lambda'_n,\beta}$ converge as $n \to \infty$ to a Gibbs measure $\mu^\tau_{\Z^d,\beta}$,
which is \emph{not} invariant with respect to parity-preserving translations that do not preserve $\Pi$.
Roughly speaking, this measure should induce an interface between an even-$0$ phase and an odd-$1$ phase,
which fluctuates to a bounded distance from most points on the boundary of $\Pi$.
An analogous phenomenon is known to occur in the three-dimensional Ising model (see~\cite{van1975interface} for a short proof) and in the ferromagnetic Potts model~\cite{gielis2002rigidity}.
Showing that this holds for the AF Potts model is an
interesting open problem. It seems that our methods are currently insufficient to tackle the problem,
but potentially more powerful extensions could be devised.

\medbreak
\noindent{\bf Glauber dynamics.}
The Boltzmann distribution is the stationary distibution of Glauber dynamics.
This is the name associated to several natural Markov chains on graph colorings, which involve resampling the
color of a random vertex, taking the rest of the coloring as boundary conditions.
Glauber dynamics are of interest both as a simulation of the physical dynamics of the spin system
and as a technical tool for estimating the partition function of the model~\cite{jerrum1986random}. For a survey of results on
Glauber dynamics on the $3$-state AF Potts model at zero-temperature, see~\cite{galvin2012phase}.

It is conjectured that on a cubic domain of side-length $n$ in high dimensions,
at high temperature, under periodic boundary conditions, the mixing time of Glauber dynamics for the $3$-state AF Potts model
is polynomial in $n^d$, while at low temperature it should behave like $\exp(n^{d-1})$.
Indeed, mixing time of order $\exp(n^{d-1})$ at zero-temperature was obtained in~\cite{galvin2007torpid}, by showing that it is difficult to transition from any typical even-0 phase sample to a typical odd-0 phase sample.
We believe that by applying our methods to periodic boundary conditions, it should be possible to obtain a similar result at positive temperature.
An interesting challenge is to understand Glauber dynamics under even-$0$ boundary conditions.
In this setting, it is possible that the mixing time is polynomial at any temperature.

\medbreak
\noindent{\bf Four or more states.}
The Koteck{\'y} conjecture has a natural extension to $q \ge 4$ states. Namely, it is believed that in sufficiently
high dimensions and sufficiently low temperature, a typical sample of the $q$-state AF Potts model has one bipartition class of $\Z^d$ mostly populated by
$\lfloor q/2\rfloor$ states and the other by the remaining $\lceil q/2\rceil$ states. This conjecture has been
posed also in~\cite{galvin2012phase} and in~\cite{feldheim2013rigidity} and is perhaps the most important open problem
regarding the AF Potts model, unresolved even for the zero-temperature case of proper $q$-colorings.
Although our methods use the special structure of $3$-colorings less than their predecessors,
they do not seem to directly apply even to the case $q=4$.
Assuming that this is the case, one may further try to determine the minimal dimension $d_0(q)$ for which such a phase transition occurs. It follows from the Dobrushin uniqueness condition~\cite{Dobrushin1968TheDe} that $d_0(q) \ge cq$.

\subsection{Discussion.}

The $3$-state AF Potts model has been a subject of interest in mathematics and mathematical physics since the late 60's.
In 1967, Lieb~\cite{lieb1967residual} calculated the entropy constant of the model on $\Z^2$ at $\beta=\infty$ (also known as the square-ice model).
In 1982, Baxter~\cite{baxter2007exactly} was able to show critical behavior in this case by mapping the model to a staggered six-vertex model which admits an exact (albeit non-rigorous) solution.
Following Baxter's results and the introduction of the Koteck{\'y} conjecture, an effort was made to better understand the behavior of the model in two and three
dimensions. Building upon Baxter's work, Saleur~\cite{Saleur1990Zeros} obtained the phase diagram for the AF Potts model in two dimensions, while the efficient Wang--Swendsen--Koteck{\'y} cluster-flip Monte Carlo algorithm \cite{wang1990three} was used to verify the conjecture empirically, as well as other critical exponent predictions \cite{salas1998three,ferreira1999antiferromagnetic}.
While these remarkable developments advanced the understanding of the model,
a mathematically rigorous proof for the Koteck{\'y} conjecture seemed out of reach.

For the sake of the discussion, we limit ourselves to box-shaped domains of side-length $m$ in dimension $d$.
We are interested in the \emph{rigidity} of the $3$-state AF Potts model, manifested by the fact that a typical coloring predominantly follows a pattern. The Koteck{\'y} conjecture is concerned with rigidity of this model in the \emph{thermodynamic limit}, i.e., when $d$ is fixed and $m$ tends to infinity. One may also consider rigidity in the \emph{hypercube} setting, in which $m$ is fixed
and $d$ tends to infinity. While formally rigidity results in this setting have no implications for the thermodynamic limit, and in particular, tell us nothing about phase coexistence, they do provide some insight into the model and may suggest rigidity also in the stronger sense.

Rigidity has been an object of study in many related models, two of which are \emph{homomorphism height functions (HHFs)} and
the \emph{hard-core model}. These models are intimately related to uniformly chosen proper $3$-colorings (the zero-temperature $3$-state AF Potts model).
An HHF is a graph homomorphism from a finite bipartite graph to $\Z$.
The study of a uniformly chosen HHF was initiated by Benjamini, H\"aggstr\"om and Mossel~\cite{Benjamini2000},
who proved several correlation inequalities and posed conjectures about this model. In fact, on a simply connected domain,
under suitable boundary conditions, a uniformly chosen HHF is equivalent to a uniformly chosen proper $3$-coloring of the domain (see e.g.~\cite{peled2010high}).
The hard-core model consists of a random independent set $I$ in a finite graph (i.e., a set containing no two adjacent vertices).
The probability of each such $I$ is proportional to $\lambda^{|I|}$, where $\lambda>0$ is a parameter called \emph{fugacity} or \emph{activity}.
In the special case $\lambda=1$, the model reduces to a uniformly chosen independent set. Note that in a proper coloring, the set of vertices
taking any given color is an independent set. Although there is no direct relation between this model and proper 3-coloring, several results on the hard-core model preceded and inspired counterparts
for the zero-temperature $3$-state AF Potts model.

The rigid structure of a uniformly chosen independent set on the hypercube (with $m=2$)
follows already from the work of Korshunov and Sapozhenko~\cite{Korshunov1983Th}. Rigidity of the hard-core model on the hypercube was shown by Kahn~\cite{kahn2001entropy} (tight bounds were later obtained by Galvin~\cite{galvin2011threshold} using more advanced methods).
Kahn's methods rely on entropy considerations which are less involved and easier to generalize than the methods used for showing rigidity in the thermodynamic limit.
Using these techniques Kahn showed in~\cite{Kahn2001hypercube} that the probability that a uniformly chosen HHF on the hypercube takes more than $C$ values decays exponentially in $d$ (for some $C$ large enough). Using more sophisticated techniques, Galvin later showed in~\cite{Galvin2003hammingcube} that $C=5$ suffices and obtained the asymptotic distribution of the image size. The usage of entropy methods culminated in the work of Engbers and Galvin~\cite{engbers2012h2,engbers2012h1} who provided rigidity results in the hypercube setting for graph homomorphisms to any finite graph $H$. In particular, their results imply rigidity for uniformly chosen proper $q$-colorings (for any $q \ge 3$) and for the hard-core model in the hypercube setting (for any $m$). Entropy techniques, however, appear to be too weak to obtain results for the thermodynamic limit.

The first result concerning the thermodynamic limit of the aforementioned models was obtained in a paper by Galvin and Kahn~\cite{galvin2004phase} from 2004. In this paper, they showed rigidity for the hard-core model in high dimensions using advanced contour arguments (their bounds were later improved by Peled and Samotij~\cite{peled2014odd}). Using related ideas, two groups obtained similar results for the zero-temperature 3-state AF Potts model. In~\cite{peled2010high}, Peled showed rigidity for uniformly chosen HHFs on $\Z^d$ (or on $\Z^2 \times \{0,1\}^d$) for high $d$, and as a byproduct, obtained rigidity for proper $3$-colorings in high dimensions. Independently, Galvin and Randall~\cite{galvin2007torpid} investigated Glauber dynamics on proper $3$-colorings in high dimensions under periodic boundary conditions, showed torpid mixing of this Markov chain and obtained a coarse form of rigidity. Together with Kahn and Sorkin~\cite{galvin2012phase}, they later refined their methods to obtain results parallel to those of \cite{peled2010high} directly on proper $3$-colorings, although with somewhat weaker bounds. Feldheim and Peled~\cite{feldheim2013rigidity} were later able to use topological arguments to obtain a relation between HHFs on the torus $\Z^d/m\Z^d$ and proper $3$-colorings of the torus, and extend Peled's bounds to periodic boundary conditions.

Our work is the first to show rigidity for the AF Potts model on $\Z^d$ at positive temperature. The positive temperature case is challenging as one cannot exploit the linear structure of HHFs to tackle it. To overcome this, we must diverge from
the traditional view of a contour, replacing it with a more general partition of the domain which we call a \emph{breakup}.
For an example in a different context of an alternative to contours, see the work of Duminil-Copin, Peled, Samotij and Spinka on the loop $O(n)$ model~\cite{duminil2014exponential}.

While our work focuses on the AF Potts model on $\Z^d$, it may be worthwhile to mention several works in other settings.
Rough estimates for the number of zero-temperature configurations (and for more general objects) were obtained in \cite{galvin2004weighted,meyerovitch2014independence}.
In a recent work, Koteck{\'y}, Sokal and Swart~\cite{kotecky2014entropy} were able to prove a positive temperature phase transition for various asymmetric lattices in two dimensions. The asymmetric structure of the lattice reduces the number of low temperature Gibbs measures from six to three,
and allows the usage of more classical Peierls arguments. In other related works, Peled, Samotij and Yehudayoff showed rigidity for HHFs (and Lipschitz functions) on bipartite expander graphs~\cite{peled2012expanders} and on grounded trees~\cite{peled2013grounded}.

\subsection{Proof outline}
\label{sec:outline}

In this section, we provide a sketch of the proof of Theorem~1.1.

\smallskip
\noindent{\bf Phases and the breakup.}
Our goal is to show that in a random coloring with even-$0$ boundary conditions $(\Lambda,\tau)$, any given vertex $\rho$ is likely to follow the even-$0$ pattern.
Recall that the even-$0$ pattern is the pattern in which even vertices are colored $0$ and odd vertices are colored $1$ or $2$.
Consider the maximal connected set of vertices which follows the even-$0$ pattern and contains $\Lambda^c$.
We will show that it is likely that this set contains $\rho$.
If $\rho$ does not belong to this set, then it is separated from infinity by an interface between the even-$0$ pattern and other patterns. Thus, our goal is to show that such an interface is unlikely.

As we are concerned with interfaces between sets, there is no need to require the sets to strictly follow a pattern, but rather, it is enough to constrain the coloring near their boundaries. This leads us to the following definition of a \emph{phase}.
For a set of vertices $U$, we write $\intB U$ for the internal boundary of $U$, i.e., the set of vertices in $U$ which are adjacent to a vertex outside $U$, and say that $U$ is even (odd) if $\intB U$ consists solely of even (odd) vertices.
We say that $U$ belongs to the even-$0$ phase if $U$ is even and $\intB U$ is entirely colored $0$. Similarly, five other phases are defined according to the parity of the set (even/odd) and the constant color on the boundary.
Equipped with this definition, given a coloring, we may partition $\Z^d$ into these six phases.
We remark that this partition is not unique. For instance, we may have to decide whether a certain region consisting of even vertices colored $0$ and odd vertices colored $1$ is part of the even-0 phase or the odd-1 phase.
In order for such a partition to be of use to us, we require it to satisfy certain properties. We further elaborate on this point after reviewing the role that
phases play in our proof.

A special feature of \emph{proper} 3-colorings (i.e., samples in the zero-temperature model) is that when $\rho$ does not follow the even-0 pattern it is always possible to define the partition in such a way that only two regions exist; a region in the even-$0$ phase containing $\Lambda^c$, and a region in either the odd-$1$ phase or the odd-$2$ phase containing $\rho$.
This feature is related to the circular graph structure formed by the possible interfaces between phases; see Figure~\reffig{fig:phase-diagram-a}.
In fact, by lifting this circular structure, one may obtain a nested structure of interfaces.
In particular, an interface always separates two phases.
This special structure was exploited in \cite{peled2010high} explicitly (via a correspondence with integer height functions), and in \cite{galvin2012phase} implicitly, allowing the application of contour arguments to the problem.
In the positive temperature case, however, this structure breaks down completely.
As adjacent vertices may take the same color, additional possibilities of interfaces arise (at the cost of introducing improper edges); these are depicted in Figure~\reffig{fig:phase-diagram-b}.
The core difficulty created by this more complex interface structure is manifested by the fact that it is no longer always possible to construct a simple partition into two regions. Instead, a much more elaborate partition might be required, involving several or even all phases.

\begin{figure}
	\centering
	\begin{subfigure}[t]{.33\textwidth}
		\centering
		\includegraphics[scale=0.9]{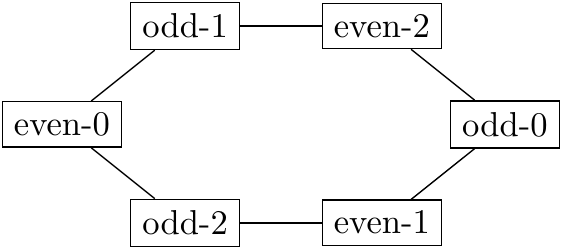}
		\caption{$\beta=\infty$.}
		\label{fig:phase-diagram-a}
	\end{subfigure}%
	\begin{subfigure}{10pt}
		\quad
	\end{subfigure}%
	\begin{subfigure}[t]{.33\textwidth}
		\centering
		\includegraphics[scale=0.9]{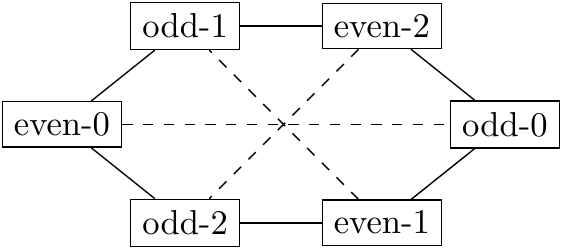}
		\caption{$\beta<\infty$.}
		\label{fig:phase-diagram-b}
	\end{subfigure}%
	\begin{subfigure}{5pt}
		\quad
	\end{subfigure}%
	\begin{subfigure}[t]{.33\textwidth}
		\centering
		\includegraphics[scale=0.9]{phase-diagram-c3\printable.pdf}
		\caption{The simplified phases.}
		\label{fig:phase-diagram-c}
	\end{subfigure}
	\caption{Possible interfaces between phases. Solid lines denote proper transitions, whereas dashed lines denote improper transitions. At zero-temperature ($\beta=\infty$), the structure of possible interfaces is circular. At positive temperature ($\beta<\infty$), the structure is more complicated due to the improper interfaces. By uniting the odd-1 phase with the even-2 phase and the odd-2 with the even-1, this structure is slightly simplified. The colors of the four phases in~(\textsc{c}) are used in figures throughout the paper.}
	\label{fig:phase-diagrams}
\end{figure}

Nevertheless, we are able to introduce a simplification.
We unite the odd-1 phase with the even-2 phase (and similarly, the odd-2 phase with the even-1 phase), treating them as a single phase. We call this the 1-phase (respectively, the 2-phase), referring to the even-$0$ phase as the 0-phase and to the odd-$0$ phase as the 3-phase.
Thus, our partition will consist of four phases rather than six, with neighboring relations depicted in Figure~\reffig{fig:phase-diagram-c}.
Notice that this structure preserves the property that an interface between two phases consists either solely of proper edges or solely of improper edges.
For instance, an edge between the 0-phase and the 1-phase is always a proper edge (as the endpoints are colored 0 and 1), while an edge between the 0-phase and 3-phase is always an improper edge (as the endpoints are both colored 0).
The distinction between these two types of interfaces plays a major role throughout the paper.
Working with these four phases comes, however, at the expense of losing the property that all phases have a well-defined parity (namely, only the 0-phase and 3-phase have such a parity).

For every coloring $f$, we define such a partition $K(f)=(K_0,K_1,K_2,K_3)$ into four phases, which we call the \emph{breakup}.
The breakup is defined through a prescribed process, illustrated in Figure~\reffig{fig:proof-illustration}.
This is done in a manner which guarantees that if $\rho$ does not follow the even-0 pattern then it does not belong to the 0-phase (i.e., $\rho \notin K_0$) and that the complement of each phase (i.e., $\Z^d \setminus K_i$) is connected.
Moreover, the breakup is defined so that the interfaces between the four different phases are connected in the sense that $\intB K_0 \cup \intB K_1 \cup \intB K_2 \cup \intB K_3$ is a connected set.
Conceptually, we view the breakup as the innermost obstruction for the persistence of the 0-phase towards $\rho$. The precise definition of the breakup and details of its properties are given in Section~\ref{sec:breakup+high-level-proof}.

\smallskip
\noindent{\bf An energy reducing and entropy gaining transformation.}
Equipped with the definition of a breakup, our goal is to show that the existence of a large breakup is unlikely. Together with a bound on the minimal size of a breakup, this will yield the desired result.

The first step towards this goal is to show that any particular (non-trivial) breakup $K$ is unlikely. To this end, we construct a one-to-many transformation ${\sf T}_K$, mapping each coloring $f$ having breakup $K$ to multiple colorings $g$ (different for every $f$) satisfying $\mu^\tau_{\Lambda,\beta}(g) \ge \mu^\tau_{\Lambda,\beta}(f)$. In the terminology of statistical mechanics, the increase of probability resulting from this transformation is seen as an \emph{energy reduction}, while the increase in the number of colorings (i.e., the number of images of each $f$) is seen as an \emph{entropy gain}.

Recall that the interfaces between the four phases may be either proper or improper. In fact, an edge on these interfaces (i.e., an edge whose endpoints belong to different phases) is proper in $f$ if and only if it has exactly one endpoint in $K_{12} := K_1 \cup K_2$ (see Figure~\reffig{fig:phase-diagram-c}).
We call the proper edges on these interfaces the \emph{regular boundary} of $K$ and the improper ones the \emph{singular boundary} of $K$.
We denote the number of edges of the regular boundary by $L$ and that of the singular boundary by $M$.

The energy reduction is related to the singular boundary. Specifically, it is the result of permuting the colors in the 2-phase and the 3-phase, so that all improper interfaces become proper (for vertices in the 2-phase, we permute colors 1 and 2, while for vertices in the 3-phase, we apply the cyclic permutation of colors $0 \to 1 \to 2 \to 0$).
Consequently, the energy of each coloring $f$ is reduced by $M$, yielding a probability gain of $e^{\beta M}$.
On the other hand, the entropy gain is related to the regular boundary and is obtained by shifting the coloring on $K_{12}$ in some cardinal direction (called down) and applying the cyclic permutation of colors $0 \to 2 \to 1 \to 0$. This shift and permutation ensure that every vertex in $\intB^{\down} K_{12}$, the upper internal boundary of $K_{12}$, is surrounded by the same color in all directions. Thus, these vertices are ``free'' to take either of the remaining two colors. Consequently, the entropy gain is exponential in the size of $\intB^{\down} K_{12}$, which is in turn proportional to $L/d$.
See Figure~\reffig{fig:transformation} for an illustration of this transformation and Lemma~\ref{lem:existence-of-transformation} for details.

The mere existence of ${\sf T}_K$ shows that the breakup $K$ is unlikely.
It is natural to try to show that the existence of any breakup at all is unlikely by applying a union bound over all possible choices of $K$.
However, this approach fails as the number of breakups is too large in comparison to our bound on the probability of a given breakup. Instead, we use a more sophisticated technique of \emph{approximations and flows} to show this. A version of this technique was used already in~\cite{galvin2004phase}.

\smallskip
\noindent{\bf Approximations.}
Roughly speaking, the large number of breakups is due to the number of potential perturbations in the interfaces between the phases. In particular, ``smoothing out'' these perturbations should significantly decrease the number of breakups. This idea is made precise through the notion of a four-approximation.

A four-approximation is an object maintaining information about the structure of a breakup $K$.
Away from the boundary of $K$, the exact partition into the four phases is known, but near the boundary, only partial information is maintained (see Figure~\ref{fig:approx} for an illustration of a four-approximation). The construction of four-approximations is somewhat involved; see Definition~\ref{def:four-approx} and Section~\ref{sec:approx} for details.
The number of four-approximations required to approximate all breakups will be small enough as to allow us to ultimately take a union bound over them (see Lemma~\ref{lem:family-of-FA}). We are therefore left with the goal of showing that it is unlikely that the breakup is approximated by a given four-approximation $A$. For this, it is convenient to classify breakups according to the entropy gain and energy reduction they yield under the above transformation ${\sf T}_K$.
Thus, we consider the set $\Col_{L,M}(A)$ of colorings $f$ whose breakup $K(f)$ has regular boundary of size $L$, singular boundary of size $M$ and is approximated by $A$.
To obtain a good bound on the probability of $\Col_{L,M}(A)$ under $\mu_{\Lambda,\beta}^\tau$, we extend ${\sf T}_K$ to $\Col_{L,M}(A)$ in a weighted manner.
We do so using the method of flows.

\smallskip
\noindent{\bf The flow.}
We give a probabilistic description of the method of flows in our context.
First, think of the transformation ${\sf T}_K$ as mapping every coloring $f$ having breakup $K$ to a probability measure, giving every possible image equal probability. We refer to this measure here as the $\sf T$-measure of $f$. Next, consider the following two-step procedure: pick a random coloring $f$ according to $\mu^\tau_{\Lambda,\beta}$, and then, if it belongs to $\Col_{L,M}(A)$, pick another random coloring $g$ according to the $\sf T$-measure of $f$ (otherwise set $g=\emptyset$). Evidently, if $\Pr(g=f_0) \le \epsilon \cdot \Pr(f=f_0)$ for every coloring $f_0$, then $\mu^\tau_{\Lambda,\beta}(\Col_{L,M}(A)) \le \epsilon$.

As there may exist colorings which are obtained by $f$ and by $g$ with comparable probabilities, this will not give a sufficiently small bound. In fact, one may check that the bound obtained in this manner is governed by the maximum number of $\sf T$-preimages in $\Col_{L,M}(A)$ among colorings $g$, which may even be as large as $|\Col_{L,M}(A)|$. Thus, this approach is essentially equivalent to a union bound.
The number of such preimages greatly varies among different $g$. In particular, even when the maximum number of preimages is large, the average number of preimages is much smaller.
Taking this into account, we may improve the flow by modifying the $\sf T$-measure of each coloring.

Recall that the multitude of images under ${\sf T}_K$ is the result of assigning one of two colors to the free vertices $\intB^{\down} K_{12}$. A key observation is that the level of difficulty of recovering the breakup (given the four-approximation) is not the same for all choices of colors for these vertices. Certain choices allow deterministic local recovery of the breakup, while others do not.
The nature of this variation is further explained in Section~\ref{sec:flow} (see Figure~\ref{fig:flow-pairs} for an illustration of a simple case of local recovery). By modifying the $\sf T$-measure so that it is biased towards choices which simplify the recovery of the breakup, we are able to improve the bound on $\mu^\tau_{\Lambda,\beta}(\Col_{L,M}(A))$, showing that it is exponentially small in $L$ and $M$. Finally, we will conclude the proof by taking a union bound over $L$, $M$ and $A$.

\subsection{Organization of the article}
The rest of the article is structured as follows.
In Section~\ref{sec:preliminaries}, definitions and preliminary results which will be needed throughout the paper are given.
The definition of the breakup and the proof of the main theorem are given in Section~\ref{sec:breakup+high-level-proof}. This section also introduces the two key lemmas; Lemma~\ref{lem:family-of-FA} regarding the existence of a small family of four-approximations, and Lemma~\ref{lem:prob-of-approx-enhanced} regarding the unlikeliness of a breakup with a given four-approximation. Section~\ref{sec:transformation} is devoted to the proof of Lemma~\ref{lem:prob-of-approx-enhanced}, while Section~\ref{sec:approx} is devoted to the proof of Lemma~\ref{lem:family-of-FA}.
In Section~\ref{sec:pattern-violations+gibbs}, we prove Theorem~\ref{thm:main2} and Theorem~\ref{thm:convergence-of-finite-volume-measure}.

\subsection{Acknowledgments}
We would like to thank Ron Peled for introducing us to the subject, for much encouragement and for many useful discussions.
We would also like thank Wojciech Samotij and Peleg Michaeli for helpful comments.

%
\section{Preliminaries}
\label{sec:preliminaries}
%

\subsection{Notation}\label{sec:notation}

Let $G=(V,E)$ be a graph.
For vertices $u,v \in V$ such that $\{u,v\} \in E$, we say that $u$ and $v$ are {\em adjacent} and write $u \sim v$.
For a subset $U \subset V$, denote by $N(U)$ the {\em neighbors} of $U$, i.e., vertices in $V$ adjacent to some vertex in $U$, and define for $t>0$,
\[ N_t(U) := \{ v \in V : |N(v) \cap U| \ge t \} .\]
In particular, $N_1(U)=N(U)$.
Denote the {\em external boundary} of $U$ by $\extB U := N(U) \setminus U$ and the {\em internal boundary} of $U$ by $\intB U := \extB U^c$.
Denote $U^+ := U \cup \extB U$ and $v^+ := \{v\}^+$.
The set of edges between two sets $U$ and $W$ is denoted by $\partial(U,W):=\{ \{u,w\} \in E : u \in U, w \in W \}$.
The {\em edge-boundary} of $U$ is denoted by $\partial U := \partial(U,U^c)$, and we write $\partial v := \partial \{v\}$ for shorthand.
We denote the graph-distance between $u$ and $v$ by $\text{dist}(u,v)$.
The \emph{diameter} of a finite connected set $U \subset V$, denoted by $\diam U$, is the maximum graph-distance between two vertices in $U$, where we follow the convention that the diameter of the empty set is zero.
For two non-empty sets $U,W \subset V$, we denote by $\dist(U,W)$ the minimum graph-distance between a vertex in $U$ and a vertex in $W$.

We consider the graph $\Z^d$ with nearest-neighbor adjacency, i.e., the edge set $E(\Z^d)$ is the set of $\{u,v\}$ such that $u$ and $v$ differ by one in exactly one coordinate.
A vertex of $\Z^d$ is called {\em even (odd)} if it is at even (odd) graph-distance from the origin.
We denote the set of even and odd vertices of $\Z^d$ by $\Even$ and $\Odd$ respectively.
For a unit vector $s \in \Z^d$, write $v^s := v+s$ for the translation of $v$ by $s$ and define $U^s := \{ v^s : v \in U \}$.
The internal boundary of $U$ \emph{in direction $s$} is then defined to be $\intB^s U := U \setminus U^s$.

{\bf Policy on constants:} In the rest of the paper, we employ the following policy on constants. We write $C,c,C',c'$ for positive absolute constants, whose values may change from line to line. Specifically, the values of $C,C'$ may increase and the values of $c,c'$ may decrease from line to line.

\subsection{Odd sets}

\begin{figure}
	\centering
	\begin{subfigure}[t]{.28\textwidth}
		\centering
		\includegraphics[scale=0.45]{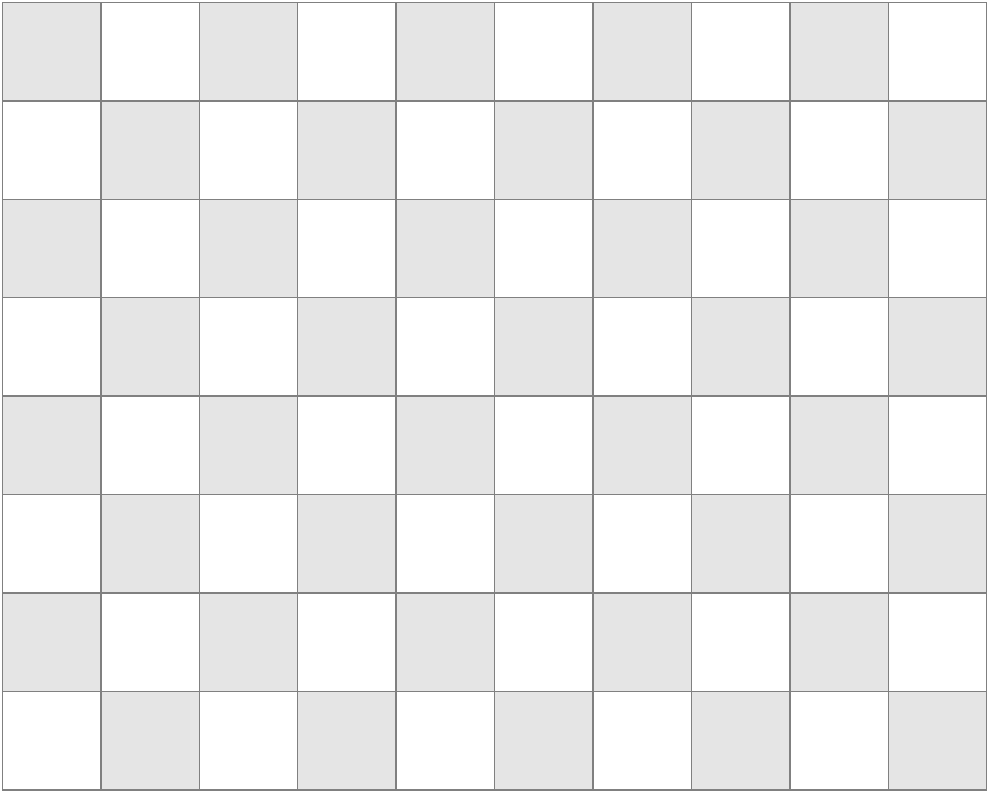}
		\caption{The even and odd bipartition classes of $\Z^d$.}
		\label{fig:even-odd-bipartition}
	\end{subfigure}%
	\begin{subfigure}{20pt}
		\quad
	\end{subfigure}%
	\begin{subfigure}[t]{.28\textwidth}
		\centering
		\includegraphics[scale=0.45]{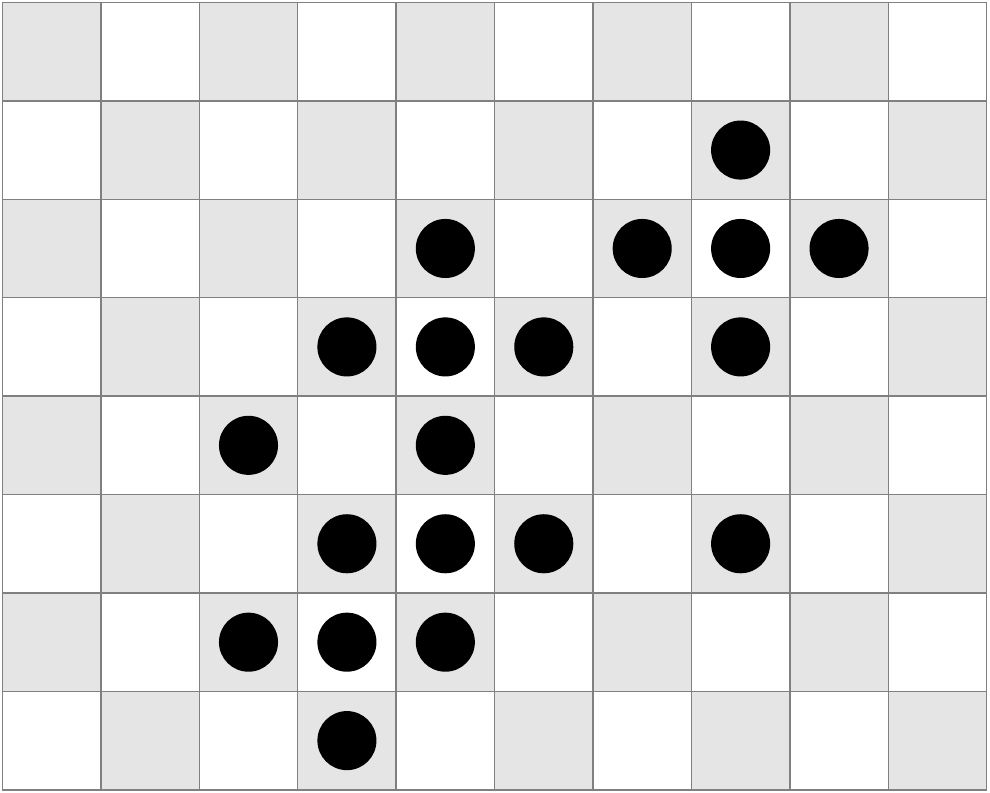}
		\caption{An odd set.}
		\label{fig:odd-set}
	\end{subfigure}%
	\begin{subfigure}{20pt}
		\quad
	\end{subfigure}%
	\begin{subfigure}[t]{.28\textwidth}
		\centering
		\includegraphics[scale=0.45]{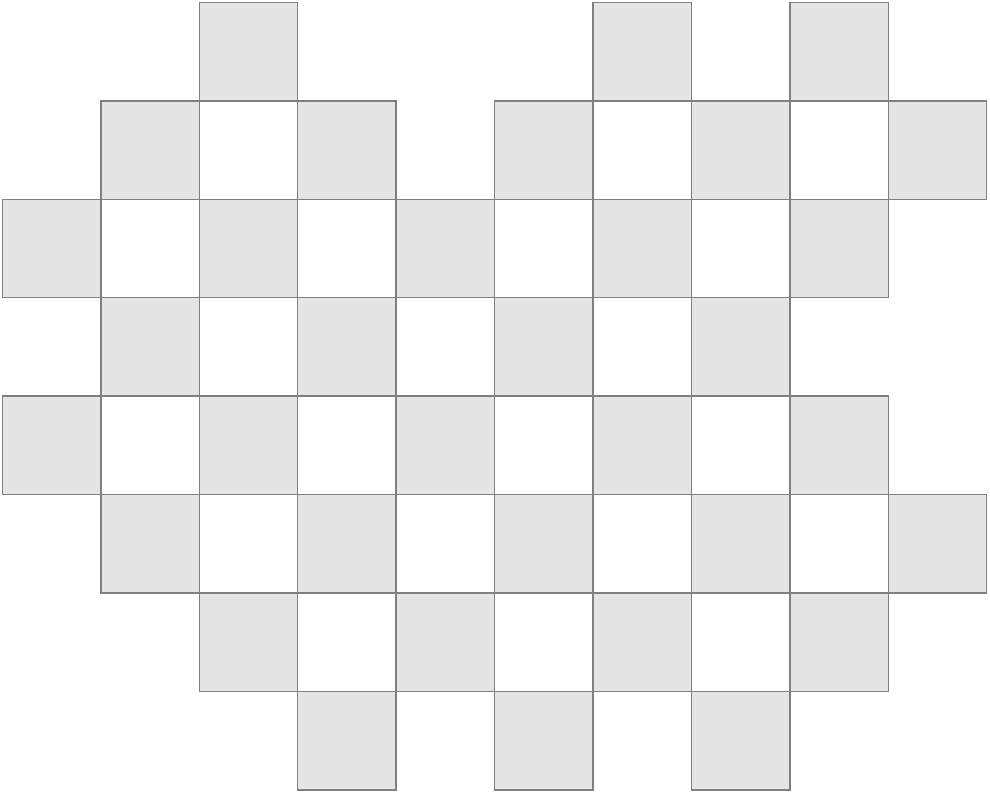}
		\caption{An odd domain.}
		\label{fig:odd-domain}
	\end{subfigure}
	\caption{Odd sets and domains. The vertices of $\Z^d$ are depicted as squares, with the even vertices in white and the odd vertices in gray. An odd set is a set whose internal boundary consists solely of odd vertices. A domain is a finite, non-empty, connected, co-connected set.}
	\label{fig:odd-sets-and-domains}
\end{figure}

We say that a set $U \subset \Z^d$ is {\em odd} if its internal boundary consists solely of odd vertices, i.e., if $\intB U \subset \Odd$. The complement of an odd set is called {\em even}.
Observe that a set $U \subset \Z^d$ is odd if and only if $(\Even \cap U)^+ \subset U$.
See Figure~\reffig{fig:odd-set}.

A nice property of odd sets, which appeared already in~\cite{borgs1999torpid}, is that the size of their boundary is the same in every direction.

\begin{lemma}\label{lem:odd-set-boundary-size}
	Let $U \subset \Z^d$ be finite and odd.
	Then, for any unit vector $s \in \Z^d$, we have
	\[ |\intB^s U| = |\Odd \cap U| - |\Even \cap U| = \tfrac{|\partial U|}{2d} .\]
\end{lemma}
\begin{proof}
We have
\begin{align*}
|\Even \cap U| &= |\Odd \cap U^s| = |\Odd \cap U^s \cap U| = |\Odd \cap U| - |\Odd \cap U \setminus U^s| \\&= |\Odd \cap U| - |U \setminus U^s| = |\Odd \cap U| - |\intB^s U| .
\end{align*}
Thus, we have established the first equality.
The second equality now follows from the first, since $|\partial U| = \sum_{s'} |\intB^{s'} U| = 2d \cdot |\intB^s U|$.
\end{proof}

Odd sets arise naturally in the context of proper $3$-colorings of $\Z^d$, and thus, played an important role in previous works \cite{peled2010high,galvin2012phase}.
In Section~\ref{sec:approx}, we introduce \emph{semi-odd pairs}, a variant of this notion.

\subsection{Isoperimetry}

The following isoperimetric inequality follows from a corresponding inequality of Bollob{\'a}s and Leader~\cite{bollobas1991edge}. Nonetheless, we provide a short proof.

\begin{lemma}\label{lem:isoperimetry}
	Let $A \subset \Z^d$ be finite. Then $|\partial A| \ge 2d \cdot |A|^{1-1/d}$.
\end{lemma}
\begin{proof}
	The proof uses the Brunn-Minkowski inequality in $\R^d$ (see, e.g.,~\cite{schneider2013convex}).
	Denote by $\lambda$ the Lebesgue measure on $\R^d$.
	For $r>0$, write $B_r := [-r/2,r/2]^d$ and note that $\lambda(B_r)=r^d$.
	Define $S := A+B_1$, where $X+Y:=\{x+y : x\in X, y\in Y\}$ for sets $X,Y \subset \R^d$, and note that $\lambda(S)=|A|$.
	The classical Brunn-Minkowski inequality applied to $S$ gives
	\[ \lambda(S + B_\epsilon) \ge \lambda(B_s + B_\epsilon) = \lambda(B_{s+\epsilon}) = (s + \epsilon)^d , \quad \epsilon>0 ,\]
	where $s := |A|^{1/d}$ so that $\lambda(B_s)=\lambda(S)$.
	Therefore,
	\[ \lim_{\epsilon \downarrow 0} \frac{\lambda(S+B_\epsilon)-\lambda(S)}{\epsilon/2} \ge \lim_{\epsilon \downarrow 0} \frac{(s+\epsilon)^d-s^d}{\epsilon/2} = 2d \cdot s^{d-1} = 2d \cdot |A|^{1-1/d} .\]
	Thus, observing that the left-hand side is equal to the surface area of $S$, which in turn equals $|\partial A|$, the lemma follows.
\end{proof}

\begin{cor}\label{cor:isoperimetry}
	Let $A \subset \Z^d$ be finite.
	\begin{enumerate}[\qquad(a)]
		\item\label{it:isoperimetry-small} If $|A| \le d$ then $|\partial A| \ge d |A|$.
		\item\label{it:isoperimetry-large} If $|A| \ge d$ then $|\partial A| \ge d^2$.
		\item\label{it:isoperimetry-regular-odd} If $A$ is odd and contains an even vertex then $|\partial A| \ge d^2$.
	\end{enumerate}
\end{cor}
\begin{proof}
	Parts~\eqref{it:isoperimetry-small} and~\eqref{it:isoperimetry-large} follow from Lemma~\ref{lem:isoperimetry} since $2 \ge e^{1/e} \ge x^{1/x}$ for all $x>0$.
	Part~\eqref{it:isoperimetry-regular-odd} follows from part~\eqref{it:isoperimetry-large}, since, if $v \in A$ is an even vertex, then $v^+ \subset A$ so that $|A| \ge |v^+| = 2d+1$.
\end{proof}

Corollary~\ref{cor:isoperimetry}\ref{it:isoperimetry-regular-odd} may be extended to the following.

\begin{lemma}\label{lem:boundary-size-via-diameter}
	Let $A \subset \Z^d$ be non-empty, finite, odd and connected. Then $|\partial A| \ge (d-1)^2 \diam A$.
\end{lemma}
\begin{proof}
	We may assume that $d \ge 2$.
	Let $u,v \in A$ be such that $k := \dist(u,v) = \diam A$ and let $p$ be a shortest-path in $A$ between $u$ and $v$.
	Denote by $\pi_i \colon \Z^d \to \Z^{d-1}$ the projection $\pi_i(x) := (x_1,\dots,x_{i-1},x_{i+1},\dots,x_d)$. Since $k = \dist(u,v)$, there exists a coordinate $1 \le i \le d$ such that $r := |\pi_i(p)|-1 \ge \dist(\pi_i(u),\pi_i(v)) \ge k (1-1/d)$. Let $B$ denote the set of vertices $w$ in $\pi_i(p)$ such that $w^+ \subset \pi_i(A)$ and note that $|B| \ge r/2$. Since no vertex has more than two neighbors in $B$, it follows that $|\pi_i(A)| \ge (2d-2)|B|/2 = (d-1)|B|$ (see Lemma~\ref{lem:sizes} below). Finally, since $|\intB^{i} A| \ge |\pi_i(A)|$, Lemma~\ref{lem:odd-set-boundary-size} implies that $|\partial A| \ge d(d-1)r \ge (d-1)^2 k$.
\end{proof}

\subsection{Co-connected sets}
\label{sec:co-connected-sets}

In this section, we fix an arbitrary connected graph $G=(V,E)$.
A set $U \subset V$ is called {\em co-connected} if its complement $V \setminus U$ is connected.
For a set $U \subset V$ and a vertex $v \in V$, we define the {\em co-connected closure} of $U$ with respect to $v$ to be the complement of the connected component of $V \setminus U$ containing $v$, where it is understood that this results in $V$ when $v \in U$.
We say that a set $U' \subset V$ is a co-connected closure of a set $U \subset V$ if it is its co-connected closure with respect to some $v \in V$.
Evidently, every co-connected closure of a set $U$ is co-connected and contains $U$.
An illustration of this notion is given in Figure~\reffig{fig:co-connect}.

\begin{figure}
    \centering
    \begin{subfigure}[t]{.28\textwidth}
        \centering
        \includegraphics[scale=0.45]{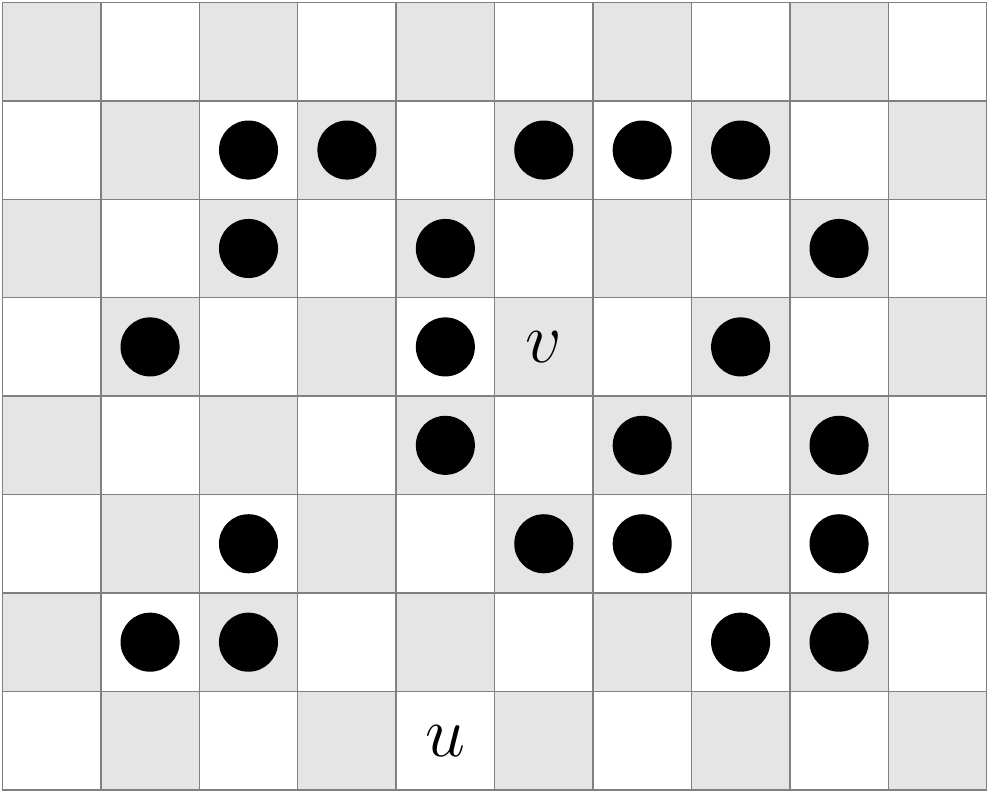}
        \caption{A non-co-connected set $U$.}
        \label{fig:co-connect-1}
    \end{subfigure}%
    \begin{subfigure}{20pt}
        \quad
    \end{subfigure}%
    \begin{subfigure}[t]{.28\textwidth}
        \centering
        \includegraphics[scale=0.45]{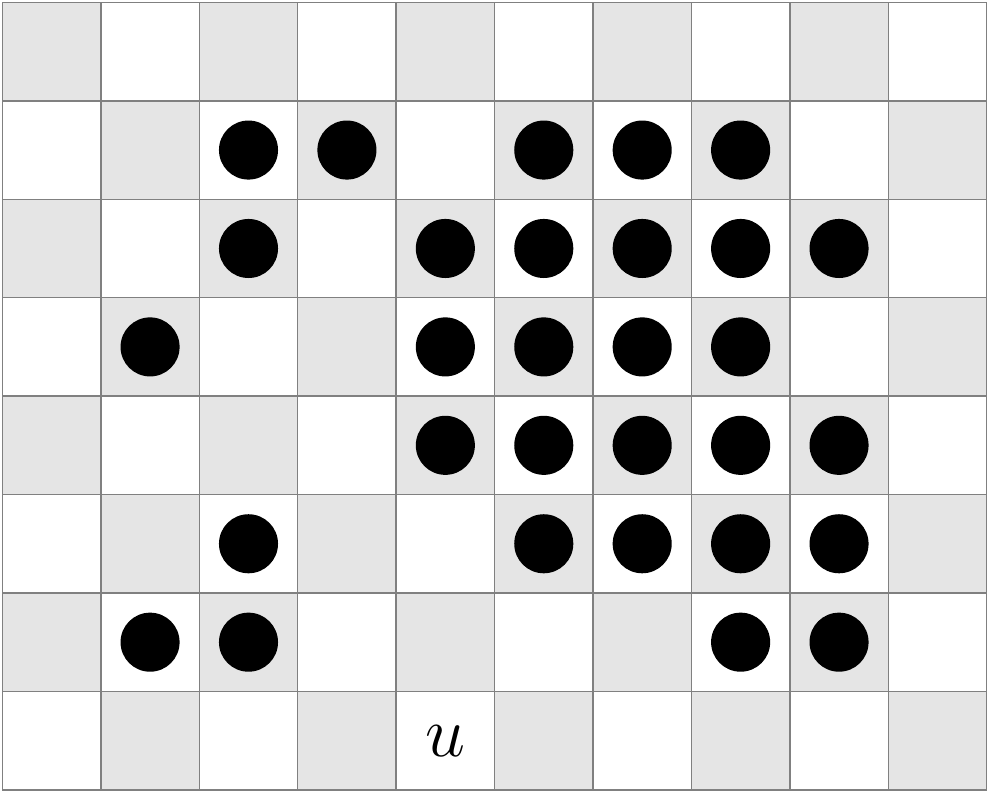}
        \caption{The co-connected closure of $U$ with respect to $u$.}
        \label{fig:co-connect-2}
    \end{subfigure}%
    \begin{subfigure}{20pt}
        \quad
    \end{subfigure}%
    \begin{subfigure}[t]{.28\textwidth}
        \centering
        \includegraphics[scale=0.45]{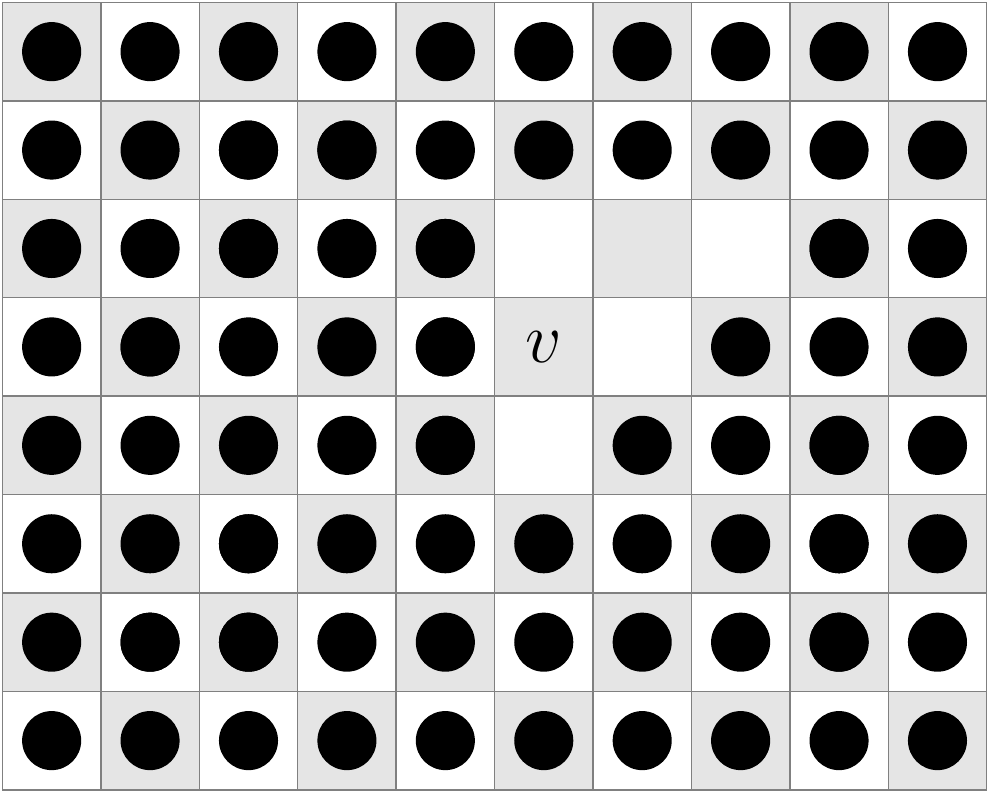}
        \caption{The co-connected closure of $U$ with respect to $v$.}
        \label{fig:co-connect-3}
    \end{subfigure}
    \caption{Co-connected closures.}
    \label{fig:co-connect}
\end{figure}

\begin{lemma}\label{lem:co-connect-properties}
Let $A,B \subset V$ be disjoint and let $A'$ be a co-connected closure of $A$. Then
\begin{enumerate}[\qquad(a)]
\item \label{it:co-connect-reduces-boundary} $\intB A' \subset \intB A$, $\extB A' \subset \extB A$ and $\partial A' \subset \partial A$.
\item \label{it:co-connect-reduces-boundary2} $\intB (B \setminus A') \subset \intB B$.
\item \label{it:co-connected-minus-set-is-co-connected} If $B$ is co-connected then $B \setminus A'$ is also co-connected.
\item \label{it:co-connect-kills-components} 
If $B$ is connected then either $B \subset A'$ or $B \cap A' = \emptyset$.
\end{enumerate}
\end{lemma}
\begin{proof}
We begin by proving part \eqref{it:co-connect-kills-components}.
Note that $(A')^c$ is either empty or is a connected component of $A^c$. Thus, since $B \subset A^c$ is connected, it is either contained in $(A')^c$ or disjoint from it.

Towards showing part~\eqref{it:co-connect-reduces-boundary}, let $v \in A'$ and let $u \notin A'$ be adjacent to $v$. We must show that $v \in A$ and $u \notin A$.
It is immediate that $u \notin A$ since $A \subset A'$.
Furthermore, $v \in A$, since otherwise, applying part~\eqref{it:co-connect-kills-components} to the connected set $\{v,u\}$ would result in a contradiction.

For part~\eqref{it:co-connected-minus-set-is-co-connected}, we must show that $A' \cup B^c$ is connected. Since $B^c$ is connected by assumption, it suffices to show that every connected component of $A'$ intersects $B^c$.
Indeed, by part~\eqref{it:co-connect-reduces-boundary}, every connected component of $A'$ contains an element of $A \subset B^c$.

Finally, we prove part \eqref{it:co-connect-reduces-boundary2}.
Let $B_1,\dots,B_k$ be the connected components of $B$.
By part~\eqref{it:co-connect-kills-components}, for every $1 \le i \le k$, either $B_i \subset A'$ or $B_i \cap A' = \emptyset$.
Thus, for every $1 \le i \le k$, either $\intB (B_i \setminus A') = \emptyset$ or $\intB (B_i \setminus A') = \intB B_i$.
Hence, $\intB (B \setminus A') = \cup_i \intB (B_i \setminus A') \subset \cup_i \intB B_i = \intB B$.
\end{proof}

The following lemma, based on ideas of Tim{\'a}r \cite{timar2013boundary}, establishes the connectivity of the boundary of subsets of $\Z^d$ which are both connected and co-connected.
For $U \subset \Z^d$, denote $\partial_{\ins \out} U := \intB U \cup \extB U$.

\begin{lemma}[{\cite[Proposition~3.1]{feldheim2013rigidity}}] \label{lem:int+ext-boundary-is-connected}
Let $A \subset \Z^d$ be connected and co-connected. Then $\intextB A$ is connected.
\end{lemma}

\begin{cor}\label{cor:boundary-of-partition-is-connected}
Let $A_1,\dots,A_k \subset \Z^d$ be a partition of $\Z^d$ into co-connected sets. Then $\cup_{i=1}^k \intB A_i$ is connected.
\end{cor}
\begin{proof}
Let $B_1,\dots,B_n$ be the collection of connected components of $A_1,\dots,A_k$, and observe that $\cup_i \intB A_i = \cup_i \intextB A_i = \cup_j \intextB B_j$.
Thus, it suffices to show that $\intextB B_j \cup \intextB B_{j'}$ is connected whenever $\dist(B_j,B_{j'}) \le 1$.
Indeed, since every connected component of a co-connected set is co-connected, this follows immediately from Lemma~\ref{lem:int+ext-boundary-is-connected}.
\end{proof}

\subsection{Graph properties}

In this section, we gather some elementary combinatorial facts about graphs.
Here, we fix an arbitrary graph $G=(V,E)$ of maximum degree $\Delta$.

\begin{lemma}\label{lem:sizes}
	Let $A\subset V$ be finite and let $t>0$. Then
	\[ |N_t(A)| \le \frac{\Delta}{t} \cdot |A| .\]
\end{lemma}
\begin{proof}
	This follows from a simple double counting argument.
	\[ t |N_t(A)|
		\le \sum_{v \in N_t(A)} |N(v) \cap A|
		= \sum_{u \in A} \sum_{v \in N_t(A)} \1_{N(u)}(v)
		= \sum_{u \in A} |N(u) \cap N_t(A)|
		\le \Delta |A| . \qedhere \]
\end{proof}

\begin{lemma}\label{lem:sizes2}
Let $W \subset U \subset V$ and let $r>t>0$. Assume that $W \subset N_r(V)$ and $W \cap N_t(U) = \emptyset$.
Then $|\partial(W,U)| \le \frac{t}{r-t} |\partial(W,U^c)|$.
\end{lemma}
\begin{proof}
On the one hand, $|\partial(W,U)| \le \sum_{w \in W} |N(w) \cap U| \le t |W|$, and on the other hand, $|\partial(W,U^c)| = \sum_{w \in W} |N(w) \cap U^c| = \sum_{w \in W} (|N(w)| - |N(w) \cap U|) \ge (r-t)|W|$.
\end{proof}

The next lemma follows from a classical result by Lov{\'a}sz~\cite[Corollary~2]{lovasz1975ratio}.

\begin{lemma}\label{lem:existence-of-covering2}
Let $A \subset V$ be finite and let $t \ge 1$. Then there exists a set $B \subset A$ of size $|B| \le \frac{1+\log \Delta}{t} |A|$ such that $N_t(A) \subset N(B)$.
\end{lemma}
%

The following standard lemma gives a bound on the number of connected subsets of a graph.
\begin{lemma}[{\cite[Chapter~45]{Bol06}}]\label{lem:number-of-connected-graphs}
The number of connected subsets of $V$ of size $k+1$ which contain a given vertex is at most $(e(\Delta-1))^k$.
\end{lemma}

For a positive integer $r$, we denote by $G^{\otimes r}$ the graph on $V$ in which two vertices are adjacent if their distance in $G$ is at most $r$.
The next simple lemma was first introduced by Sapozhenko~\cite{sapozhenko1987onthen}.

\begin{lemma}[{\cite[Lemma 2.1]{sapozhenko1987onthen}}]\label{lem:r-connected-sets}
Let $S,T \subset V$ and let $a,b$ be positive integers.
Assume that $S$ is connected in $G^{\otimes a}$, $\dist(s,T) \le b$ for all $s \in S$ and $\dist(S,t) \le b$ for all $t \in T$. Then $T$ is connected in $G^{\otimes (a+2b)}$.
\end{lemma}

%
\section{The breakup and the proof of the main theorem}
\label{sec:breakup+high-level-proof}
%

In this section, we prove Theorem~\ref{thm:main}, formalizing the ideas presented in Section~\ref{sec:outline}.
In the course of the proof, we state two key lemmas, Lemma~\ref{lem:family-of-FA} and Lemma~\ref{lem:prob-of-approx-enhanced}, whose proofs comprise most of the remaining paper.

Let $(\Lambda,\tau)$ be even-0 boundary conditions, i.e., $\Lambda$ is an odd domain and $\tau^{-1}(0)=\Even$.
Denote $\Col := \cC_\Lambda^\tau$ and let $f \in \Col$. Our first aim is to identify the different phases within $f$.
To this end, we label each vertex $v \in \Z^d$ with one of four labels $\kappa(v) \in \{0,1,2,3\}$.
\begin{equation}\label{eq:def-K-coloring}
\kappa(v) := \begin{cases}
	0 		&\text{if $f(v)=0$ and $v$ is even}\\
	3 		&\text{if $f(v)=0$ and $v$ is odd}\\
	f(v) 	&\text{if $f(v) \neq 0$ and $v$ is odd} \\
	3-f(v)	&\text{if $f(v) \neq 0$ and $v$ is even}
\end{cases} .
\end{equation}
Recall the definition of the co-connected closure of a set with respect to a vertex from Section~\ref{sec:co-connected-sets}.
The co-connected closure of a set $U \subset \Lambda$ with respect to infinity is its co-connected closure with respect to any vertex in $\Lambda^c$ (this does not depend on the vertex).
Let $\rho \in \Z^d$ be an arbitrary vertex, which is henceforth fixed throughout the proof.
Let $K'_0$ be the co-connected closure of $\kappa^{-1}(0)$ with respect to $\rho$.
Let $K'_3$ be the co-connected closure of $\kappa^{-1}(3) \setminus K'_0$ with respect to infinity, let $K'_2$ be the co-connected closure of $\kappa^{-1}(2) \setminus (K'_0 \cup K'_3)$ with respect to infinity and let $K'_1$ be the co-connected closure of $\kappa^{-1}(1) \setminus (K'_0 \cup K'_3 \cup K'_2)$ with respect to infinity.
Finally, define the \emph{breakup of $f$ around $\rho$} to be $K(f) = K(f,\rho)$, where $K(f) := (K_0(f), K_1(f), K_2(f), K_3(f))$ and $K_1(f) := K'_1$, $K_2(f) := K'_2 \setminus K'_1$, $K_3(f) := K'_3 \setminus (K'_1 \cup K'_2)$ and $K_0(f) := K'_0 \setminus (K'_1 \cup K'_2 \cup K'_3)$.
See Figure~\reffig{fig:proof-illustration} for an illustration of these definitions.
Observe that, by definition,
\begin{equation}\label{eq:K-0-is-co-finite}
\Z^d \setminus K_0(f) \subset \Lambda
\end{equation}
and
\begin{equation}\label{eq:K-0-does-not-contain-rho}
\rho \in K_0(f) \quad\iff\quad K_0(f) = \Z^d \quad\iff\quad  \rho \text{ is even and } f(\rho)=0 .
\end{equation}
We say that the breakup is \emph{trivial} if $|\Z^d \setminus K_0(f)| \le 1$ and $K_3(f)=\emptyset$.
By~\eqref{eq:K-0-does-not-contain-rho}, we have
\begin{equation}\label{eq:origin-non-trivial-four-section}
\begin{aligned}
 &\text{For even $\rho$:}\quad
 	&&K(f) \text{ is trivial}
 \quad\iff\quad
 	K_0(f) = \Z^d
 \quad\iff\quad
 f(\rho)=0
 \\
 &\text{For odd $\rho$:}\quad
 	&&K(f) \text{ is trivial}
 \quad\iff\quad
 	\substack{K_0(f)=\Z^d\setminus\{\rho\},\\ K_1(f) \cup K_2(f) = \{\rho\}}
 \quad\iff\quad
 \substack{
 f(\rho) \neq 0,\\ f|_{N(\rho)}=0}
\end{aligned}
\end{equation}
Thus, we can show parts~(a) and~(b) of Theorem~\ref{thm:main} by providing suitable bounds on the probability that the breakup of $f$ is non-trivial.

%
%

\begin{figure}
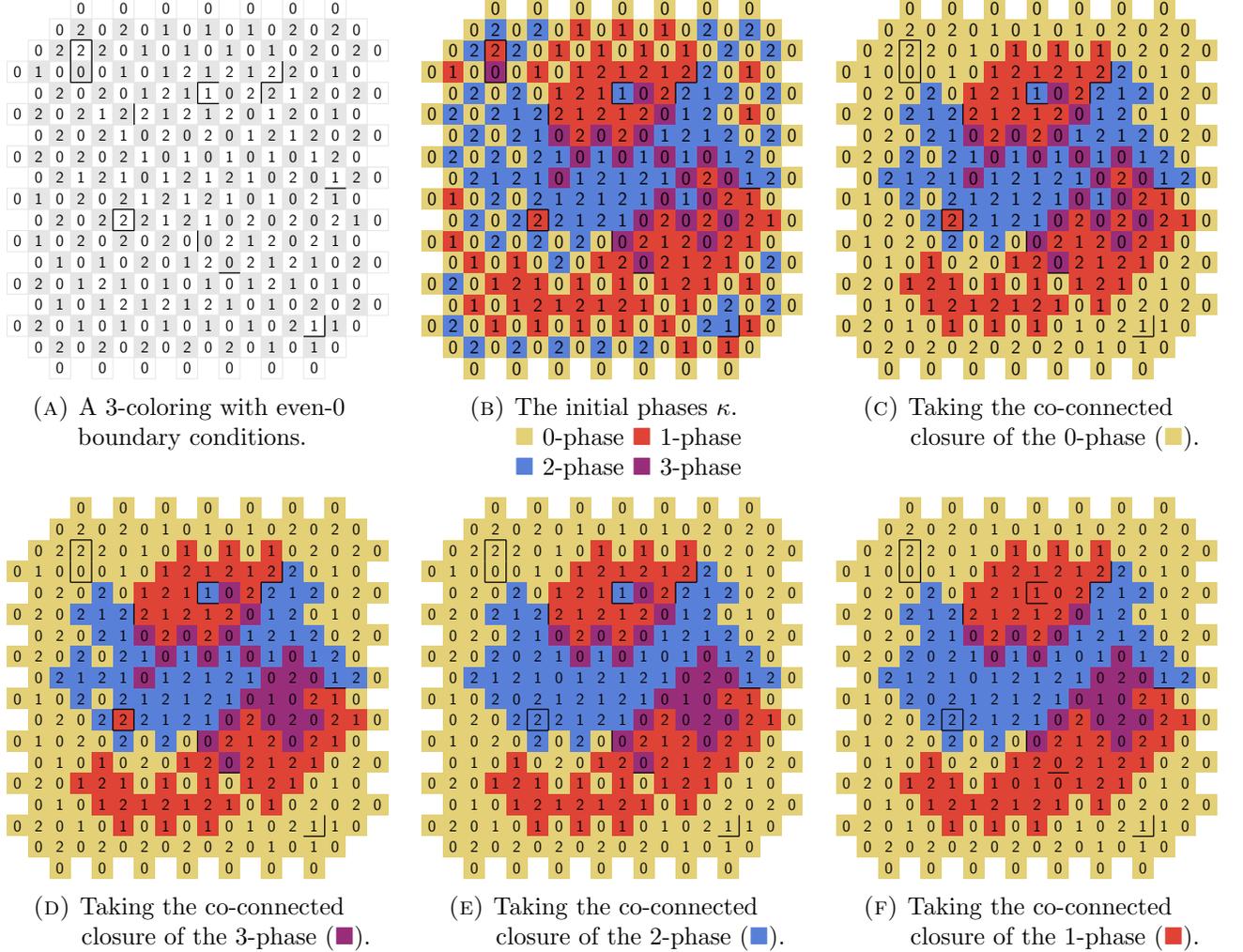

	\captionsetup{width=0.89\textwidth}
	\captionsetup[subfigure]{justification=centering}
	\begin{subfigure}[t]{.3\textwidth}
		\includegraphics[scale=0.3]{coloring-sample3a\printable.pdf}
		\caption{A $3$-coloring with even-$0$ \phantom{.}boundary conditions.}
		\label{fig:proof-illustration-1}
	\end{subfigure}%
	\begin{subfigure}{20pt}
		\quad
	\end{subfigure}%
	\begin{subfigure}[t]{.3\textwidth}
		\includegraphics[scale=0.3]{coloring-sample3b\printable.pdf}
		\caption{The initial phases $\kappa$.\\
		\qquad{\color{fsc0\grayscale}$\blacksquare$} 0-phase
		{\color{fsc1\grayscale}$\blacksquare$} 1-phase \\
		\qquad{\color{fsc2\grayscale}$\blacksquare$} 2-phase
		{\color{fsc3\grayscale}$\blacksquare$} 3-phase}
		\label{fig:proof-illustration-2}
	\end{subfigure}%
	\begin{subfigure}{20pt}
		\quad
	\end{subfigure}%
	\begin{subfigure}[t]{.3\textwidth}
		\includegraphics[scale=0.3]{coloring-sample3c\printable.pdf}
		\caption{Taking the co-connected\\ \quad\qquad closure of the $0$-phase ({\color{fsc0\grayscale}$\blacksquare$}).}
		\label{fig:proof-illustration-3}
	\end{subfigure}
	\medbreak
	\begin{subfigure}[t]{.3\textwidth}
		\includegraphics[scale=0.3]{coloring-sample3d\printable.pdf}
		\caption{Taking the co-connected\\ \quad\qquad closure of the $3$-phase ({\color{fsc3\grayscale}$\blacksquare$}).}
		\label{fig:proof-illustration-4}
	\end{subfigure}%
	\begin{subfigure}{20pt}
		\quad
	\end{subfigure}%
	\begin{subfigure}[t]{.3\textwidth}
		\includegraphics[scale=0.3]{coloring-sample3e\printable.pdf}
		\caption{Taking the co-connected\\ \quad\qquad closure of the $2$-phase ({\color{fsc2\grayscale}$\blacksquare$}).}
		\label{fig:proof-illustration-5}
	\end{subfigure}%
	\begin{subfigure}{20pt}
		\quad
	\end{subfigure}%
	\begin{subfigure}[t]{.3\textwidth}
		\includegraphics[scale=0.3]{coloring-sample3f\printable.pdf}
		\caption{Taking the co-connected\\ \quad\qquad closure of the $1$-phase ({\color{fsc1\grayscale}$\blacksquare$}).}
		\label{fig:proof-illustration-6}
	\end{subfigure}
	\caption{ The breakup of a coloring, determined by a 5-step procedure. The numbers represent the coloring (identical in all figures) and bold lines represent improper edges. Initially, in~\textsc{(b)}, each vertex $v$ is assigned a phase $\kappa(v)$ according to its parity and color, as in~\eqref{eq:def-K-coloring}. In every subsequent step, the $i$-phase is expanded to its co-connected closure $K'_i$. This is done first for $i=0$ and then for $i=3$, $i=2$ and $i=1$ (in that order).
		Thus, in~(\textsc{\subref{fig:proof-illustration-3}}) for example, the vertices of $K'_0$ are assigned the $0$-phase, while the remaining vertices retain their previously assigned phases from~(\textsc{\subref{fig:proof-illustration-2}}). The phases assigned in the last step~(\textsc{\subref{fig:proof-illustration-6}}) constitute the breakup. The breakup has the property that the color of every vertex on the boundary of some phase is determined by the parity and phase of the vertex (see Definition~\ref{def:adapted-fs}a). In fact, the same holds in each step.}
	\label{fig:proof-illustration}
\end{figure}

The next definition captures several geometric properties of the breakup which are relevant for the construction of four-approximations.

\begin{definition}[four-section]
\label{def:four-section}
A tuple $K=(K_0,K_1,K_2,K_3)$ of disjoint subsets of $\Z^d$ is called a {\em four-section} if the following holds:
\begin{enumerate}[\qquad(a)]
\item \label{it:def-FS-1} $K_0 \cup K_1 \cup K_2 \cup K_3 = \Z^d$.
\item \label{it:def-FS-3} $K_0$ is even and $K_3$ is odd.
\item \label{it:def-FS-5} If $v \in \Z^d$ satisfies $N(v) \subset K_i$ for some $i \in \{0,1,2,3\}$ then $v \in K_i \cup K_{3-i}$.
\end{enumerate}
\end{definition}

Note that this definition is symmetric with respect to $K_1$ and $K_2$, and that the roles of $K_0$ and $K_3$ are symmetric up to a change of parity.
For a four-section $K$, we define its (edge-)boundary $\partial K$ and its vertex-boundary $\intB K$ to be
\begin{align*}
\partial K &:= \partial K_0 \cup \partial K_1 \cup \partial K_2 \cup \partial K_3 ,\\
\intB K &:= \intB K_0 \cup \intB K_1 \cup \intB K_2 \cup \intB K_3 .
\end{align*}
An analogous definition of $\extB K$ is redundant as it would coincide with $\intB K$.
We say that a four-section $K$ is \emph{connected} if $\intB K$ is connected.

As will be explained in Section~\ref{sec:shift}, our methods do not rely on the precise definition of the breakup of $f$, but rather only on its geometric properties (given in the definition above) and on a handful of its properties with respect to $f$. The latter are captured in the following definition.

\begin{definition}[adapted four-section]
\label{def:adapted-fs}
A four-section $K$ is said to be {\em adapted} to a coloring $f$ if the following holds:
\begin{enumerate}[\qquad(a)]
\item[(a)] $f(v)=0 \quad\text{for all } v \in \intB K_0 \cup \intB K_3$,
\item[] $f(v)=1 \quad\text{for all } v \in (\Odd \cap \intB K_1) \cup (\Even \cap \intB K_2)$,
\item[] $f(v)=2 \quad\text{for all } v \in (\Even \cap \intB K_1) \cup (\Odd \cap \intB K_2)$.
\item[(b)] For every $v \in \intB K_1 \cup \intB K_2$ there exists $u \sim v$ such that $f(u) \neq 0$.
\end{enumerate}
\end{definition}

\begin{lemma}\label{lem:K-is-a-four-section}
Let $f \in \Col$ and assume that the breakup $K(f)$ is non-trivial. Then $K(f)$ is a connected four-section adapted to $f$.
\end{lemma}

\begin{proof}
Let $f \in \Col$ have a non-trivial breakup and denote $K := K(f)$ and $K_i := K_i(f)$.

We begin by showing that $K$ is a four-section.
The fact that $K_0,K_1,K_2,K_3$ are pairwise disjoint follows directly from the definition of the breakup.
To see \refFSpartition{}, observe that $K_0 \cup K_1 \cup K_2 \cup K_3 = K'_0 \cup K'_1 \cup K'_2 \cup K'_3$, and that $K'_1 \cup K'_2 \supset \kappa^{-1}(\{1,2\}) \setminus (K'_0 \cup K'_3)$, $K'_1 \cup K'_2 \cup K'_3 \supset \kappa^{-1}(\{1,2,3\}) \setminus K'_0$ and $K'_0 \cup K'_1 \cup K'_2 \cup K'_3 \supset \kappa^{-1}(\{0,1,2,3\})$.

In light of~\eqref{eq:def-K-coloring}, \refFSodd{} will follow if we show that
\begin{equation}\label{eq:boundary-of-K}
\intB K_i \subset \intB K'_i \subset \kappa^{-1}(i).
\end{equation}
The fact that $\intB K_i \subset \intB K'_i$ follows from
Lemma~\ref{lem:co-connect-properties}\ref{it:co-connect-reduces-boundary2} (applied iteratively) and the definitions of $K_i$ and $K'_i$. Lemma~\ref{lem:co-connect-properties}\ref{it:co-connect-reduces-boundary} implies that $\intB K'_i \subset \kappa^{-1}(i)$.

Recalling that $K$ is non-trivial, it is easy to check that \refFSisolated{} holds when $|\Z^d \setminus K_0| \le 1$.
Thus, to establish~\refFSisolated{}, it suffices to show that, for any $i \in \{0,1,2,3\}$, $\Z^d \setminus K_i$ does not contain isolated vertices when $|\Z^d \setminus K_0|>1$.
Since $|\Z^d \setminus K_0|>1$ and $K_1$, $K_2$, $K_3$ are finite by~\eqref{eq:K-0-is-co-finite}, we have that $\Z^d \setminus K_i$ is not a singleton for $i \in \{0,1,2,3\}$.
Hence, it suffices to show that
\begin{equation}\label{eq:K-is-co-connected}
K_i \text{ is co-connected}, \quad i \in \{0,1,2,3\} .
\end{equation}
The set $K_1$ is co-connected by definition.
The fact that $K_2$ is co-connected follows directly from  Lemma~\ref{lem:co-connect-properties}\ref{it:co-connected-minus-set-is-co-connected}.
Similarly, to see that $K_3$ and $K_0$ are co-connected, note that $K_3 = (K'_3 \setminus K'_2) \setminus K'_1$ and $K_0 = ((K'_0 \setminus K'_3) \setminus K'_2) \setminus K'_1$, and apply Lemma~\ref{lem:co-connect-properties}\ref{it:co-connected-minus-set-is-co-connected} iteratively.

The fact that $K$ is a connected four-section now follows from~\eqref{eq:K-is-co-connected} and Corollary~\ref{cor:boundary-of-partition-is-connected}.

It remains to show that $K$ is adapted to $f$.
\refAFSvalues{} follows directly from~\eqref{eq:boundary-of-K} and~\eqref{eq:def-K-coloring}.
For \refAFSneighbor{}, let $v \in \Z^d$ satisfy $N(v) \subset f^{-1}(0)$.
By the first part, it suffices to show that $v^+ := N(v) \cup \{v\}$ is contained in either $K_1 \cup K_2$ or $K_0 \cup K_3$ when $f(v) \neq 0$.
To this end, we first show that $v^+$ is contained in $K'_0 \cup K'_3$.
Indeed, if $v$ is odd then $N(v) \subset \kappa^{-1}(0)$, and so $v^+ \subset K'_0$, by the definition of $K'_0$ and since $K$ is non-trivial and $f(v) \neq 0$.
Otherwise, $v$ is even so that $N(v) \subset \kappa^{-1}(3)$. Thus, if $N(v) \cap K'_0 = \emptyset$ then $v^+ \subset K'_3$.
On the other hand, if $N(v) \cap K'_0 \neq \emptyset$ then~\eqref{eq:def-K-coloring} and~\eqref{eq:boundary-of-K} imply that $v \in K'_0$, and hence, $v^+ \subset K'_0 \cup K'_3$.
Recalling the definitions of $K'_1$ and $K'_2$, Lemma~\ref{lem:co-connect-properties}\ref{it:co-connect-kills-components} now implies that, for $i \in \{1,2\}$, if $v^+ \cap K'_i \neq \emptyset$ then $v^+ \subset K'_i$.
In particular, if $v^+ \cap (K'_1 \cup K'_2) \neq \emptyset$ then $v^+ \subset K'_1 \cup K'_2$.
Thus, either $v^+ \subset K'_1 \cup K'_2 = K_1 \cup K_2$ or $v^+ \subset (K'_0 \cup K'_3) \setminus (K'_1 \cup K'_2) = K_0 \cup K_3$.
\end{proof}

For a four-section $K$, we write $K_{12} := K_1 \cup K_2$ and $K_{123} := K_1 \cup K_2 \cup K_3$ for short.
Recall that $K(f)=K(f,\rho)$ is the breakup of $f$ around the fixed vertex $\rho$.
Denote
\[ \FS_\rho := \big\{ K\text{ four-section} ~:~ \rho \in K_{123},~K_{123}\text{ is finite},~ K\text{ is connected} \big\} ,\]
and note that, by Lemma~\ref{lem:K-is-a-four-section}, \eqref{eq:K-0-does-not-contain-rho} and~\eqref{eq:K-0-is-co-finite}, we have
\begin{equation}\label{eq:breakup-in-FS-rho}
K(f)\text{ is non-trivial} \quad\implies\quad K(f) \in \FS_{\rho} .
\end{equation}
Thus, our goal is to show that the event $\{ K(f) \in \FS_{\rho} \}$ is unlikely.
Observe that if $K$ is adapted to $f$ then
\begin{equation}\label{eq:good-bad-boundaries-of-K}
\begin{aligned}
\partial K_{12} &= \partial(K_1 \cup K_2, K_0 \cup K_3) \subset \{ \{u,v\} \in E(\Z^d) : f(u) \neq f(v) \} ,\\
\partial K \setminus \partial K_{12} &= \partial(K_1,K_2) \cup \partial(K_0,K_3) \subset \{ \{u,v\} \in E(\Z^d) : f(u)=f(v) \} .
\end{aligned}
\end{equation}
We call $\partial K_{12}$ and $\partial  K \setminus \partial K_{12}$ the {\em regular} and {\em singular boundary} of $K$, respectively, and denote the {\em singularities} in $K$ by
\[ K^{\anomaly} := \{ v \in \Z^d ~:~ \partial v \cap \partial K \setminus \partial K_{12} \neq \emptyset \} \subset \intB K .\]
For integers $L \ge 0$ and $M \geq 0$, denote
\[ \FS_{L,M} := \big\{ K\text{ four-section} ~:~ |\partial K_{12}|=L,~|\partial K \setminus \partial K_{12}| = M,~ 0<|K_{123}|<\infty \big\} .\]
Let $K$ be a four-section such that $K_{123}$ is finite and non-empty.
Observe that \refFSodd{} implies that $K_{123}$ is odd.
Hence, by Corollary~\ref{cor:isoperimetry}\ref{it:isoperimetry-regular-odd}, if $\Even \cap K_{123} \neq \emptyset$, then $|\partial K| \ge |\partial K_{123}| \ge d^2$.
On the other hand, if $K_{123} \subset \Odd$, then by \refFSisolated{}, $|\partial K \setminus \partial K_{12}| = |\partial K_{123}| = 2d |K_{123}| \ge 2d$.
Therefore,
\begin{equation}\label{eq:L+M-is-large}
\FS_{L,M} \text{ is non-empty} \quad\implies\quad L+M \ge d^2 \quad\text{or}\quad M \ge 2d .
\end{equation}

Our goal now is to bound the probability that $K(f)$ belongs to $\FS_{L,M,\rho} := \FS_{L,M} \cap \FS_{\rho}$.
To do so, we require the notion of a four-approximation.
As was explained in the outline, a four-approximation is a means to record information about a four-section.
For organizational purposes and as the definition is quite involved, it is provided in the next section; see Definition~\ref{def:four-approx}.
We divide the core of our proof into two lemmas.
The first shows that a small set of four-approximations suffices to approximate every four-section in $\FS_{L,M,\rho}$.

For integers $L,M \ge 0$ and a four-approximation $A$, denote
\[ \FS_{L,M}(A) := \big\{ K \in \FS_{L,M} ~:~ A \text{ is a four-approximation of } K \big\} .\]

\begin{lemma}\label{lem:family-of-FA}
	For any integers $L,M \ge 0$, there exists a family $\cA$ of four-approximations of size
	\[ |\cA| \le \exp\left( C L d^{-3/2} \log^{3/2} d + CM \log d \right) \]
	such that
	$\FS_{L,M,\rho} \subset \cup_{A \in \cA} \FS_{L,M}(A)$.
\end{lemma}

The second lemma bounds the probability that the breakup is approximated by a given four-approximation.
In fact, our proof yields a bound not only for the breakup $K(f)$ defined above, but for more general objects, namely, adapted four-sections.
For a coloring $f \in \Col$, denote by $\FS(f)$ the set of four-sections which are adapted to $f$. In particular, Lemma~\ref{lem:K-is-a-four-section} implies that
\begin{equation}\label{eq:breakup-is-adapted}
K(f)\text{ is non-trivial} \quad\implies\quad K(f) \in \FS(f).
\end{equation}

\begin{lemma}\label{lem:prob-of-approx-enhanced}
Let $\beta>0$ and let $f \sim \mu^\tau_{\Lambda,\beta}$.
For any integers $L,M \ge 0$ and any four-approximation $A$, we have
\[ \Pr\big(\FS(f) \cap \FS_{L,M}(A) \neq \emptyset \big) \le \exp(-c L/d - \beta M + CM) .\]
Moreover, for any adjacent vertices $u,v \in \Lambda$, we have
\[ \Pr\big(\FS(f) \cap \FS_{L,0}(A) \neq \emptyset \textnormal{ and } f(u)=f(v) \big) \le e^{-c L/d} \cdot \Pr\big( f(u+e)=f(v+e)\textnormal{ for some }e \in 0^+ \big) .\]
\end{lemma}

Note that when $\beta \ge C \log d$, the coefficient of $M$ in Lemma~\ref{lem:prob-of-approx-enhanced} outweighs the corresponding term in Lemma~\ref{lem:family-of-FA}. Similarly, when the dimension $d$ is large enough, the same holds for the coefficient of $L$. We now use this to prove the main theorem, showing that when $d$ and $\beta$ sufficiently large, a non-trivial breakup is unlikely.

\begin{proof}[Proof of Theorem~\ref{thm:main}]
	Let $(\Lambda,\tau)$ be even-0 boundary conditions and let $f \sim \mu^\tau_{\Lambda,\beta}$.
	We assume throughout the proof that $d$ is large enough for our arguments to hold.
	Let $L \ge 0$ and $M \ge 0$ be integers, and let $\cA$ be a family of four-approximations obtained by applying Lemma~\ref{lem:family-of-FA}.
	By the union bound and Lemma~\ref{lem:prob-of-approx-enhanced},
	\begin{align}
	\Pr(\FS(f) \cap \FS_{L,M,\rho} \neq \emptyset)
	 &
	 \le \sum_{A \in \cA} \Pr(\FS(f) \cap \FS_{L,M}(A) \neq \emptyset) \notag \\
	 &\le \exp\left( C L d^{-3/2} \log^{3/2} d + CM \log d \right) \cdot \exp(- c L/d - \beta M + CM) \notag \\
	 &\le \exp(- c' L/d - \beta' M) , \label{eq:prob-of-breakup-with-LM}
	\end{align}
	where $\beta' := \beta - C' \log d$.
	Therefore, denoting
	$\FS_{\ge L, \ge M,\rho} := \bigcup_{\ell \ge L, m \ge M} \FS_{\ell,m,\rho}$,
	and assuming that $\beta' \ge 1$ (using $\beta \ge C \log d$), we have
	\begin{equation}\label{eq:prob-of-breakup-with-LM-or-more}
	\Pr(\FS(f) \cap \FS_{\ge L, \ge M,\rho} \neq \emptyset) \le \sum_{\ell=L}^\infty \sum_{m=M}^\infty e^{-c' \ell/d - \beta' m}
	\le Cd \cdot e^{-c'L/d - \beta'M} .
	\end{equation}
	Thus, by~\eqref{eq:breakup-in-FS-rho}, \eqref{eq:L+M-is-large} and~\eqref{eq:breakup-is-adapted}, we obtain
	\begin{align}
	\Pr(K(f)\text{ is non-trivial})
	 &= \Pr(K(f) \in \FS_{\ge 0, \ge 0,\rho}) \notag \\
	 &\le \Pr(K(f) \in \FS_{\ge d^2-2d, \ge 0,\rho}) + \Pr(K(f) \in \FS_{\ge 0, \ge 2d, \rho}) \notag \\
	 &\le Cd \cdot e^{-c'd} + Cd \cdot e^{-2\beta' d} \le e^{-cd} . \label{eq:prob-of-non-trivial-breakup}
	\end{align}
	Let $u \in \Lambda$ be even and let $v \in \Lambda$ be odd.
	Then, by~\eqref{eq:origin-non-trivial-four-section},
	\[ \Pr(f(u) \neq 0) \le \Pr(K(f,u)\text{ is non-trivial}) \le e^{-cd} . \]
	The same bound easily follows for the event $\{f(v)=0\}$ as well. To obtain the stronger bound as in the statement of the theorem, we show that, on the event $\{f(v)=0\}$, the breakup must be larger than in general.
	Namely, denoting $K := K(f,v)$, $L := |\partial K_{12}|$ and $M := |\partial K \setminus \partial K_{12}|$, we improve the general bound in~\eqref{eq:L+M-is-large} by showing that either $L \ge d^3/2$ or $M \ge 3d/2$.
	To this end, assume that $f(v)=0$ and denote $S := N(v) \cap K_{123}$ and $s := |S|$.
	Then, since $K_{123}$ is odd by \refFSodd{},
	\[ |K_{123}| \ge |S^+| \ge |S^+ \setminus \{v\}| \ge 2ds - \binom{s}{2} \ge s(2d-s/2) .\]
	Thus, Lemma~\ref{lem:isoperimetry} implies that if $s \ge d/2$ then
	\[ L+M \ge |\partial K_{123}| \ge 2d \cdot (7d^2/8)^{1-1/d} \ge d^3 .\]
	On the other hand, if $s \le d/2$ then
	\[ M \ge |\partial(K_0,K_3)| \ge |N(v) \setminus S| = 2d - s \ge 3d/2 ,\]
	where for the first inequality we used~\eqref{eq:good-bad-boundaries-of-K}, and for the second inequality we used the fact that $v \in K_3$ by~\eqref{eq:K-0-does-not-contain-rho} and \refAFSvalues{}.
	Therefore, assuming that $\beta' \ge 2\beta/3$ (using again $\beta \ge C \log d$), we have by~\eqref{eq:origin-non-trivial-four-section},
	\begin{align*}
	\Pr(f(v) = 0)
	&= \Pr(f(v) = 0 \text{ and } K(f,v)\text{ is non-trivial}) \\
	&\le \Pr(K(f,v) \in \FS_{\ge d^3/2, \ge 0,v}) + \Pr(K(f,v) \in \FS_{\ge 0, \ge 3d/2,v}) \\
	&\le Cd \cdot e^{-c'd^2} + Cd \cdot e^{-3\beta'd/2} \le e^{-cd^2} + e^{-\beta d} .
	\end{align*}

It remains to show part~(c) of the theorem.
Assume that $u$ and $v$ are adjacent and define the events $E:=\{f(u)=f(v)\}$ and $F := \{K^{\anomaly}(f,v) = \emptyset\}$.
We bound separately the probabilities of the events $E \setminus F$ and $E \cap F$.
For the first event, using~\eqref{eq:origin-non-trivial-four-section} and repeating a computation similar to the one in~\eqref{eq:prob-of-non-trivial-breakup} (in particular, using~\eqref{eq:breakup-in-FS-rho}, \eqref{eq:L+M-is-large}, \eqref{eq:breakup-is-adapted} and~\eqref{eq:prob-of-breakup-with-LM-or-more}), we obtain
\begin{align*}
\Pr(E \setminus F)
 &\le \Pr(K(f,v) \text{ is non-trivial and } K^{\anomaly}(f,v) \neq \emptyset) \\
 &= \Pr(K(f,v) \in \FS_{\ge 0, \ge 1,v}) \\
 &\le \Pr(K(f,v) \in \FS_{\ge d^2-2d, \ge 1,v}) + \Pr(K(f) \in \FS_{\ge 0, \ge 2d, v}) \\
 &\le Cd \cdot e^{-c'd - (\beta- C \log d)} + Cd \cdot e^{-2d(\beta- C \log d)} \le e^{-cd-\beta} .
\end{align*}
For the second event, denoting $E' := \cup_{e \in N(0)} \{ f(u+e)=f(v+e) \}$, we have by Lemma~\ref{lem:prob-of-approx-enhanced} that
\[ \Pr\left(E \cap F \cap \{\FS(f) \cap \FS_{L,0}(A) \neq \emptyset \} \right) \le e^{-c L/d} \cdot \Pr(E \cup E') .\]
Repeating computations similar to those in~\eqref{eq:prob-of-breakup-with-LM} and~\eqref{eq:prob-of-non-trivial-breakup}, we obtain
\[ \Pr(E \cap F) \le e^{-cd} \cdot \Pr(E \cup E') .\]
Thus,
\[ \Pr(E) = \Pr(E \setminus F) + \Pr(E \cap F) \le e^{-cd-\beta} + e^{-cd} \cdot (\Pr(E) + \Pr(E')) .\]
Therefore, we have $\alpha := \sup_{u \sim v} \Pr(f(u)=f(v)) \le e^{-cd-\beta} + e^{-cd} \cdot (2d+1)\alpha$.
Hence,
\[ \Pr(f(u)=f(v)) \le \alpha \le \frac{e^{-cd-\beta}}{1 - (2d+1)e^{-cd}} \le e^{-c'd-\beta} . \qedhere \]
\end{proof}

%
\section{Transformation and flow}
\label{sec:transformation}
%

This section is dedicated to the proof of Lemma~\ref{lem:prob-of-approx-enhanced}. That is, our goal is to prove an upper bound on the probability that there exists an adapted four-section which has regular boundary of size $L$, has singular boundary of size $M$ and is approximated by a particular four-approximation $A$.
Given a four-section $K$, we first define a transformation ${\sf T}_K$ mapping every coloring to which $K$ is adapted to many distinct colorings with fewer singularities.
This is made precise in Lemma~\ref{lem:existence-of-transformation}.
The existence of such a transformation implies that any given breakup is unlikely.
Next, we provide the formal definition of a four-approximation.
In order to bound the probability that the breakup is approximated by a given four-approximation $A$, we then introduce the technique of \emph{flows}, outlined in the introduction and stated precisely in Lemma~\ref{lem:flow}.
The particular flow used in our setting is subsequently defined, and the bounds obtained in Lemma~\ref{lem:flow-bounds} then yield Lemma~\ref{lem:prob-of-approx-enhanced}.
The proofs of Lemma~\ref{lem:existence-of-transformation} and Lemma~\ref{lem:flow-bounds} are given in Section~\ref{sec:shift} and Section~\ref{sec:flow}, respectively.
Throughout this section, we view the set of colors $\{0,1,2\}$ taken by functions in $\Col_{\Z^d}$ as $\Z/3\Z$. Thus, all arithmetic operations on values of $3$-colorings are taken modulo~$3$.

As before, we fix even-0 boundary conditions $(\Lambda,\tau)$ and denote $\Col := \cC_\Lambda^\tau$.
Let $\down$ be a unit vector in $\Z^d$ and denote $\up := -\down$.
We now define the Flip, Shift and Mod transformations.
For a four-section $K$, denote
\[ \Col_K := \{ f \in \Col ~:~ K \in \FS(f) \} .\]
For a coloring $f \in \Col_K$, define
\begin{align*} \text{Flip}_K(f)(v) := \begin{cases}
 -f(v) 		&\text{if } v \in K_2 \\
 f(v)+1		&\text{if } v \in K_3 \\
 f(v)		&\text{otherwise}
\end{cases} ,\qquad
\text{Shift}_K(f)(v) := \begin{cases}
 f(v) 				&\text{if } v \in K_{03} \\
 f(v^{\up}) - 1		&\text{if } v \in K_{12} \cap K_{12}^{\down} \\
 1				&\text{if } v \in K_{12} \cap K_0^{\down} \\
 0				&\text{if } v \in K_{12} \cap K_3^{\down} \\
\end{cases} .
\end{align*}
For $h\colon \intB^{\down} K_{12} \to \{0,1\}$, define
\[ \text{Mod}_{K,h}(f)(v) := \begin{cases}
 h(v)+1	&\text{if } v \in K_{12} \cap K_0^{\down} \\
 -h(v)		&\text{if } v \in K_{12} \cap K_3^{\down} \\
f(v)		&\text{otherwise}
\end{cases} .\]
Finally, we define the flip+shift+mod transformation
\[ {\sf T}_K \colon \Col_K \times \{0,1\}^{\intB^{\down} K_{12}} \to \Col \]
by
\[ {\sf T}_K (f,h) := \big(\text{Mod}_{K,h} \circ \text{Shift}_K \circ \text{Flip}_K \big)(f) .\]
Observe that ${\sf T}_K (f,h)$ coincides with $f$ on $K_0$, so that ${\sf T}_K$ is well-defined, by~\eqref{eq:K-0-is-co-finite}.
There is some freedom in the choice of the transformation ${\sf T}_K$.
The following lemma summarizes the properties of ${\sf T}_K$ which we require for our arguments (see Figure~\reffig{fig:transformation} for an illustration of this transformation).

\begin{figure}
	\captionsetup[subfigure]{justification=centering}
    \begin{subfigure}[t]{.3\textwidth}
        \includegraphics[scale=0.3]{coloring-sample3f\printable.pdf}
        \caption{A coloring $f$ and \\ \quad its breakup $K$.}
        \label{fig:transformation-breakup}
    \end{subfigure}%
    \begin{subfigure}{20pt}
        \quad
    \end{subfigure}%
    \begin{subfigure}[t]{.3\textwidth}
        \includegraphics[scale=0.3]{coloring-flip\printable.pdf}
        \caption{The coloring \phantom{\hspace{25pt}} $f' := \text{Flip}_K(f)$.}
        \label{fig:transformation-flip}
    \end{subfigure}%
    \begin{subfigure}{20pt}
    	\quad
    \end{subfigure}%
    \begin{subfigure}[t]{.35\textwidth}
        \includegraphics[scale=0.3]{coloring-shift\printable.pdf}
        \caption{The coloring \phantom{\hspace{35pt}} $f'':=\text{Shift}_K(f')$.}
        \label{fig:transformation-shift}
    \end{subfigure}
    \caption{The transformation ${\sf T}_K$ in stages. The backgrounds represent the phases of the initial breakup $K$ as in Figure~\reffig{fig:proof-illustration} ({\color{fsc0\grayscale}$\blacksquare$} 0-phase, {\color{fsc1\grayscale}$\blacksquare$} 1-phase, {\color{fsc2\grayscale}$\blacksquare$} 2-phase, {\color{fsc3\grayscale}$\blacksquare$} 3-phase). The final coloring ${\sf T}_K(f,h)$ is obtained from $f''$ by modifying the value of $f''$ on a subset of $\intB^{\down} K_{12}$. Observe that the singularities of $f$ on the boundary of the breakup (i.e., $\partial K \setminus \partial K_{12}$) are not singularities of $f'$. Also notice that every vertex in $\intB^{\down} K_{12}$ is surrounded by a single color in $f''$.}
    \label{fig:transformation}
\end{figure}

For a coloring $f \in \Col$, denote
\[ E^{\anomaly}(f) := \big\{ \{u,v\} \in E(\Z^d) : f(u)=f(v) \big\} .\]
For a four-section $K$, let $B(K)$ be an independent set of $\intB^{\down} K_{12}$ (i.e., a set containing no two adjacent vertices) of maximal size and denote
\[ \cH_K := \big\{ h\colon \intB^{\down} K_{12} \to \{0,1\} ~:~  h(v) = 0 \text{ for all } v \in \intB^{\down} K_{12} \setminus B(K) \big\} .\]

\begin{lemma}\label{lem:existence-of-transformation}
Let $L,M \ge 0$ be integers and let $K \in \FS_{L,M}$.
For any $f \in \Col_K$ and any $h \in \cH_K$, denoting $g := {\sf T}_K(f,h)$, we have
\begin{enumerate}[\qquad(a)]
 \item \label{it:transformation-injective} ${\sf T}_K$ is an injective map to $\Col$.
 \item \label{it:transformation-probability}
 $|E^{\anomaly}(g)| = |E^{\anomaly}(f)| - M$.
 \item \label{it:transformation-easy-recover} $K_{12} \cap g^{-1}(2)^{\up} \subset f^{-1}(0)$.
  \item \label{it:transformation-anomalies} $E(K_{03}) \cap E^{\anomaly}(f) \setminus \partial K \subset E^{\anomaly}(g)$ and $E(K_{12}) \cap E^{\anomaly}(f) \setminus \partial K \subset E^{\anomaly}(g)^{\up}$.
\end{enumerate}
\end{lemma}

The proof is deferred to Section~\ref{sec:shift}.
The mere existence of a transformation ${\sf T}_K$ satisfying \eqref{it:transformation-injective} and \eqref{it:transformation-probability} above shows immediately that $\mu^\tau_{\Lambda,\beta}(\Col_K) \le e^{-\beta M} |\cH_K|^{-1}$.
Since $|\cH_K|=2^{|B(K)|}$ and $|B(K)| \ge |\intB^{\down} K_{12}|/2$, this bound is at most $e^{-\beta M} 2^{-|\intB^{\down} K_{12}|/2}$.
The following lemma provides a lower bound on $|\intB^{\down} K_{12}|$ in terms of $L$ and $M$.

\begin{lemma}\label{cl:size-of-boundary}
Let $L,M \ge 0$ be integers, let $K \in \FS_{L,M}$ and let $\down$ be any unit vector in $\Z^d$. Then
\[ |\intB^{\down} K_{12}| \ge \frac{L}{2d} - M .\]
\end{lemma}
\begin{proof}
By Lemma~\ref{lem:odd-set-boundary-size}, since $K_0$ is even and $K_3$ is odd, we have
\[ |\intB^{\up} K_0|+|\intB^{\up} K_3| = \tfrac{1}{2d}|\partial K_0|+ \tfrac{1}{2d}|\partial K_3| \ge \tfrac{1}{2d} |\partial K_0 \cup \partial K_3| \ge \tfrac{1}{2d}|\partial K_{12}| = \tfrac{L}{2d} .\]
Hence, writing $|\intB^{\up} K_j| = |K_j \setminus K_j^{\up}| = |K_j \cap K_{12}^{\up}| + |K_j \cap K_{3-j}^{\up}|$ for $j \in \{0,3\}$,
we obtain
\begin{align*}
|\intB^{\down} K_{12}| &= |K_{12} \setminus K_{12}^{\down}| = |K_{12} \cap K_{03}^{\down}| = |K_{12}^{\up} \cap K_{03}|
 = |K_{12}^{\up} \cap K_3| + |K_{12}^{\up} \cap K_0| \\&= |\intB^{\up} K_0|+|\intB^{\up} K_3| - |K_0 \cap K_3^{\up}| - |K_3 \cap K_0^{\up}| \ge \tfrac{L}{2d} - M . \qedhere
\end{align*}
\end{proof}

Hence, we obtain that $\mu^\tau_{\Lambda,\beta}(\Col_K)$ is at most $e^{-(\beta-1/2) M} 2^{-L/4d}$, which is exponentially small in $L$ and $M$ for, say, $\beta \ge 1$.
However, as mentioned in the introduction, naively applying the union bound over all $K \in \FS_{L,M}$ does not yield a meaningful bound. Instead, we shall simultaneously handle the subset of four-sections $\FS_{L,M}(A)$, those approximated by a given four-approximation $A$. We now give the precise definition of this notion, followed by an explanation.

\begin{definition}[four-approximation]
	\label{def:four-approx}
	Let $A=(A_i,A_j,\DilemmaOdd ij,\DilemmaEven ij)_{i \in \{1,2\},j \in \{0,3\}}$ be a twelve-tuple of subsets of $\Z^d$ such that for every $i \in \{1,2\}$ and $j\in\{0,3\}$ the following holds:
	\begin{enumerate}[\qquad(a)]
		\item\label{it:FA1} $\{ A_0,A_1,A_2,A_3, \DilemmaEven 10, \DilemmaEven 20, \DilemmaEven 13, \DilemmaEven 23, \DilemmaOdd 10 \cup \DilemmaOdd 20, \DilemmaOdd 13 \cup \DilemmaOdd 23 \}$ is a partition of $\Z^d$.
		\item\label{it:FA2} 
		$\DilemmaOdd i0, \DilemmaEven i3 \subset \Odd$ and $\DilemmaEven i0, \DilemmaOdd i3 \subset \Even$.
		\item\label{it:FA13} The subgraph induced by $\DilemmaOdd ij \cup \DilemmaEven ij$ has maximum degree at most $\sqrt{d}$.
	\end{enumerate}
	We say that $A$ is a {\em four-approximation} of a four-section $K$ if for every $i \in \{1,2\}$ and $j\in\{0,3\}$,
	\begin{enumerate}[\qquad(a)]
		\setcounter{enumi}{3}
		\item\label{it:FA9} $A_i \subset K_i$, $A_j \subset K_j$, $\DilemmaOdd ij \subset K_{12j}$, $\DilemmaOdd ij \setminus (\DilemmaOdd 1j \cap \DilemmaOdd 2j) \subset K_{ij}$ and $\DilemmaEven ij \subset K_{ij}$.
		\item\label{it:FA12} 
		$\DilemmaEven ij \cap N(\DilemmaOdd ij \setminus K_i) = \DilemmaEven ij \cap K_j$ and $\DilemmaOdd ij \cap N(\DilemmaEven ij \setminus K_j) = \DilemmaOdd ij \cap K_i$.
		\item\label{it:FA11} \qquad~~$N(\DilemmaOdd ij \cap K_i)  \subset \DilemmaEven ij \cup K_j$ and \qquad~~$N(\DilemmaEven ij \cap K_j) \subset \DilemmaOdd ij \cup K_i$.
		
	\end{enumerate}
\end{definition}

\begin{figure}
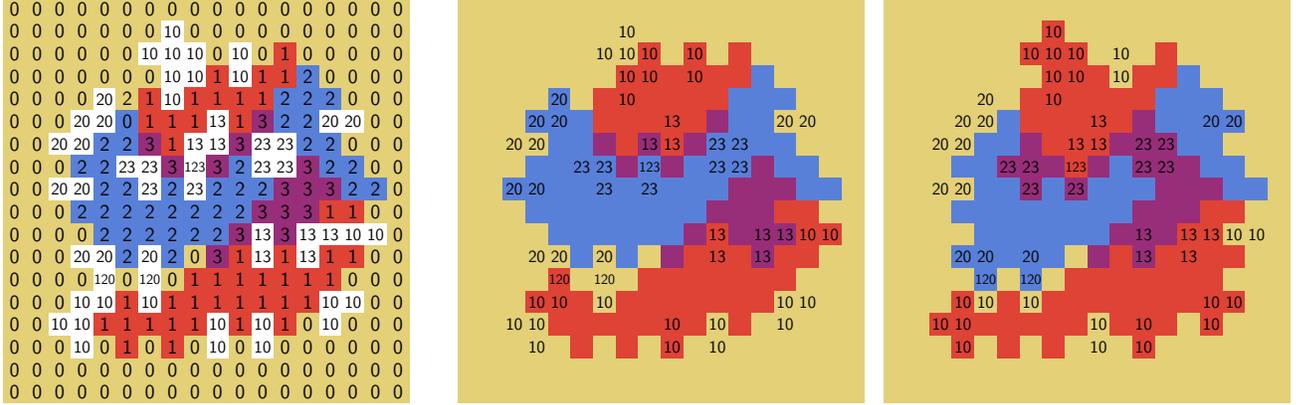

	\captionsetup[subfigure]{justification=centering}
    \begin{subfigure}[t]{.3\textwidth}
        \includegraphics[scale=0.3]{approx-sample\printable.pdf}
        \caption{A four-approximation $A$.}
        \label{fig:approx-sample}
    \end{subfigure}%
    \begin{subfigure}{25pt}
        \quad
    \end{subfigure}%
    \begin{subfigure}[t]{.6\textwidth}
        \includegraphics[scale=0.3]{approx-four-section-a\printable.pdf}~~~%
        \includegraphics[scale=0.3]{approx-four-section-b\printable.pdf}
        \caption{Two possible four-sections approximated by $A$.}
        \label{fig:approx-four-sections}
    \end{subfigure}
    \caption{A four-approximation and two four-sections approximated by it are illustrated. Vertices belonging to some $A_l$ are known to be in $K_l$; these are depicted in~(\textsc{\subref{fig:approx-sample}}) by $l$ and a corresponding color. Vertices belonging to $\DilemmaOdd ij \cup \DilemmaEven ij$ are not known precisely; these are depicted by $ij$ and a white background. When a vertex belongs to more than one such set, we write all participating indices.}
    \label{fig:approx}
\end{figure}

A four-approximation is illustrated in Figure~\reffig{fig:approx-sample}.
In a four-approximation, each vertex of $\Z^d$ either belongs to some $A_i$, in which case it is determined to be in $K_i$, or it belongs to some $\DilemmaEven ij$, in which case it is determined to be in $K_{ij}$, or it belongs to at least one $\DilemmaOdd ij$, in which case it is determined to be in $K_{12j}$. Although a vertex may belong to more than one $\DilemmaOdd ij$, the reader may wish to have in mind the case where every vertex belongs to a single such set, in which case any vertex in $\DilemmaOdd ij$ is determined to be in $K_{ij}$.
Thus, one may regard part~\eqref{it:FA13} as constituting a restriction on the structure of the set of vertices which are not determined to be in any $K_i$.
The relevance of parts~\eqref{it:FA9}-\eqref{it:FA11} will become apparent through their application below.
Here and in the rest of the section, we maintain the convention that $i$ denotes an element in $\{1,2\}$, $j$ an element in $\{0,3\}$, and $l$ any element in $\{0,1,2,3\}$.

Fix integers $L,M \ge 0$ and a four-approximation $A$.
For $K \in \FS_{L,M}(A)$, define
\begin{equation}\label{eq:def-D}
\begin{aligned}
D_{\ins i j} = D_{\ins i j}(K) &:= \DilemmaOdd ij \cap K_i ,\\
D_{\out i j} = D_{\out i j}(K) &:= \DilemmaEven ij \cap K_j ,
\end{aligned}
\end{equation}
and note that $\{D_{\ins ij},D_{\out ij}\}_{i,j}$ are pairwise disjoint by~\refFA{1}.
A key property of this definition is that if one knows $(D_{\ins ij},D_{\out ij})_{i,j}$ then one may recover $K \in \FS_{L,M}(A)$.
That is, the map $K \mapsto (D_{\ins ij}, D_{\out ij})_{i,j}$ is injective on $\FS_{L,M}(A)$.
Indeed, using~\refFA{1} and~\refFA{9}, one may check that
\begin{align*}
K_j &= A_j \cup D_{\out 1j} \cup D_{\out 2j} \cup \big((\DilemmaOdd 1j \cup \DilemmaOdd 2j) \setminus (D_{\ins 1j} \cup D_{\ins 2j})\big) ,\\
K_i &= A_i \cup D_{\ins i0} \cup D_{\ins i3} \cup (\DilemmaEven {i}{0} \setminus D_{\out i0}) \cup (\DilemmaEven {i}{3} \setminus D_{\out i3}).
\end{align*}
See Figure~\reffig{fig:recovery-pairs} and Figure~\reffig{fig:recovery-triples} for an illustration.
In fact, both maps $K \mapsto (D_{\ins ij})_{i,j}$ and $K \mapsto (D_{\out ij})_{i,j}$ are injective on $\FS_{L,M}(A)$, as follows immediately from~\refFA{12}, which in our notation becomes
	\begin{equation}\label{eq:recover-D-out-from-D-ins}
	\begin{aligned}
	D_{\ins ij} &= \DilemmaOdd ij \cap N(\DilemmaEven ij \setminus D_{\out ij}), \\
	D_{\out ij} &= \DilemmaEven ij \cap N(\DilemmaOdd ij \setminus D_{\ins ij}).
	\end{aligned}
	\end{equation}
Finally, \refFA{11} (togther with~\refFA{13}) implies that every vertex in $D_{\ins ij} \cup D_{\out ij}$ has many regular boundary edges. This is later used in~\eqref{eq:choice-of-shift-direction} to find a suitable direction $\down$ for the shift.
We remark that in our construction of the four-approximations in Section~\ref{sec:approx}, the singularities of $K$ are known, in the sense that $K_l \cap K^{\anomaly} \subset A_l$, and moreover, the subgraph induced by $\DilemmaOdd ij \cup \DilemmaEven ij$ contains no isolated vertices. However, we have no need for these properties in this section.

\begin{figure}
	\centering
	\includegraphics[scale=1]{recovery-pairs\printable.pdf}
	\caption{Reconstructing the four-section in an unknown region of two adjacent vertices.
		Here, $u \in \DilemmaEven 10$, $v \in \DilemmaOdd 20$, the neighbors of $u$ (except $v$) belong to $A_1$ ({\color{fsc1\grayscale}$\blacksquare$}) and the neighbors of $v$ (except $u$) belong to $A_0$ ({\color{fsc0\grayscale}$\blacksquare$}).
		The two possible configurations of $K$ are depicted in~\textsc{(b,c)}. Each such configuration is determined by $D_{\ins 10}$ and $D_{\out 10}$. In~\textsc{(b)}, $u \notin D_{\out 10}$ and $v \in D_{\ins 10}$. In~\textsc{(c)}, $u \in D_{\out 10}$ and $v \notin D_{\ins 10}$.}
	\label{fig:recovery-pairs}
	\bigskip
	\includegraphics[scale=1]{recovery-triples\printable.pdf}
	\caption{Reconstructing the four-section in an unknown region bordering both $K_1$ and $K_2$.
		Here, $u \in \DilemmaEven 10$, $v \in \DilemmaOdd 10 \cap \DilemmaOdd 20$, $w \in \DilemmaEven 20$, the neighbors of $u$ (except $v$) belong to $A_1$ ({\color{fsc1\grayscale}$\blacksquare$}), the neighbors of $v$ (except $u$ and $w$) belong to $A_0$ ({\color{fsc0\grayscale}$\blacksquare$}) and the neighbors of $w$ (except $v$) belong to $A_2$ ({\color{fsc2\grayscale}$\blacksquare$}).
		The three possible configurations of $K$ are depicted in~\textsc{(b,c,d)}. Each such configuration is determined by $D_{\ins 10}$, $D_{\ins 20}$ and $D_{\out 10}$, $D_{\out 20}$. In~\textsc{(b)}, $u \notin D_{\out 10}$, $v \in D_{\ins 10}$ and $w \in D_{\out 20}$. In~\textsc{(c)}, $u \in D_{\out 10}$, $v \in D_{\ins 20}$ and $w \notin D_{\out 20}$. In~\textsc{(d)}, $u \in D_{\out 10}$, $v \notin D_{\ins 10} \cup D_{\ins 20}$ and $w \in D_{\out 20}$.}
	\label{fig:recovery-triples}
\end{figure}

Recall that we wish to bound the probability of the event
\[ \Col_{L,M}(A) := \big\{ f \in \Col ~:~ \FS(f) \cap \FS_{L,M}(A) \neq \emptyset \big\} = \bigcup_{K \in \FS_{L,M}(A)} \Col_K .\]
In order to so, we use the so-called technique of flows, captured by the following simple lemma.

\begin{lemma}\label{lem:flow}
Let $E$ and $F$ be two events in a discrete probability space and
let $\epsilon>0$. If there exists a map $\nu \colon E \times F \to [0,1]$ such that $\sum_{f \in F} \nu(e,f) = 1$ for every $e\in E$ and $\sum_{e \in E} \nu(e,f) \Pr(e) \le \epsilon \Pr(f)$ for every $f\in F$ then $\Pr(E) \leq \epsilon \Pr(F)$.
\end{lemma}
\begin{proof}
We have
\[ \Pr(E)
= \sum_{e \in E} \Pr(e)
= \sum_{e \in E} \Pr(e) \sum_{f \in F} \nu(e,f)
= \sum_{f \in F} \sum_{e \in E} \nu(e,f) \Pr(e)
\le \epsilon \sum_{f \in F} \Pr(f) = \epsilon \Pr(F) . \qedhere \]
\end{proof}
One should think of $\nu$ as prescribing a flow of $\nu(e,f)\Pr(e)$ mass from $e$ to $f$. Thus, assumptions on the total mass exiting each $e$ and entering each $f$ allow to compare $\Pr(E)$ and $\Pr(F)$.
For our purpose, we shall define a flow
\begin{equation}\label{eq:flow-def}
\nu \colon \Col_{L,M}(A) \times \Col \to [0,1] .
\end{equation}
For the sake of clarity, we define $\nu$ on each four-section separately. Thus, for every $K \in \FS_{L,M}(A)$, we shall define
\[ \nu_{K\vphantom{()}} \colon \Col_K \times \Col \to [0,1] ,\]
and then set $\nu(f,g) := \nu_{\bar{K}(f)}(f,g)$, where $\bar{K}(f)$ is an arbitrary element of $\FS(f) \cap \FS_{L,M}(A)$.
One may regard $\nu_{K\vphantom{()}}$ as a weighted variant of ${\sf T}_K$. In particular,
\begin{equation}\label{eq:flow-image}
\nu_{K\vphantom{()}}(f,g) > 0 \quad\iff\quad g \in \cG_{K,f} := \big\{ {\sf T}_K(f,h) ~:~ h \in \cH_K \big\} .
\end{equation}

One may now consider defining $\nu_K$ uniformly by $\nu_K(f,g)=\1_{\cG_{K,f}}(g)/|\cG_{K,f}|$.
Using the facts that $|\cG_{K,f}| = |\cH_K|$ and $\mu^\tau_{\Lambda,\beta}(f) = e^{-\beta M} \cdot \mu^\tau_{\Lambda,\beta}(g)$ by Lemma~\ref{lem:existence-of-transformation}\ref{it:transformation-probability}, together with the trivial bound $\sum_f \1(\nu(f,g)>0) \le |\FS_{L,M}(A)|$, one then obtains the bound $\mu^\tau_{\Lambda,\beta}(\Col_{L,M}(A)) \le |\FS_{L,M}(A)| \cdot e^{-\beta M} |\cH_K|^{-1}$, recovering the ineffective union bound argument.
The problem with such a uniform flow is that we are essentially counting the number of preimages of $g$, i.e., the number of $f$ such that $g \in \cG_{K,f}$ for some $K \in \FS(f) \cap \FS_{L,M}(A)$, which greatly varies among different $g$. To overcome this, we define $\nu$ in a weighted manner, biasing $\nu(f,g)$ according to the number of such preimages, as to balance the total weight $\sum_f \nu(f,g)$ entering various $g$. This will allow us to apply Lemma~\ref{lem:flow} to obtain a better bound (recalling that $\mu^\tau_{\Lambda,\beta}(f)/\mu^\tau_{\Lambda,\beta}(g) = e^{-\beta M}$ whenever $\nu(f,g)>0$).

\begin{figure}
	\centering
	\includegraphics[scale=1]{flow-pairs\printable.pdf}
	\caption{The contribution to the flow due to an unknown pair $(u,v)$ surrounded by known vertices. Here, $u \in \DilemmaEven 10$, $v \in \DilemmaOdd 10$, the neighbors of $u$ (except $v$) belong to $A_1$ ({\color{fsc1\grayscale}$\blacksquare$}) and the neighbors of $v$ (except $u$) belong to $A_0$ ({\color{fsc0\grayscale}$\blacksquare$}).
The three possible configurations of $f$ having $f(u) \neq f(v)$ are depicted in~\textsc{(b,c,d)}. Each such configuration is mapped by ${\sf T}_K$ to two possible configurations of $g$, yielding four configurations in total, depicted in~\textsc{(e,f,g,h)}.
The two configurations in~\textsc{(e,h)} have unique preimages, while the two in~\textsc{(f,g)} have two preimages.
In cases~\textsc{(b,d)}, $w \in B^0_{1, 10}$, while in case~\textsc{(c)}, $v \in B_{0,10}$. In the former cases, the image having a unique preimage is characterized by having $g(w)=2$. The transition weights are chosen as to balance the total weight entering each configuration $g$.}
	\label{fig:flow-pairs}
\end{figure}

To understand the variation in the number of preimages, consider the simplest case of an unknown pair $(u,v) \in \DilemmaEven 10 \times \DilemmaOdd 10$ of adjacent vertices (with $v=u^{\up}$), surrounded by known vertices (in which case, $N(u)\setminus\{v\} \subset A_1$ and $N(v)\setminus\{u\} \subset A_0$), as illustrated in Figure~\reffig{fig:flow-pairs}. In this case, the local configuration $(f(u),f(v))$ has two possible singular states $(0,0)$ and $(1,1)$ and three possible non-singular states $(0,1)$, $(0,2)$ and $(2,1)$. The variation is a consequence of the behavior of ${\sf T}_K$ with respect to the non-singular states. Namely, the states $(0,1)$ and $(0,2)$ may lead to a local configuration $(g(v),g(u),g(u^{\down}))$ which has either one or two (local) preimages. As shown in Figure~\reffig{fig:flow-pairs}, these two states are characterized by the fact that $u^{\down}$ is known and $u$ is unknown, and the configuration $g$ having a unique preimage is characterized by the fact that $g(u^{\down})=2$.
The general case may involve more complex constellations of unknown vertices (see Figure~\reffig{fig:flow}). However, as we will show, the case when $\DilemmaOdd 10 \cup \DilemmaEven 10$ consists only of isolated adjacent pairs (and perhaps isolated vertices) turns out to be the extremal case in some sense.

Thus, we are led to the following definition for $\nu_K$.
First, observe that $\{ K_i \cap K_j^{\down} \}_{i,j}$ is a partition of $\intB^{\down} K_{12}$.
Using the information in $A$, we further partition $K_i \cap K_j^{\down}$ into
\begin{align*}
B^1_{1,ij} = B^1_{1,ij}(K) &:= K_i \cap K_j^{\down} \cap \DilemmaOdd ij^c \cap (\DilemmaEven ij^c)^{\down}  ,\\
B^0_{1,ij} = B^0_{1,ij}(K) &:= K_i \cap K_j^{\down} \cap \DilemmaOdd ij^c \cap (\DilemmaEven ij)^{\down}  ,\\
B_{0,ij} = B_{0,ij}(K) &:= K_i \cap K_j^{\down} \cap \DilemmaOdd ij .
\end{align*}
We write $B^1_1$, $B^0_1$ and $B_0$ for the union of these sets over $i$ and $j$.
Finally, denoting $B = B(K)$, we define
\[ \nu_{K\vphantom{()}}(f,g) := \begin{cases}
 (1/2)^{|B \cap B^1_1| + |B \cap B_0|} \cdot (1/4)^{|B \cap B^0_1 \setminus g^{-1}(2)|} \cdot (3/4)^{|B \cap B^0_1 \cap g^{-1}(2)|} &\text{if } g \in \cG_{K,f} \\
 0 &\text{if } g \notin \cG_{K,f}
 \end{cases} .\]

We now determine the choice of the direction $\down$.
We may assume that $\FS_{L,M}(A)$ is non-empty, since otherwise Lemma~\ref{lem:prob-of-approx-enhanced} holds trivially.
Therefore, there exists some $K \in \FS_{L,M}(A)$.
By~\refFA{12} and~\refFA{11},
\begin{align*}
\partial(\DilemmaOdd ij,\DilemmaEven ij) &~\subset~ \partial(D_{\ins ij}, \DilemmaEven ij) \cup \partial(\DilemmaOdd ij, D_{\out ij}) ~=~ \partial(D_{\ins ij} \cup D_{\out ij}, \DilemmaOdd ij \cup \DilemmaEven ij) , \\
 \partial(K_i,K_j) &~\supset~ \partial(D_{\ins ij}, \DilemmaEven ij ^c) \cup \partial(D_{\out ij}, \DilemmaOdd ij ^c) ~=~ \partial(D_{\ins ij} \cup D_{\out ij}, (\DilemmaOdd ij \cup \DilemmaEven ij)^c) .
\end{align*}
Thus, Lemma~\ref{lem:sizes2} and~\refFA{13} imply that $|\partial(\DilemmaOdd ij,\DilemmaEven ij)| \le \frac{\sqrt{d}}{2d-\sqrt{d}} |\partial(K_i,K_j)|$, and so
\[ \sum_{i,j} |\partial(\DilemmaOdd ij,\DilemmaEven ij)| \le \sum_{i,j} \frac{|\partial(K_i,K_j)|}{\sqrt{d}} = \frac{L}{\sqrt{d}} .\]
Therefore, since $|\partial(\DilemmaOdd ij,\DilemmaEven ij)|=\sum_s |\DilemmaOdd ij \cap (\DilemmaEven ij)^s|$, where $s$ ranges over the $2d$ unit vectors in $\Z^d$, there exists a unit vector $s=\down$ such that
\begin{equation}\label{eq:choice-of-shift-direction}
\sum_{i,j} |\DilemmaOdd ij \cap (\DilemmaEven ij)^{\down}| \le \frac{L}{d^{3/2}} .
\end{equation}

The first part of Lemma~\ref{lem:prob-of-approx-enhanced} is an immediate consequence of Lemma~\ref{lem:flow} and the following lemma whose proof is given in Section~\ref{sec:flow}.
In fact, the second part of Lemma~\ref{lem:prob-of-approx-enhanced} follows just the same, since Lemma~\ref{lem:existence-of-transformation}\ref{it:transformation-anomalies} and~\eqref{eq:flow-image} imply that $\nu_K(f,g)=0$ when $\{u,v\} \notin \partial K$, $f(u)=f(v)$, $g(u) \neq g(v)$ and $g(u^{\down}) \neq g(v^{\down})$.

\begin{lemma}\label{lem:flow-bounds}
Let $\down$ be a unit vector in $\Z^d$ satisfying~\eqref{eq:choice-of-shift-direction}.
Then the above defined flow $\nu$ satisfies
\begin{align}\label{eq:flow-out}
&\sum_g \nu(f,g) = 1 , &&f \in \Col_{L,M}(A) , \\
\label{eq:flow-in}
&\sum_f \nu(f,g) \cdot \mu^\tau_{\Lambda,\beta}(f) \le \exp(-c L/d - \beta M + CM) \cdot \mu^\tau_{\Lambda,\beta}(g) , && g \in \Col .
\end{align}
\end{lemma}

\subsection{The shift transformation}
\label{sec:shift}

In this section, we prove Lemma~\ref{lem:existence-of-transformation}.
Fix a four-section $K$.
For $f \in \Col_K$, denote $f' := \text{Flip}_K(f)$ and $f'' := \text{Shift}_K(f')$.
Throughout this section, we repeatedly use the fact that $K_0$ is even and $K_3$ is odd, without an explicit reference to~\refFSodd{}.
By~\refAFSvalues{} and by the definition of $\text{Flip}_K$, we have
\begin{equation}\label{eq:values-of-f'}
\begin{aligned}
f'(v)=0 &\quad\text{for all } v \in \intB K_0 ,\\
f'(v)=1 &\quad\text{for all } v \in \intB K_3 ,\\
f'(v)=1 &\quad\text{for all } v \in \Odd \cap (\intB K_1 \cup \intB K_2) ,\\
f'(v)=2 &\quad\text{for all } v \in \Even \cap (\intB K_1 \cup \intB K_2) .
\end{aligned}
\end{equation}
Thus, by the definition of $\text{Shift}_K$,
\begin{equation}\label{eq:values-of-f''-lower}
f''(v) = f'(v) = f'(v^{\up})-1 \quad \text{for all } v \in \intB^{\down} K_{03} .
\end{equation}
Moreover,
\begin{equation}\label{eq:values-of-f''}
\begin{aligned}
f''(v)=1 \text{ and } f''(u)=0 &\quad\text{for all } v \in K_{12} \cap K_0^{\down}  \text{ and } u \in N(v) ,\\
f''(v)=0 \text{ and } f''(u)=1 &\quad\text{for all } v \in K_{12} \cap K_3^{\down} \text{ and } u \in N(v) .
\end{aligned}
\end{equation}
To see the first statement, let $v \in K_{12} \cap K_0^{\down}$ and $u \in N(v)$, and note that $v$ is odd and $u$ is even.
We have $f''(v)=1$ by the definition of $\text{Shift}_K$.
If $u \notin K_{12}$ then $u \in \intB K_0$, and thus, $f''(u)=f'(u)=0$ by \eqref{eq:values-of-f'}.
If $u \in K_{12}$ then $u^{\up} \in \extB K_0 \subset \intB K_{12} \cup \intB K_3$, and thus, $f''(u)=0$ if $u^{\up} \in \intB K_3$ and $f''(u)=f'(u^{\up})-1=0$ if $u^{\up} \in \intB K_{12}$ by \eqref{eq:values-of-f'}.
We omit the proof of the second statement as it follows similar lines.
It is straightforward to check that \eqref{eq:values-of-f''} is equivalent to
\begin{equation}\label{eq:values-of-f''2}
\begin{aligned}
f''(v)=1 &\quad\text{for all } v \in \Odd \cap (\intB^{\down} K_{12})^+ ,\\
f''(v)=0 &\quad\text{for all } v \in \Even \cap (\intB^{\down} K_{12})^+ .
\end{aligned}
\end{equation}

\subsubsection{Part (\ref{it:transformation-injective})} \label{lem:shift-injective}
Let $g \in \Col$. We show that there exist at most one $f \in \Col_K$ and one $h\colon \intB^{\down} K_{12} \to \{0,1\}$ such that $g = {\sf T}_K(f,h)$.
Assume that $(f,h)$ is such a pair.
Then $f'$ is given by
\[ f'(v) = \begin{cases}
 g(v) 						&\text{if } v \in K_{03} \\
 g(v^{\down}) + 1				&\text{if } v \in K_{12}
\end{cases} .\]
Indeed, the only non-trivial case to check is when $v \in \intB^{\up} K_{12}$, in which case $g(v^{\down})+1=f''(v^{\down})+1$, and thus it follows from \eqref{eq:values-of-f''-lower}.

Recalling the definition of $\text{Mod}_{K,h}$, it is straightforward to check that $h$ is given by
\[ h(v) = \begin{cases}
 g(v)-1	&\text{if } v \in K_{12} \cap K_0^{\down} \\
 -g(v)		&\text{if } v \in K_{12} \cap K_3^{\down}
\end{cases} . \]
Thus, since $f$ is clearly determined by $f'$, part~\eqref{it:transformation-injective} follows. 

\subsubsection{Parts (\ref{it:transformation-probability}) and (\ref{it:transformation-anomalies})} \label{lem:shift-removes-anomalies}
Let $f \in \Col_K$, let $h \in \cH_K$ and denote $g := {\sf T}_K(f,h)=\text{Mod}_{K,h}(f'')$.
Recalling that $M = |\partial K \setminus \partial K_{12}|$, parts~\eqref{it:transformation-probability} and~\eqref{it:transformation-anomalies} will follow if we show the following four statements:
\begin{align}
E^{\anomaly}(f) &\supset \partial K \setminus \partial K_{12} , \label{eq:shift-reduces-anomalies-0} \\
E^{\anomaly}(f') &= E^{\anomaly}(f) \setminus (\partial K \setminus \partial K_{12}) \subset E(K_{03}) \cup E(K_{12}) , \label{eq:shift-reduces-anomalies-1} \\
E^{\anomaly}(f'') &= \big(E^{\anomaly}(f') \cap E(K_{03})\big) \uplus \big(E^{\anomaly}(f') \cap E(K_{12})\big)^{\down} , \label{eq:shift-reduces-anomalies-2} \\
E^{\anomaly}(g) &= E^{\anomaly}(f'') . \label{eq:shift-reduces-anomalies-3}
\end{align}

Recalling \eqref{eq:good-bad-boundaries-of-K}, \eqref{eq:shift-reduces-anomalies-0} is immediate.
Using \eqref{eq:values-of-f'} and the definition of $\text{Flip}_K$, it is straightforward to check \eqref{eq:shift-reduces-anomalies-1}.

To see \eqref{eq:shift-reduces-anomalies-2}, we begin by showing that the right-hand side is contained in $E^{\anomaly}(f'')$. Indeed, the fact that $E^{\anomaly}(f') \cap E(K_{03}) \subset E^{\anomaly}(f'')$ is immediate from the definition of $\text{Shift}_K$, while the fact that $(E^{\anomaly}(f') \cap E(K_{12}))^{\down} \subset E^{\anomaly}(f'')$ follows from the definition of $\text{Shift}_K$ and~\eqref{eq:values-of-f''-lower}.
Moreover, one may easily check that these sets are disjoint.
Towards showing the opposite containment, let $e \in E^{\anomaly}(f'')$.
We must show that either $e \in E^{\anomaly}(f') \cap E(K_{03})$ or $e^{\up} \in E^{\anomaly}(f') \cap E(K_{12})$.
To this end, we consider three cases:

Case 1. $e \subset K_{03}$; follows immediately from the definition of $\text{Shift}_K$.

Case 2. $e \subset K_{12}^{\down}$; follows from~\eqref{eq:values-of-f''-lower} and the definition of $\text{Shift}_K$.

Case 3. $e \cap \intB^{\down} K_{12} \neq \emptyset$; is impossible by~\eqref{eq:values-of-f''2}.

To see that one of these cases must hold, assume towards a contradiction that neither holds.
As case 1 does not hold, we may write $e=\{u,v\}$, where $u \in K_{12}$.
As case 3 does not hold, $u^{\up} \in K_{12}$.
As case 2 does not hold, $v^{\up} \in K_{03}$.
As case 2 does not hold, $v \in K_{03}$.
If $u$ is odd then $v \in K_0$ and~\eqref{eq:values-of-f'} implies that $f'(u^{\up})=2$ and $f'(v)=0$.
If $u$ is even then $v \in K_3$ and~\eqref{eq:values-of-f'} implies that $f'(u^{\up})=1$ and $f'(v)=1$.
Thus, using the definition of $\text{Shift}_K$, we see that in either case, $e \notin E^{\anomaly}(f'')$, which is a contradiction.

Finally, we show \eqref{eq:shift-reduces-anomalies-3}.
By the definition of $\text{Mod}_{K,h}$, we have $g(v)=f''(v)$ for $v \notin \intB^{\down} K_{12}$.
Moreover, by the definitions of $\text{Shift}_K$ and $\text{Mod}_{K,h}$ and the assumption on $h$, we have $g(v)=f''(v)$ for $v \notin B(K)$.
Thus, since~\eqref{eq:values-of-f''2} implies that the vertices in $B(K)$ are not singularities of $f''$, it suffices to show that these vertices are not singularities of $g$.
Indeed, since $B(K)$ is an independent set, this follows from~\eqref{eq:values-of-f''} and the definition of $\text{Mod}_{K,h}$.

\subsubsection{Part (\ref{it:transformation-easy-recover})}
Let $f \in \Col_K$, let $h\colon \intB^{\down} K_{12} \to \{0,1\}$ and denote $g := {\sf T}_K(f,h)$.
We show that $K_{12} \cap g^{-1}(2)^{\up} \subset f^{-1}(0)$.
Let $v \in K_{12}$ and assume that $f(v)\neq 0$. We must show that $g(v^{\down}) \neq 2$.
To this end, we consider three cases.

If $v^{\down} \in K_{03}$ then $g(v^{\down})=f''(v^{\down})=f'(v^{\down}) \neq 2$, by~\eqref{eq:values-of-f'}, since $v^{\down} \in \intB K_{03}$.

If $v^{\down} \in K_{12}$ and $v \in K_1$ then $g(v^{\down})=f''(v^{\down})=f'(v)-1=f(v)-1 \neq 2$.

If $v^{\down} \in K_{12}$ and $v \in K_2$ then $g(v^{\down})=f''(v^{\down})=f'(v)-1=-f(v)-1 \neq 2$.

\subsection{The flow}
\label{sec:flow}

In this section, we prove Lemma~\ref{lem:flow-bounds}.
We first show that \eqref{eq:flow-out} holds.
To this end, let $K \in \FS_{L,M}(A)$ and let $f \in \Col_K$.
Let $g \in \cG_{K,f}$ so that $g={\sf T}_K(f,h)$ for a unique $h \in \cH_K$, since ${\sf T}_K$ is injective, by Lemma~\ref{lem:existence-of-transformation}\ref{it:transformation-injective}.
Observe that, by the definition of ${\sf T}_K$ and $\text{Mod}_{K,h}$, for $v \in \intB^{\down} K_{12}$, we have $h(v)=1 \iff g(v)=2$. Thus, we may write
\begin{equation*}\label{eq:flow-def-with-h}
\nu_{K\vphantom{()}}(f,g) = (1/2)^{|B \cap B^1_1| + |B \cap B_0|} \cdot (1/4)^{|B \cap B^0_1 \cap h^{-1}(0)|} \cdot (3/4)^{|B \cap B^0_1 \cap h^{-1}(1)|} .
\end{equation*}
We interpret the above as defining a probability distribution $\nu_{K,f}$ on $\cH_K$. The value of $h$ at each vertex $v \in B$ is independently decided according to a Bernoulli random variable $h(v)$, where $\nu_{K,f}(h(v)=1)=1/2$ if $v \in B^1_1 \cup B_0$ and $\nu_{K,f}(h(v)=1)=3/4$ if $v \in B^0_1$.
Thus, $\nu_K(f,g)=\nu_{K,f}({\sf T}_K(f,h)=g)$, and \eqref{eq:flow-out} follows.

It remains to prove \eqref{eq:flow-in}.
In light of Lemma~\ref{lem:existence-of-transformation}\ref{it:transformation-probability} and by the definition of $\mu^\tau_{\Lambda,\beta}$, \eqref{eq:flow-in} will follow if we show that
\begin{equation}\label{eq:flow-in2}
\sum_f \nu(f,g) \leq \exp(-c L/d + CM) , \quad g \in \Col .
\end{equation}
In order to prove~\eqref{eq:flow-in2}, we henceforth fix a coloring $g \in \Col$ and define
\begin{equation}\label{eq:def-D-outtwo}
\begin{aligned}
\DilemmaEvenTwo ij &:= \DilemmaEven ij \setminus g^{-1}(2)^{\up} ,\\
D_{\outtwo ij} = D_{\outtwo i j}(K) &:= \DilemmaEvenTwo ij \cap K_j = D_{\out ij} \setminus g^{-1}(2)^{\up} .
\end{aligned}
\end{equation}
In the particular case of an isolated unknown pair, depicted in Figure~\reffig{fig:flow-pairs}, we observed that when $g(w)=2$, one may uniquely recover the preimage locally. This is generalized by the fact that the map $K \mapsto ( D_{\outtwo ij}(K) )_{i,j}$ is injective on
\[ \FS_{L,M}(A,g) := \big\{ K \in \FS_{L,M}(A) ~:~ \text{there exists } f \in \Col_K \text{ such that } g \in \cG_{K,f} \big\} .\]
In fact, this is a consequence of a special relation between the sets $D_{\ins ij}$ and $D_{\outtwo ij}$, namely that together they form a minimal cover of $\DilemmaOdd ij \cup \DilemmaEvenTwo ij$ (see Figure~\reffig{fig:flow}).
Let us now define this notion.

\begin{figure}
	\centering
	\begin{subfigure}[t]{.42\textwidth}
		\centering
		\includegraphics[scale=0.4]{flow-1\printable.pdf}
		\caption{\grayscaleText{A region of unknown vertices in $\DilemmaOdd ij \cup \DilemmaEven ij$, denoted by $*$ with a white background.}{A region of unknown vertices in $\DilemmaOdd ij \cup \DilemmaEven ij$. The vertices in $\DilemmaOdd ij$ ($\DilemmaEven ij$) are denoted by $*$ with a white (gray) background.} The vertices in $A_i$ are depicted by {\color{fsc1\grayscale}$\blacksquare$} and those in $A_j$ by {\color{fsc0\grayscale}$\blacksquare$}.}
		\label{fig:flow-1}
	\end{subfigure}%
	\begin{subfigure}{20pt}
		\quad
	\end{subfigure}%
	\begin{subfigure}[t]{.52\textwidth}
		\centering
		\includegraphics[scale=0.4]{flow-2\printable.pdf}
		\includegraphics[scale=0.4]{flow-3\printable.pdf}
		\caption{On the left, a coloring $g$ in the image of the transformation. Given such a $g$, we may identify the set $\DilemmaEven ij \setminus \DilemmaEvenTwo ij$ of unknown vertices in $v \in \DilemmaEven ij$ having $g(v^{\down})=2$. These vertices (depicted by $\times$) do not play a role in the recovery of the four-section.}
		\label{fig:flow-2-3}
	\end{subfigure}
	\vspace{5pt}
	
	\begin{subfigure}[t]{1\textwidth}
		\centering
		\includegraphics[scale=0.4]{flow-4a\printable.pdf}
		\includegraphics[scale=0.4]{flow-4b\printable.pdf}\qquad
		\includegraphics[scale=0.4]{flow-5a\printable.pdf}
		\includegraphics[scale=0.4]{flow-5b\printable.pdf}
		\caption{Two examples of $(D_{\ins ij},D_{\outtwo ij})$ and their corresponding four-section. Observe that $D_{\ins ij} \cup D_{\outtwo ij}$ is a minimal vertex-cover of $\DilemmaOdd ij \cup \DilemmaEvenTwo ij$. This is manifested in the figure by the fact that there are no two adjacent $*$. The corresponding four-sections are obtained by adding each vertex in $D_{\ins ij}$ to $K_i$ and each vertex in $D_{\outtwo ij}$ to $K_j$, determining for each vertex in $\DilemmaEven ij \setminus \DilemmaEvenTwo ij$ whether it belongs to $K_j$ or $K_i$ according to whether or not it is adjacent to a vertex in $\DilemmaOdd ij \setminus D_{\ins ij}$, and then determining the remaining vertices according to their neighbors.}
		\label{fig:flow-5}
	\end{subfigure}
	\caption{Bounding the total flow entering a given $g$ from various colorings $f$. For each four-section $K \in \FS_{L,M}(A,g)$, at most one $f \in \Col_K$ contributes a positive amount to this total (see Lemma~\ref{lem:existence-of-transformation}\ref{it:transformation-injective}). In turn, $K$ is determined by $(D_{\ins ij},D_{\outtwo ij})_{i,j}$. Moreover, $D_{\ins ij} \cup D_{\outtwo ij}$ is a minimal cover of $\DilemmaOdd ij \cup \DilemmaEvenTwo ij$ (see Lemma~\ref{lem:D-is-minimal-cover}).
Since the contribution can be expressed in terms of $D_{\ins ij}$ and $D_{\outtwo ij}$ (see Lemma~\ref{lem:flow-bound}), this allows us to bound the total flow entering $g$ using Lemma~\ref{lem:sum-over-minimal-covers}. The figure illustrates the process of recovering $K$ from $D_{\ins ij} \cup D_{\outtwo ij}$ in a region surrounded by known vertices.}
	\label{fig:flow}
\end{figure}

Let $W$ be a bipartite graph with bipartition classes $W_\ins$ and $W_\out$.
We say that a set $V \subset W_\ins \cup W_\out$ is a {\em vertex-cover} of $W$ if every edge of $W$ has an endpoint in $V$.
Given sets $V_\ins \subset W_\ins$ and $V_\out \subset W_\out$, we say that $(V_\ins,V_\out)$ is a {\em minimal cover} of $(W_\ins,W_\out)$ if $V_\ins \cup V_\out$ is a minimal vertex-cover of $W_\ins \cup W_\out$.
Denote by $MC(W_\ins,W_\out)$ the set of all minimal covers of $(W_\ins,W_\out)$.
One should note that
\begin{equation}\label{eq:minimal-cover-def}
(V_\ins,V_\out) \in MC(W_\ins,W_\out) \quad\iff\quad
V_\ins = N(W_\out \setminus V_\out) \quad\text{and}\quad V_\out = N(W_\ins \setminus V_\ins) .
\end{equation}

\begin{lemma}\label{lem:D-is-minimal-cover}
	Let $K \in \FS_{L,M}(A,g)$. Then $(D_{\ins ij}, D_{\outtwo ij})$ is a minimal cover of $(\DilemmaOdd ij, \DilemmaEvenTwo ij)$.
\end{lemma}

\begin{proof}
	Observe that $D_{\outtwo ij} = \DilemmaEvenTwo ij \cap N(\DilemmaOdd ij \setminus D_{\ins ij})$ follows from~\eqref{eq:recover-D-out-from-D-ins}, by subtracting $g^{-1}(2)^{\up}$ from both sides in the second row. Thus, by~\refFA{2} and~\eqref{eq:minimal-cover-def}, it remains to show that
	\[ D_{\ins ij} = \DilemmaOdd ij \cap N(\DilemmaEvenTwo ij \setminus D_{\outtwo ij}) .\]
	To see this, first observe that the right-hand side is contained in $D_{\ins ij}$ by~\eqref{eq:recover-D-out-from-D-ins}.
	Towards showing the opposite containment, let $v \in D_{\ins ij}$ and let $f \in \Col_K$ be such that $g \in \cG_{K,f}$.
	By~\refFA{13} and~\refFA{11}, $v$ has a neighbor in $K_j$ so that $v \in \intB K_i$. Thus, by~\refAFSneighbor{}, $f(u) \neq 0$ for some $u \sim v$.
	By~\refAFSvalues{}, $u \notin K_j$.
	Thus, \refFA{11} now implies that $u \in \DilemmaEven ij$.
	Moreover, Lemma~\ref{lem:existence-of-transformation}\ref{it:transformation-easy-recover} implies that $g(u^{\down}) \neq 2$, so that $u \in \DilemmaEvenTwo ij$. Hence, $u \in \DilemmaEvenTwo ij \setminus D_{\outtwo ij}$ so that $v \in N(\DilemmaEvenTwo ij \setminus D_{\outtwo ij})$, as required.
	\end{proof}


The next lemma, whose proof we postpone to Section~\ref{sec:flow-bound}, enables us to bound the flow in terms of the sets $D_{\ins ij}$ and $D_{\outtwo ij}$.

\begin{lemma}\label{lem:flow-bound}
	Let $K \in \FS_{L,M}(A)$ and $f \in \Col_K$ be such that $g \in \cG_{K,f}$. Then
	\[ \nu_K(f,g) \le (3/4)^{L/2d} \cdot 4^{9M + Ld^{-3/2}} \cdot (2/3)^{\sum_{i,j} |D_{\ins ij}(K)|} \cdot (1/3)^{\sum_{i,j} |D_{\outtwo ij}(K)|} .\]
\end{lemma}

Lemma~\ref{lem:existence-of-transformation}\ref{it:transformation-injective} and the fact that the map $K \mapsto (D_{\ins ij})_{i,j}$ is injective on $\FS_{L,M}(A)$ imply that
\[ \sum_{K,f \colon g \in \cG_{K,f}} (2/3)^{\sum_{i,j} |D_{\ins ij}(K)|} \cdot (1/3)^{\sum_{i,j} |D_{\outtwo ij}(K)|} \le \prod_{i,j} \sum_{(U_\ins, U_\outtwo) \in \mathcal{D}_{ij}} (2/3)^{|U_\ins|} \cdot (1/3)^{|U_\outtwo|} ,\]
where
\[ \mathcal{D}_{ij} := \big\{ (D_{\ins ij}(K), D_{\outtwo ij}(K)) ~:~ K \in \FS_{L,M}(A,g) \big\} .\]
Thus, in light of~\eqref{eq:flow-image} and Lemma~\ref{lem:flow-bound},~\eqref{eq:flow-in2} will follow if we show that the right-hand side in the above displayed equation is at most $1$.
Since $\mathcal{D}_{ij} \subset \MC(\DilemmaOdd ij, \DilemmaEvenTwo ij)$ by Lemma~\ref{lem:D-is-minimal-cover}, this follows immediately from the next lemma.

\begin{lemma}\label{lem:sum-over-minimal-covers}
	Let $W$ be a finite bipartite graph with bipartition classes $W_\ins$ and $W_\out$.
	Then, for any $0 \le p \le 1$,
	\[ \sum_{(U_\ins,U_\out) \in MC(W_\ins,W_\out)} p^{|U_\ins|} \cdot (1-p)^{|U_\out|} \le 1 .\]
\end{lemma}

In fact, when $0<p<1$, equality holds if and only if every vertex in $W$ has at most one neighbor. The proof of Lemma~\ref{lem:sum-over-minimal-covers} is given in Section~\ref{sec:minimal-covers}.

\subsection{Bounding the flow}
\label{sec:flow-bound}

In this section, we prove Lemma~\ref{lem:flow-bound}.
Let $K \in \FS_{L,M}(A)$, let $f \in \Col_K$ and let $g \in \cG_{K,f}$. Note that $\nu_K(f,g)=\prod_{i,j} \nu_{i,j}(f,g)$, where
\[ \nu_{i,j}(f,g) := (1/2)^{|B \cap B^1_{1, ij}| + |B \cap B_{0, ij}|} \cdot (1/4)^{|B \cap B^0_{1, ij} \setminus g^{-1}(2)|} \cdot (3/4)^{|B \cap B^0_{1, ij} \cap g^{-1}(2)|} .\]
We begin by observing that if $v,u \in \intB^{\down} K_{12}$ are adjacent then $\{ v^{\up}, u^{\up} \} \in \partial(K_0,K_3)$ so that $v,u \in K^{\anomaly}$. Thus, the vertices in $\intB^{\down} K_{12} \setminus (K^{\anomaly})^{\down}$ are isolated in $\intB^{\down} K_{12}$, implying that they are contained in every maximal independent set of $\intB^{\down} K_{12}$. Hence, $\intB^{\down} K_{12} \setminus (K^{\anomaly})^{\down} \subset B$ and
\[ \nu_{i,j}(f,g)
 \le (1/2)^{|B^1_{1, ij}| + |B_{0, ij}|} \cdot (1/4)^{|B^0_{1, ij} \setminus g^{-1}(2)|} \cdot (3/4)^{|B^0_{1, ij} \cap g^{-1}(2)|} \cdot 4^{|K^{\anomaly}|} .\]
Next, we express the sets $B^0_{1,ij}$ and $B_{0,ij}$ in terms of $D_{\ins ij}$ and $D_{\out ij}$.
By~\refFA{11}, we have
\begin{align*}
B^0_{1,ij} &= D_{\out ij}^{\down} \cap \DilemmaOdd ij^c , \\
B_{0, ij} &= D_{\ins ij} \cap (D_{\out ij} \cup \DilemmaEven ij^c)^{\down} ,
\end{align*}
so that, by~\eqref{eq:def-D-outtwo},
\begin{align*}
 B^0_{1,ij} \setminus g^{-1}(2) &= D_{\outtwo ij}^{\down} \cap \DilemmaOdd ij^c = D_{\outtwo ij}^{\down} \setminus (D_{\outtwo ij}^{\down} \cap \DilemmaOdd ij) ,\\
  B^0_{1,ij} \cap g^{-1}(2) &= (D_{\out ij} \setminus D_{\outtwo ij})^{\down} \cap \DilemmaOdd ij^c = (D_{\out ij} \setminus D_{\outtwo ij})^{\down} \setminus \big((D_{\out ij} \setminus D_{\outtwo ij})^{\down} \cap \DilemmaOdd ij \big) ,\\
  B_{0,ij} &= D_{\ins ij} \setminus (D_{\ins ij} \cap (\DilemmaEven ij \setminus D_{\out ij})^{\down}) .
\end{align*}
Putting these together, we obtain
\begin{align*}
\nu_{i,j}(f,g)
 &\le (1/2)^{|B^1_{1, ij}| + |B_{0, ij}|} \cdot (1/4)^{|B^0_{1, ij} \setminus g^{-1}(2)|} \cdot (3/4)^{|B^0_{1, ij} \cap g^{-1}(2)|} \cdot 4^{|K^{\anomaly}|} \\
%
%
 &= (1/2)^{|B^1_{1,ij}| + |D_{\ins ij}|} \cdot (1/4)^{|D_{\outtwo ij}|} \cdot (3/4)^{|D_{\out ij} \setminus D_{\outtwo ij}|} \cdot 4^{|K^{\anomaly}|} \\
&\quad \cdot (1/2)^{-|D_{\ins ij} \cap (\DilemmaEven ij \setminus D_{\out ij})^{\down}|} \cdot (1/4)^{-|(D_{\outtwo ij})^{\down} \cap \DilemmaOdd ij|} \cdot (3/4)^{-|(D_{\out ij} \setminus D_{\outtwo ij})^{\down} \cap \DilemmaOdd ij|} \\
%
%
 &\le (1/2)^{|D_{\ins ij}|} \cdot (1/4)^{|D_{\outtwo ij}|} \cdot (3/4)^{|B^1_{1,ij}| + |D_{\out ij} \setminus D_{\outtwo ij}|} \cdot 4^{|K^{\anomaly}|+|D_{\ins ij} \cap (\DilemmaEven ij \setminus D_{\out ij})^{\down}| + |D_{\out ij}^{\down} \cap \DilemmaOdd ij|} \\
  &\le (1/2)^{|D_{\ins ij}|} \cdot (1/4)^{|D_{\outtwo ij}|} \cdot (3/4)^{|B^1_{1,ij}| + |D_{\out ij} \setminus D_{\outtwo ij}|} \cdot 4^{|K^{\anomaly}|+|\DilemmaOdd ij \cap (\DilemmaEven ij)^{\down}|} ,
\end{align*}
where in the last inequality, we used the fact that $D_{\ins ij} \cap (\DilemmaEven ij \setminus D_{\out ij})^{\down}$ and $D_{\out ij}^{\down} \cap \DilemmaOdd ij$ are disjoint subsets of $\DilemmaOdd ij \cap (\DilemmaEven ij)^{\down}$.

We now combine the contributions of $\nu_{i,j}(f,g)$ to $\nu_K(f,g)$ over all pairs $(i,j)$.
Denote $D_\ins := \cup_{i,j} D_{\ins ij}$, $D_\out := \cup_{i,j} D_{\out ij}$ and $D_\outtwo := \cup_{i,j} D_{\outtwo ij}$ and recall that $B^1_1 = \cup_{i,j} B^1_{1,ij}$. Then
\[ \nu_K(f,g)
 \le (1/2)^{|D_\ins|} \cdot (1/4)^{|D_\outtwo|} \cdot (3/4)^{|B^1_1|+|D_\out| - |D_\outtwo|} \cdot 4^{4|K^{\anomaly}|+\sum_{i,j} |\DilemmaOdd ij \cap (\DilemmaEven ij)^{\down}|} .\]
Noting that $K_i \cap K_j^{\down} = B^1_{1,ij} \cup B^0_{1,ij} \cup B_{0,ij} \subset B^1_{1,ij} \cup D_{\ins ij} \cup (D_{\out ij})^{\down}$, we have
\[ |B^1_1| + |D_\ins| + |D_\out| \ge |\intB^{\down} K_{12}| .\]
Therefore, by~\eqref{eq:choice-of-shift-direction} and Lemma~\ref{cl:size-of-boundary}, and as $|K^{\anomaly}| \le 2M$, we obtain
\begin{align*}
\nu_K(f,g)
 &\le (1/2)^{|D_\ins|} \cdot (1/4)^{|D_\outtwo|} \cdot (3/4)^{|\intB^{\down} K_{12}| - |D_\ins| - |D_\outtwo|} \cdot 4^{8M + Ld^{-3/2}} \\
&\le (2/3)^{|D_\ins|} \cdot (1/3)^{|D_\outtwo|} \cdot (3/4)^{L/2d} \cdot  4^{9M + Ld^{-3/2}} .
\end{align*}

\subsection{Minimal vertex-covers}
\label{sec:minimal-covers}

Let $G$ be a graph.
We say that a set $U \subset V(G)$ is a {\em vertex-cover} of $G$ if every edge of $G$ has an endpoint in $U$.
We say that a vertex-cover is \emph{minimal} if it is minimal with respect to inclusion.
Denote by $\MC(G)$ the set of all minimal vertex-covers of $G$.
Lemma~\ref{lem:sum-over-minimal-covers} is a special case of the following lemma.

\begin{lemma}\label{lem:minimum-cover}
		Let $G=(V,E)$ be a finite graph, let $\{p_v\}_{v \in V}$ be non-negative numbers and assume that $p_u + p_v \le 1$ for all $\{u,v\} \in E$.
		Then
		\[ \sum_{U \in \MC(G)} \prod_{u \in U} p_u \le 1 .\]
\end{lemma}

\begin{proof}
	Let $E = \{e_1,e_2,\dots,e_n\}$ be an ordering of the edges of $G$ and write $e_i = \{u_i,v_i\}$.
	We prove the statement using a probabilistic method, by constructing a random set $X_n \subset V$ such that $\Pr(X_n=U) \ge \prod_{u \in U} p_u$ for every $U \in \MC(G)$.
	Since $\sum_{U \in \MC(G)} \Pr(X_n=U) = \Pr(X_n \in \MC(G)) \le 1$, the lemma will follow.

	For each $1 \le i \le n$, let $Y_i$ be an independent variable such that $Y_i=\{u_i\}$ with probability $p_{u_i}$, $Y_i=\{v_i\}$ with probability $p_{v_i}$ and $Y_i=\emptyset$ with probability $1-p_{u_i}-p_{v_i}$. Define $X_0 := \emptyset$ and, for $1 \le i \le n$, define $X_i:=X_{i-1}$ if $X_{i-1} \cap e_i \neq \emptyset$ and $X_i:=X_{i-1} \cup Y_i$ otherwise.
	
	Let $U \in \MC(G)$.
	Define $U_0:=\emptyset$. If $U_{i-1} \cap e_i \neq \emptyset$ then define $U_i := U_{i-1}$. Otherwise, choose $y_i \in U \cap e_i$ (which exists since $U$ is a cover) and define $U_i := U_{i-1} \cup \{y_i\}$.
	Let us show that $U_n=U$.
	By construction, $U_n \subset U$.
	To see that $U \subset U_n$, let $u \in U$.
	Since $U$ is a minimal cover, there exists an $1 \le i \le n$ such that $e_i = \{u,v\}$ and $v \notin U$.
	In particular, either $y_i=u$ so that $u \in U_i \subset U_n$, or $u \in U_{i-1} \subset U_n$.
	To see that $\Pr(X_n=U) \ge \prod_{u \in U} p_u$, observe that, for every $1 \le i \le n$,
	\[
		\Pr(X_i=U_i \mid X_1=U_1, \dots, X_{i-1}=U_{i-1}) = \begin{cases}1 &\text{if } U_{i-1} \cap e_i \neq \emptyset \\ p_{y_i} &\text{if } U_{i-1} \cap e_i = \emptyset \end{cases}.\]
	Thus,
	\[ \Pr(X_n=U) \ge \Pr(X_1=U_1,\dots,X_n=U_n)
	 = \prod_{i=1}^n \begin{cases}1 &\text{if } U_{i-1} \cap e_i \neq \emptyset \\ p_{y_i} &\text{if } U_{i-1} \cap e_i = \emptyset \end{cases} = \prod_{u \in U} p_u  . \qedhere \]
\end{proof}

We remark that the above lemma holds also for hypergraphs with essentially the same proof (where the above condition on $\{p_v\}$ becomes $\sum_{u \in e} p_u \le 1$ for every hyper-edge $e$).

%
\section{Approximations}
\label{sec:approx}
%

This section is devoted to the proof of Lemma~\ref{lem:family-of-FA}. That is, given integers $L,M \ge 0$, we show that there exists a small family of four-approximations $\cA$ (see Definition~\ref{def:four-approx}) which covers $\FS_{L,M,\rho}$, in the sense that each four-section $K \in \FS_{L,M,\rho}$ has a four-approximation in $\cA$.
The construction of the family of four-approximations is done progressively, in four steps.
In the $\ell$-th step, we define the notion of a \emph{level-$\ell$-approximation of a four-section $K$}.
We also use the term \emph{level-$\ell$-approximation} for $A$ which is a level-$\ell$-approximation of at least one four-section.
Formally, a level-$\ell$-approximation will be a tuple of subsets of $\Z^d$.
However, the sense in which each level-$\ell$-approximation actually approximates $K$ will vary among levels.
In each level, we exploit different structural properties of four-sections, which, by means of a small enumeration, allow us to obtain significantly more information about $K$.
Consequentially, the number of level-$\ell$-approximations needed to cover $\FS_{L,M,\rho}$ increases from one level to the next.
To prove Lemma~\ref{lem:family-of-FA}, we must control this increase so that the number of sets does not become too big.
Before describing the nature of these approximations, we give a short outline of this section.

Denote by $\FS^\ell(A)$ the collection of all four-sections $K$ such that $A$ is a level-$\ell$-approximation of $K$. Also, for a family of level-$\ell$-approximations $\cA$, denote (with slight abuse of notation) $\FS^\ell(\cA) :=  \cup_{A \in \cA} \FS^\ell(A)$.
We begin by showing that there is a small family $\cA$ of level-$1$-approximations which covers $\FS_{L,M,\rho}$ in the sense that $\FS_{L,M,\rho} \subset \FS^1(\cA)$ (Lemma~\ref{lem:family-of-level-1-approx}).
Next, for $\ell=1$ and $\ell=2$, we show that for each level-$\ell$-approximation $A$ there exists a small family
$\cA'$ of level-$(\ell+1)$-approximations which covers $\FS^{\ell}(A)$ in the sense that $\FS^{\ell}(A) \subset \FS^{\ell+1}(\cA')$ (Lemmas~\ref{lem:family-of-level-2-approx} and~\ref{lem:family-of-level-3-approx}).
Then we show that every level-$3$-approximation $A$ gives rise to a single level-$4$-approximation $A'$ such that $\FS^3(A) \subset \FS^4(A')$ (Lemma~\ref{lem:algorithm-for-level-4-approx}).
Together, this gives a small family of level-$4$-approximations which covers $\FS_{L,M,\rho}$ (Corollary~\ref{lem:family-of-level-4-approx}).
Finally, we show that every level-$4$-approximation $A$ gives rise to a four-approximation which covers $\FS^4(A)$ (Lemma~\ref{lem:level-4-approx-gives-FA}). This yields Lemma~\ref{lem:family-of-FA}.

Let us now give some details about the essence of the approximations at each level.
A level-$1$-approximation will be a $2$-tuple, consisting of a small set $U$ which approximates the boundary of $K$ away from singularities,
and of another set $W$ which consists of all the singularities in $K$. These sets will closely approximate $\intB K$ in the sense that if we remove $N(U) \cup W \cup N_{d/18}(W)$ from $\Z^d$, then each connected component in the remaining vertices will be contained entirely in $K_l$ for some $l\in\{0,1,2,3\}$. By recording the location in $K$ of all such components which are not too small and of each vertex in $U \cup W \cup N_{d/18}(W)$, we obtain a level-$2$-approximation.
This, in turn, is a $4$-tuple, recording the exact location in $K$ of all but a small number of vertices which are near the boundary of $K$.
The level-$3$-approximation adds geometric information about the boundaries of the sets $K_{ij}$, for $i\in\{1,2\}$ and $j\in\{0,3\}$.
In particular, if the location of a vertex
in, say, $\{K_{10},K_{23}\}$ is not recorded by the level-$3$-approximation, then the locations of the vast majority of its neighbors are.
Finally, a level-$4$-approximation is a data structure which records for each vertex to which $K_l$'s it may belong.
This data structure is a convenient way to fully utilize the information given by the level-$3$-approximation.

In the rest of the section, we maintain the convention that $i$ denotes an element in $\{1,2\}$, $j$ an element in $\{0,3\}$ and $l$ any element in $\{0,1,2,3\}$. In light of the duality between $K_1$ and $K_2$ and between $K_0$ and $K_3$, it will also be useful to denote $\bar{l} := 3-l$.
We assume throughout this section that $d$ is large enough for our arguments to hold.

\subsection{Level-1-approximation.}

In this section, we construct level-$1$-approximations, and show the existence of a small family of these objects which approximates any four-section in $\FS_{L,M,\rho}$.
We start with a basic property of four-sections, which we then exploit for the construction of level-1-approximations.

\begin{lemma}\label{lem:four-cycle-property}
	Let $K$ be a four-section, let $\{u,v\} \in \partial K_{12}$ with $u \in K_i$ and $v \in K_j$ and let $e \in \Z^d$ be a unit vector. Then one of the following holds:
	\begin{itemize}
		\item $\{u,u+e\} \in \partial(K_i,K_j)$ or $\{v,v+e\} \in \partial(K_i,K_j)$,
		 \item $\{u,u+e\} \in \partial K \setminus \partial K_{12}$ or $\{v,v+e\} \in \partial K \setminus \partial K_{12}$.
		 \item $\{u+e,v+e\} \in \partial(K_1,K_2)$.
	\end{itemize}
\end{lemma}

For an illustration of all possible configurations of the locations of $u+e$ and $v+e$  in $K$ when $u \in K_1$ and $v \in K_0$, see Figure~\reffig{fig:four-cycles}.

\begin{proof}
	Assume that the first two statements do not hold, so that $u+e \in K_i \cup K_{\bar{j}}$ and $v+e \in K_{\bar{i}} \cup K_j$, where we recall that $\bar{l} = 3-l$.
	Since $v \in \intB K_j$, \refFSodd{} implies that $u+e \notin \intB K_{\bar{j}}$ and $v+e \notin \intB K_j$. Hence, applying \refFSodd{} again, $u \notin K_{\bar{j}}$ and $u+e \notin K_j$ imply that $u+e \notin K_{\bar{j}}$ and $v+e \notin K_j$.
	Thus, $u+e \in K_i$ and $v+e \in K_{\bar{i}}$ so that, in particular, $\{u+e,v+e\} \in \partial(K_1,K_2)$.
\end{proof}

\begin{figure}
	\centering
	\includegraphics[scale=1]{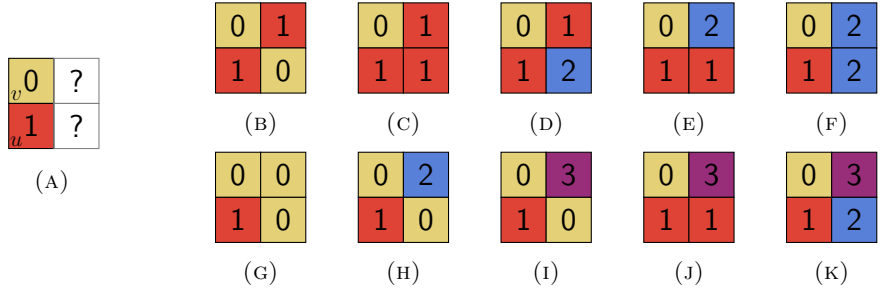}
	\caption{Four-cycle constellations: $(u,u+e,v+e,v)$ is a four-cycle with $u \in K_1$ and $v \in K_0$. The numbers represent the phases of the vertices, i.e., their location in the partition $\{K_0,K_1,K_2,K_3\}$.
	In cases~\textsc{(b,c,d,g,h,i)}, either $\{u,u+e\}$ or $\{v,v+e\}$ belongs to $\partial(K_0,K_1)$.
	In cases~\textsc{(d,e,f,i,j,k)}, at least one of $\{u,u+e\},\{v+v+e\},\{u+e,v+e\}$ belongs to $\partial K \setminus \partial K_{12}$.}
	\label{fig:four-cycles}
\end{figure}

\begin{cor}\label{cor:boundary-edge-sees-much-boundary}
	Let $K$ be a four-section and let $\{u,v\} \in \partial K_{12}$ with $u \in K_i$ and $v \in K_j$. Then
	\[ |\partial u \cap \partial K| + |\partial v \cap \partial K| \ge 2d .\]
	Moreover,
	\[ |N(u) \cap K_j| + |N(v) \cap K_i| + |\partial N(u) \cap \partial N(v) \cap \partial K \setminus \partial K_{12}| \ge 2d .\]
\end{cor}

\begin{proof}
	For the first part, it suffices to show that for any unit vector $e$ in $\Z^d$, either $\{u,u+e\} \in \partial K$ or $\{v,v+e\} \in \partial K$, while for the second part, it suffices to show that for any such unit vector, either $u+e \in K_j$ or $v+e \in K_i$ or one of $\{u,u+e\},\{v,v+e\},\{u+e,v+e\}$ belongs to $\partial K \setminus \partial K_{12}$.
	Both statements follow from Lemma~\ref{lem:four-cycle-property}.
\end{proof}

For a four-section $K$, denote the {\em revealed vertices} in $K$ by
\[ K^{\rev} := \{ v \in \Z^d ~:~ |\partial v \cap \partial K_{12}| \geq d \} .\]
That is, a vertex is revealed if it sees the regular boundary of $K$ in at least half of the $2d$ directions.
We say that a set $U \subset \Z^d$ \emph{separates} $K$ if every edge in $\partial K$ has an endpoint in $U$, or equivalently, if every connected component of $\Z^d \setminus U$ is contained in $K_l$ for some $l \in \{0,1,2,3\}$.

\begin{cor}\label{cor:revealed-and-anomalies-separate}
	Let $K$ be a four-section. Then $K^{\rev} \cup K^{\anomaly}$ separates $K$.
\end{cor}
\begin{proof}
	Let $\{u,v\} \in \partial K$. We must show that either $u$ or $v$ belongs to $K^{\rev} \cup K^{\anomaly}$.
	Assume that $u,v \notin K^{\anomaly}$.
	Then, by Corollary~\ref{cor:boundary-edge-sees-much-boundary}, $|\partial u \cap \partial K_{12}| + |\partial v \cap \partial K_{12}| \geq 2d$, so that, in particular, either $u$ or $v$ belongs to $K^{\rev}$.
\end{proof}

\begin{definition}[level-$1$-approximation]
	\label{def:level-1-approx}
	A pair $(U,W)$ of subsets of $\Z^d$ is called a \emph{level-$1$-approximation} of a four-section $K$ if $W=K^{\anomaly}$ and the following holds:
	\begin{enumerate}[\qquad(a)]
		\item \label{it:def-level-1-approx-separate} $N(U) \cup W \cup N_{d/18}(W)$ separates $K$.
		\item \label{it:def-level-1-approx-many-revealed} $N_{d/9}(K_l \cap K^{\rev}) \subset N(K_l \cap U)$ for all $l \in \{0,1,2,3\}$.
		\item \label{it:def-level-1-approx-size} $|U| \le C  (L/d^{3/2} + M/\sqrt{d}) \sqrt{\log d}$, where $L := |\partial K_{12}|$ and $M := |\partial K \setminus \partial K_{12}|$.
	\end{enumerate}
\end{definition}

We remark that part~\eqref{it:def-level-1-approx-many-revealed} guarantees that vertices which are adjacent to many revealed vertices are, in a sense, captured by level-1-approximations. More details will be given after the definition of a level-2-approximation. The constants $1/18$ and $1/9$, which arise in our proof, are of no particular significance and (after suitable modifications) others constants could be used.

We now show the existence of level-1-approximations. The proof is accompanied by Figure~\ref{fig:existence-of-U}.

\begin{lemma}\label{lem:existence-of-U}
	For every four-section $K$ there exists a set $U \subset (\intB K)^+$
	such that $K \in \FS^1((U,K^{\anomaly}))$, i.e., such that $(U,K^{\anomaly})$ is a level-$1$-approximation of $K$.
\end{lemma}

\begin{figure}
	\centering
	\includegraphics[scale=1.2]{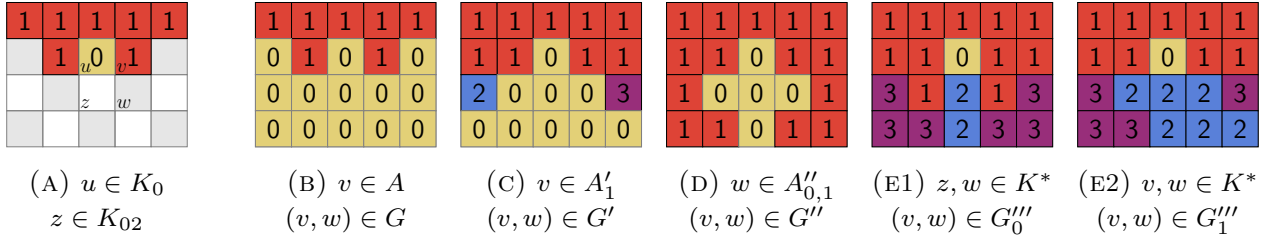}
	\caption{Constructing the separating set. In~\textsc{(a)}, a non-singular revealed vertex $u \in K_0$ is depicted along with a neighbor $z \in K_{02}$. Every four-cycle $(u,v,w,z)$ such that $v \in K_1$ falls into one of four types. A typical situation corresponding to each such type is shown in the figure. In~\textsc{(b)}, $v$ has at least $r$ regular boundary edges. In~\textsc{(c)}, $v$ belongs to at least $r$ four-cycles having a singularity. In~\textsc{(d)}, $w \in K_0$ and has at least $2d-2r$ neighbors in $K_1$. In~\textsc{(e)}, one of the edges $\{z,w\}$ or $\{v,w\}$ is a singularity. At least $1/4$ of all such four-cycles belong to the same type. If it is the type corresponding to~\textsc{(b)}, then $u$ is adjacent to many vertices which have many regular boundary edges and one such neighbor of $u$ is included in $B$; if it is~\textsc{(c)}, then $u$ is adjacent to many vertices which are themselves adjacent to many singularities and one such neighbor is included in $B'_1$; if it is~\textsc{(d)}, then $z$ is adjacent to many vertices in $K_0$ which have almost all their neighbors in $K_1$ and $z$ is included in $B''_{0,1}$; if it is~\textsc{(e)}, then either $z$ has many singular edges so that $z \in B'''$ or $u$ is adjacent to many singularities so that $u \in N_{d/18}(K^{\anomaly})$. The set $U$ is obtained by combining $B$, $B'_i$, $B''_{i,j}$ and $B'''$.}
	\label{fig:existence-of-U}
\end{figure}

\begin{proof}
	Let $K$ be a four-section and denote $L := |\partial K_{12}|$ and $M := |\partial K \setminus \partial K_{12}|$.
	Denote $r:=\sqrt{d \log d}$ and $t:=\frac{d}9$.
	Define
	\begin{align*}
		A &:= \intB K \cap \{ v : |\partial N(v) \cap \partial K \setminus \partial K_{12}| \ge r \}, \\
		A'_l &:= \intB K_l \cap \{ v : |\partial v \cap \partial K_{12}| \ge r \}, &&l \in \{0,1,2,3\} ,\\
		A''_{l,m} &:= \intB K_l \cap N_{2d-2r}(K_m), &&l,m \in \{0,1,2,3\}\text{ such that }m\notin \{l,\bar{l}\}.
	\end{align*}
	Observe that, by Lemma~\ref{lem:sizes}, we have
	\begin{equation}\label{eq:a-s}
		|A| \le \frac{4dM}r,\quad \sum_l |A'_l| \le \frac{2L}r\quad\text{ and }\quad \sum_{l,m} |A''_{l,m}| \le \sum_{l,m} \frac{|\partial (K_l,K_m)|}{2d-2r} = \frac{L}{d-r}.
	\end{equation}
	We now use Lemma~\ref{lem:existence-of-covering2} with $A$ and $A'_l$ to obtain sets $B \subset \intB K$ and $B'_l \subset \intB K_l$ such that
	\begin{equation}\label{eq:u-s}	
		\begin{aligned}
			|B|\le\frac{2\log d}t|A|\quad \text{and} \quad N_t(A)&\subset N(B),\\
			|B'_l|\le\frac{2\log d}t|A'_l| \quad\text{and}\quad N_t(A'_l)&\subset N(B'_l).
		\end{aligned}
	\end{equation}
	We also define $B''_{l,m}:=N_t(A''_{l,m})\setminus K_m$ for $l,m\in\{0,1,2,3\}$ such that $m\notin \{l,\bar{l}\}$. By Lemma~\ref{lem:sizes}, we have
	\[ |B''_{l,m}| \le \frac{2r}{t}|A''_{l,m}| .\]
	Denote $B':=\cup_l B'_l$, $B'':=\cup_{l,m}B''_{l,m}$ and $B''':= \{ v : |\partial v \cap \partial K \setminus \partial K_{12}| \ge t/2\}$, and observe that $|B'''|\le \frac{4M}{t}$.
	Finally, we define $U:=B\cup B'\cup B''\cup B'''$. Clearly, $U \subset (\intB K)^+$. Moreover, by~\eqref{eq:a-s} and~\eqref{eq:u-s}, we have
	\[ |U| \le \frac{2\log d}t\left(\frac{4dM}r+ \frac{2L}r \right)+\frac{2rL}{t(d-r)}+ \frac{4M}{t} .\]
	Plugging in the values of $r$ and $t$, we obtain
	\[ |U|\le  \frac{36\log d}d \cdot \frac{2dM + L}{\sqrt{d\log d}} +\frac{36\sqrt{\log d}L}{d^{3/2}}+ \frac{36M}{d}\le
	C\sqrt{\log d}\left(\frac{L}{d^{3/2}} + \frac{M}{\sqrt{d}}\right).\]

	In order to show that $(U,K^{\anomaly})$ is a level-$1$-approximation of $K$, it remains to show that~\refLAAseparate{} and~\refLAAmanyRevealed{} hold.
	To see~\refLAAmanyRevealed{}, note that $K_l \cap K^{\rev} \subset A'_l$ so that, by~\eqref{eq:u-s},
\[ N_{d/9}(K_l \cap K^{\rev}) \subset N_t(A'_l) \subset N(B'_l) \subset N(K_l \cap U).\]
	Towards showing~\refLAAseparate{}, observe that, by Corollary~\ref{cor:revealed-and-anomalies-separate}, it suffices to show that
\[ K^{\rev} \subset N(U) \cup K^{\anomaly} \cup N_{d/18}(K^{\anomaly}) .\]
To see this, let $u \in K_l \cap K^{\rev} \setminus K^{\anomaly}$ for some $l \in \{0,1,2,3\}$,
	and let $m \in \{0,1,2,3\} \setminus \{l,\bar{l}\}$ be such that $u \in N_{d/2}(K_m)$.
	By~\refFSisolated{}, there exists a vertex $z \in N(u) \setminus K_m$.
	Let $F$ denote the set of pairs $(v,w)$ such that $(u,v,w,z)$ is a four-cycle and $v \in K_m$, and note that $|F| \ge d/2-1$.
	Denote
	\begin{align*}
		G'''_0 &:= \{ (v,w) \in F ~:~ \{z,w\} \in \partial K \setminus \partial K_{12} \}, \\
		G'''_1 &:= \{ (v,w) \in F ~:~ \{v,w\} \in \partial K \setminus \partial K_{12} \},
	\end{align*}
	and
	$G''' := G'''_0 \cup G'''_1$.
	Note that, by Lemma~\ref{lem:four-cycle-property}, for any $(v,w) \in F \setminus G'''$, we have $w \in K_l$.
	Denote
	\begin{align*}
		G &:= \{ (v,w) \in F \setminus G''' ~:~ v \in A \}, \\
		G' &:= \{ (v,w) \in F \setminus G''' ~:~ v \in A'_m \}, \\
		G'' &:= \{ (v,w) \in F \setminus G''' ~:~ w \in A''_{l,m} \} .
	\end{align*}
    Corollary~\ref{cor:boundary-edge-sees-much-boundary} implies that $F \setminus G''' = G \cup G' \cup G''$. Thus, either $|G|$, $|G'|$, $|G''|$ or $|G'''|$ is at least $|F|/4 \ge t$.
	Now observe that
	if $|G| \ge t$ then $u \in N_t(A) \subset N(B)$.
	If $|G'| \ge t$ then $u \in N_t(A'_m) \subset N(B')$.
	If $|G''| \ge t$ then $z \in N_t(A''_{l,m}) \setminus K_m$ so that $u \in N(B'')$.
	If $|G'''| \ge t$ then, either $|G'''_0| \ge t/2$ so that $z \in B'''$ and $u \in N(B''')$, or $|G'''_1| \ge t/2$ so that $u \in N_{t/2}(K^{\anomaly})$.
    Hence $u \in N(U) \cup N_{d/18}(K^{\anomaly})$, as required.
\end{proof}

We now show the existence of a small family of level-$1$-approximations which covers $\FS_{L,M,\rho}$.
To do so, we exploit the connectivity of the boundary of $K$.

\begin{lemma}\label{lem:family-of-level-1-approx}
	For any integers $L,M \ge 0$, there exists a family $\cA$ of level-$1$-approximations of size
	\[ |\cA| \le \exp\left( C L d^{-3/2} \log^{3/2} d + CM \log d \right) \]
	such that $\FS_{L,M,\rho} \subset \FS^1(\cA)$.
\end{lemma}

\begin{proof}
	Recalling~\eqref{eq:L+M-is-large}, we may assume that $L+M \ge d^2$ or $M \ge 2d$.
	Let $(\Z^d)^{\otimes 5}$ be the graph obtained from $\Z^d$ by adding edges between vertices at distance at most $5$ in $\Z^d$.
    Setting $V:=\{\rho+ke_1 : 0 \le k < (L+M)^2 \}$ and $r:= CLd^{-3/2}\sqrt{\log d} +3M$, we let $\cA'$ be the collection of all subsets of $\Z^d$ of size at most $r$,
    which are connected in $(\Z^d)^{\otimes 5}$ and intersect $V^{++}$.
    Let $\cA$ be the set of pairs $(U,W)$ such that $U \cup W \in \cA'$.

    Since every vertex in $U \cup W$ is either in $U$, in $W$ or in both, and since $(\Z^d)^{\otimes 5}$ has maximum degree at most $Cd^5$, Lemma~\ref{lem:number-of-connected-graphs} implies that
    \[ |\cA| \le 3^r |\cA'| \le |V^{++}| \cdot \big(Cd^5\big)^r \le  \exp\left(CL d^{-3/2} \log^{3/2} d + CM \log d\right),\]
    where the rightmost inequality uses the fact that $|V^{++}|\le (L+M)^2 (2d+1)^2$ and our assumption that either $L+M \ge d^2$ or $M \ge 2d$.

	We now show that $\FS_{L,M,\rho} \subset \FS^1(\cA)$.
	Let $K \in \FS_{L,M,\rho}$. We must show that there exists an $A \in \cA$ such that $A$ is a level-$1$-approximation of $K$.
	By Lemma~\ref{lem:existence-of-U}, there exists a set $U \subset (\intB K)^+$ such that $(U,K^*)$ is a level-$1$-approximation of $K$. Thus, it suffices to show that $(U,K^*) \in \cA$, or equivalently, that $U \cup K^* \in \cA'$.
	
	To this end, we first show that $U \cup K^*$ is $(\Z^d)^{\otimes 5}$-connected.
	Since $U \subset (\intB K)^+$ and $K^* \subset \intB K$, we see that $\dist(v,\intB K) \le 1$ for all $v \in U \cup K^*$.
	Since $N(U) \cup K^* \cup N_{d/18}(K^*)$ separates $K$, $\intB K \subset (N(U) \cup K^* \cup N_{d/18}(K^*))^+ \subset (U \cup K^*)^{++}$ so that we obtain $\dist(v,U \cup K^*) \le 2$ for all $v \in \intB K$.
	Thus, since $\intB K$ is connected, Lemma~\ref{lem:r-connected-sets} implies that $U \cup K^*$ is $(\Z^d)^{\otimes 5}$-connected.
	
    Next, we show that $(U \cup K^*) \cap V^{++} \neq \emptyset$.
	Indeed, by Lemma~\ref{lem:isoperimetry}, $|K_{123}| \le |\partial K_{123}|^2 \le (L+M)^2$ so that there exists a vertex $v \in V \cap \intB K$.
	Since $\intB K \subset (U \cup K^*)^{++}$, there exists a $u \in U \cup K^*$ such that $\dist(u,v) \le 2$, i.e., $u \in (U \cup K^*) \cap V^{++}$.

	To conclude that $U \cup K^* \in \cA'$, it remains to show that $|U \cup K^*|\le r$. Indeed, since $|K^*| \le 2M$, this follows from~\refLAAsize{}.
\end{proof}

\subsection{Level-2-approximation.}

In this section, we construct level-$2$-approximations. 
\begin{definition}[level-$2$-approximation]
	\label{def:level-2-approx}
	A tuple $A=(A_0,A_1,A_2,A_3)$ of disjoint subsets of $\Z^d$ is called a {\em level-$2$-approximation} of a four-section $K$ if for every $l \in \{0,1,2,3\}$ the following holds:
	\begin{enumerate}[\qquad(a)]
		\item \label{it:def-weak-FA-1} $A_l \subset K_l$.
		\item \label{it:def-weak-FA-anomaly} $K_l \cap K^{\anomaly} \subset A_l$.
		\item \label{it:def-weak-FA-many-anomalies} $K_l \cap N_{d/18}(K^{\anomaly}) \subset A_l$.
		\item \label{it:def-weak-FA-sees-many-i} $N_{d/9}(K_l \cap K^{\rev}) \subset N(A_l)$.
		\item \label{it:def-weak-FA-size} $|K_l \setminus A_l| \le C(L/d + M) \sqrt{d \log d}$, where $L := |\partial K_{12}|$ and $M := |\partial K \setminus \partial K_{12}|$.
	\end{enumerate}
\end{definition}

Thus, a level-2-approximation records the locations in $K$ of all but a small number of non-singular vertices as is expressed by parts~\eqref{it:def-weak-FA-1}, \eqref{it:def-weak-FA-anomaly} and~\eqref{it:def-weak-FA-size}.
In recovering information about the four-section from the approximation, we often use parity considerations which may be invalid near large numbers of singularities. The purpose of part~\eqref{it:def-weak-FA-many-anomalies} is to overcome this obstacle.
Part~\eqref{it:def-weak-FA-sees-many-i} essentially guarantees that near every edge of the boundary $\partial(K_i,K_j)$ there is a vertex in either $A_i$ or $A_j$. While we do not prove this explicitly, we shall use part~\eqref{it:def-weak-FA-sees-many-i} to prove a related property in Lemma~\ref{lem:level-4-properties}\ref{it:def-FA-no-mixing}.

The next lemma shows that one may upgrade a level-$1$-approximation to a small family of level-$2$-approximations which cover (at least) the same set of four-sections. Here, the main property we exploit is the existence of a small set which separates $K$.

\begin{lemma}\label{lem:family-of-level-2-approx}
	For any integers $L,M \ge 0$ and any level-$1$-approximation $(U,W)$, there exists a family $\cA$ of level-$2$-approximations of size
	\[ |\cA| \le \exp\left( C L d^{-3/2} \log^{1/2} d + CM \right) \]
	such that $\FS^1((U,W)) \cap \FS_{L,M} \subset \FS^2(\cA)$.
\end{lemma}

\begin{proof}
Let $K \in \FS^1((U,W)) \cap \FS_{L,M}$. We first construct $A=A(K)$, a level-$2$-approximation of $K$. We then show that $\cA := \{A(K): K \in \FS^1((U,W)) \cap \FS_{L,M}\}$ is not too large.

Denote $Y := N(U) \cup W \cup N_{d/18}(W)$ and $X := \Z^d \setminus Y$.
Say that a connected component of $X$ is {\em small} if its size is at most $d$, and that it is {\em large} otherwise.
For $l \in \{0,1,2,3\}$, let $A'_l$ be the union of all the large components of $X$ which are contained in $K_l$, and denote $U_l := U \cap K_l$ and $W_l := (W \cup N_{d/18}(W)) \cap K_l$.
Finally, define $A_l := A'_l \cup U_l \cup W_l$.

We now show that $A:=(A_0,A_1,A_2,A_3)$ is a level-$2$-approximation of $K$.
Part \eqref{it:def-weak-FA-1} is immediate from the construction.
Parts \eqref{it:def-weak-FA-anomaly} and \eqref{it:def-weak-FA-many-anomalies} follow directly from the definitions of $W_l$ and $A_l$ and from the fact that $W=K^{\anomaly}$ by the definition of a level-$1$-approximation.
Part~\eqref{it:def-weak-FA-sees-many-i} follows from the definition of $U_l$ and~\refLAAmanyRevealed{}.
It remains to show part~\eqref{it:def-weak-FA-size}.
By~\refLAAseparate{}, $Y$ separates $K$, i.e., every connected component of $X$ is contained in $K_l$ for some $l$.
Thus, denoting by $S$ the union of all the small components of $X$, we have
\[ K_l \setminus A_l \subset S \cup N(U) .\]
Since any small component $B$ of $X$ satisfies $|B| \le |\partial B|/d$ (by Corollary~\ref{cor:isoperimetry}\ref{it:isoperimetry-small}) and $\partial B \subset \partial Y$, and since $|N_{d/18}(W)| \le 72M$ (by Lemma~\ref{lem:sizes}), we have
\begin{equation}\label{eq:rough-approx-1}
	|S| \le \frac{|\partial Y|}{d} \le \frac{2d(|N(U)| + |W| + |N_{d/18}(W)|)}{d} \le 2 |N(U)| + CM .
\end{equation}
Hence, by~\refLAAsize{},
\begin{equation}\label{eq:rough-approx-2}
|K_l \setminus A_l| \leq |S| + |N(U)| \leq 3 |N(U)| + CM \leq 6d |U| + CM \leq C(L/d + M) \sqrt{d \log d} .
\end{equation}
We have thus shown that $K \in \FS^2(A)$.

\smallskip
To conclude the proof, it remains to bound $|\cA|$.
Let $\ell$ be the number of large components of $X$, and observe that
\[ |\cA| \le 4^{|U \cup W \cup N_{d/18}(W)| + \ell} .\]
Since $|W| \le 2M$ and $|N_{d/18}(W)| \le CM$, we have
\[ |U \cup W \cup N_{d/18}(W)| \le CLd^{-3/2}\sqrt{\log d} + CM .\]
Observe that any large component $B$ of $X$ satisfies $|\partial B| \ge d^2$ (by Corollary~\ref{cor:isoperimetry}\ref{it:isoperimetry-large}) and $\partial B \subset \partial Y$. Therefore, by~\eqref{eq:rough-approx-1} and~\eqref{eq:rough-approx-2},
\[ \ell \leq |\partial Y| / d^2 \leq C(L/d + M) \sqrt{\tfrac{\log d}{d}} .\]
Thus,
$|\cA| \le \exp( C L d^{-3/2} \log^{1/2} d + CM)$,
as required.
\end{proof}

\subsection{Level-3-approximation.}

In this section, we refine the level-$2$-approximation, obtaining more geometric information about the four-section and its boundaries.
This process takes advantage of a certain property of four-sections which we call semi-oddness.
Conceptually, we think of the notion of a semi-odd pair $(U_\ins,U_\out)$ introduced below,
as an approximation of a set $U$. The sets $U_\ins$ and $U_\out$ correspond to our approximation of $U$ and $U^c$ respectively.
The properties of the semi-odd pair are derived from the fact that the ``unknown boundary'' of $U$ must be odd, that is, the boundary must be
odd in every region not contained in $U_\ins \cup U_\out$.

For a set $\xi \subset E(\Z^d)$ and a set $U \subset \Z^d$, we write $\intB^\xi U$ for the internal boundary of $U$ in the graph $(\Z^d,E(\Z^d)\setminus \xi)$, i.e.,
\[ \intB^\xi U := \big\{u \in U : \exists v \in U^c\text{ such that }\{u,v\}\in E(\Z^d) \setminus \xi \big\}.\]
A set $U \subset \Z^d$ is called {\em $\xi$-odd} if $\xi\subset \partial U$ and $\intB^\xi U \subset \Odd$.
Similarly, it is called {\em $\xi$-even} if $\xi\subset \partial U$ and $\intB^\xi U \subset \Even$.
Note that $U$ is odd if and only if $U$ is $\emptyset$-odd, and that any set $U$ is $\partial U$-odd.

The following lemma relates $\xi$-odd sets and four-sections.

\begin{lemma}\label{lem:K-ij-is-xi-odd}
	Let $K$ be a four-section.
	Then $K_{10}$ and $K_{20}$ are $\partial(K_1,K_2)$-even and $K_{13}$ and $K_{23}$ are $\partial(K_1,K_2)$-odd.
\end{lemma}
\begin{proof}
	By symmetry, it suffices to consider the case of $K_{13}$. Denote $\xi := \partial(K_1,K_2)$ and observe that $\partial(K_{13},K_{20}) = \partial K_{13}$.
Let $v \in \Even \cap K_{13}$ and let $u \in \Odd \cap K_{20}$ be adjacent to $v$. To check that $K_{13}$ is $\xi$-odd, we must show that $\{v,u\} \in \xi$.
	By~\refFSodd{}, $v \notin K_3$ and $u \notin K_0$, and thus, $v \in K_1$ and $u \in K_2$. Hence $\{v,u\} \in \xi$, as required.
\end{proof}

We say that a pair $(U_\ins,U_\out)$ of disjoint sets is {\em $\xi$-odd} if $U_\ins$ is $\xi$-odd and $U_\out$ is $\xi$-even, and we say that it is {\em semi-odd} if it is $\partial(U_\ins,U_\out)$-odd.
For a pair $(U_\ins,U_\out)$ of disjoint sets, we use the notation
\[ U_* := (U_\ins \cup U_\out)^c \]
for the \emph{unknown} vertices. Observe that, by definition, if $(U_\ins,U_\out)$ is a semi-odd pair then
\begin{equation}\label{eq:semi-odd-disjoint-boundary}
\begin{aligned}
	N(U_\out) \cap U_* &\subset \Odd ,\\
	N(U_\ins) \cap U_* &\subset \Even .
\end{aligned}
\end{equation}
We say that $(V_\ins,V_\out)$ is a {\em refinement} of $(U_\ins,U_\out)$ if $U_\ins \subset V_\ins$ and $U_\out \subset V_\out$.
A triple $(U_\ins,U,U_\out)$ is called semi-odd if $(U_\ins,U^c)$ and $(U,U_\out)$ are semi-odd. Observe that a triple $(U_\ins,U,U_\out)$ is semi-odd if and only if $U_\ins \subset U$, $U_\out \subset U^c$ and
\begin{equation}\label{eq:semi-odd-triple}
\partial(\Even \cap U, \Odd \setminus U) \subset \partial(U_\ins,U_\out) .
\end{equation}
A pair $(U_\ins,U_\out)$ of disjoint sets is called \emph{$t$-tight} if the subgraph of $\Z^d$ induced by $U_*$ has maximum degree less than $t$. That is,
\begin{equation}\label{eq:odd-t-approx}
(U_\ins,U_\out) \text{ is $t$-tight} \quad\iff\quad U_* \cap N_t(U_*) = \emptyset .
\end{equation}
See Figure~\reffig{fig:odd-approx} for an illustration of these notions.

\begin{definition}[level-$3$-approximation]
	\label{def:level-3-approx}
	A tuple $A=(A_0,A_1,A_2,A_3,A_{10},A_{23},A_{20},A_{13})$ of subsets of $\Z^d$ is called a {\em level-$3$-approximation} of a four-section $K$ if the following holds:
	\begin{enumerate}[\qquad(a)]
		\item \label{it:def-weak-FA-extends-rough-FA} $(A_0,A_1,A_2,A_3)$ is a level-$2$-approximation of $K$.
		\item \label{it:def-weak-FA-semiodd} $(A_{13},K_{13},A_{20})$ and $(A_{23},K_{23},A_{10})$ are semi-odd triples.
		\item \label{it:def-weak-FA-sqrt-d-approx} $(A_{13},A_{20})$ and $(A_{23},A_{10})$ are $\sqrt{d}$-tight semi-odd pairs.
	\end{enumerate}
\end{definition}

Thus, part~\eqref{it:def-weak-FA-semiodd} reflects the semi-odd properties of $K_{13}$ and $K_{23}$ as stated in Lemma~\ref{lem:K-ij-is-xi-odd}, and part~\eqref{it:def-weak-FA-sqrt-d-approx} guarantees that every unknown vertex is almost entirely surrounded by known vertices.

\begin{figure}
	\centering
	\includegraphics[scale=0.5]{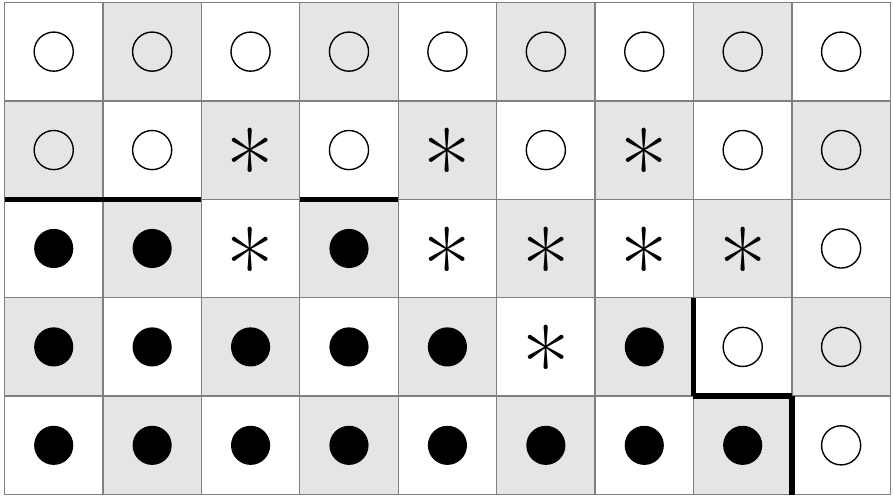}\qquad\qquad
	\includegraphics[scale=0.5]{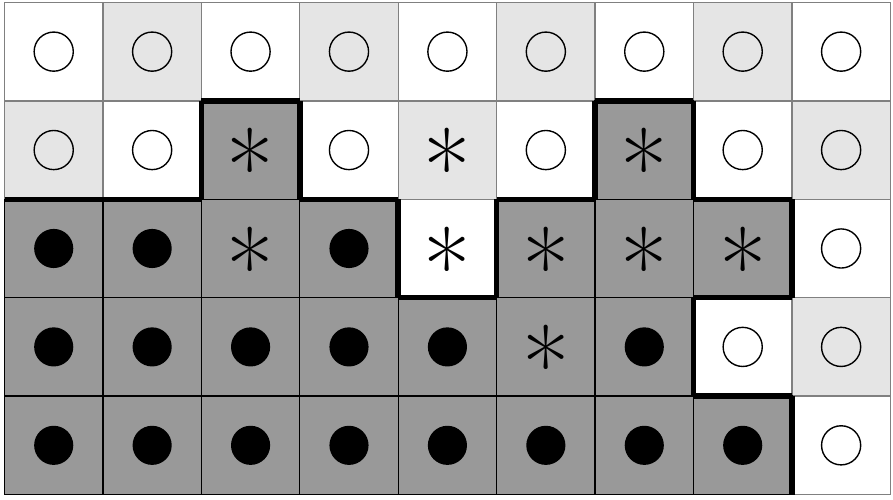}
	\caption{On the left, a $4$-tight semi-odd pair $(U_\ins,U_\out)$. The vertices in $U_\ins$, $U_\out$ and $U_*$ are denoted by $\ins$, $\out$ and $*$, respectively. The known boundary edges $\partial(U_\ins,U_\out)$ are depicted by a bold line. On the right, a set $U$ satisfying that $(U_\ins,U,U_\out)$ is a semi-odd triple.}
	\label{fig:odd-approx}
\end{figure}

The next lemma shows that one may upgrade a level-$2$-approximation to a small family of level-$3$-approximations which cover (at least) the same set of four-sections.

\begin{lemma}\label{lem:family-of-level-3-approx}
	For any integers $L,M \ge 0$ and any level-$2$-approximation $A$, there exists a family $\cA$ of level-$3$-approximations of size
	\[ |\cA| \le \exp\left( \tfrac{C\log^{3/2} d}{\sqrt{d}} (L/d + M) \right) \]
	such that $\FS^2(A) \cap \FS_{L,M} \subset \FS^3(\cA)$.
\end{lemma}

The proof of Lemma~\ref{lem:family-of-level-3-approx} is based on the following technical lemma whose proof is given at the end of this section.
For a semi-odd pair $(U_\ins,U_\out)$ and an integer $L\ge0$, denote
\begin{align*}
\SO(U_\ins,U_\out) &:= \big\{ U \subset \Z^d ~:~ (U_\ins,U,U_\out) \text{ is semi-odd} \big\} ,\\
\SO^\partial_L(U_\ins,U_\out) &:= \big\{ U \in \SO(U_\ins,U_\out) ~:~ |\partial U \setminus \partial(U_\ins,U_\out)| \le L \big\}.
\end{align*}
One should think of $\SO^\partial_L(U_\ins,U_\out)$ as the collection of sets approximated by $(U_\ins,U_\out)$, which have an ``unknown boundary'' of length $L$.
For a family $\cE$ of semi-odd pairs, set $\SO(\cE) := \cup_{(U_\ins,U_\out) \in \cE} \SO(U_\ins,U_\out)$.

\begin{lemma}\label{lem:approx-level2}
Let $(V_\ins,V_\out)$ be a semi-odd pair.
For every integer $L \ge 0$ there exists a family $\cE$ of $\sqrt{d}$-tight semi-odd pairs of size
\[ |\cE| \leq \exp\left(\tfrac{C\log d}{d} \left(|V_*| + \tfrac{L}{\sqrt{d}}\right)\right) \]
such that $\SO^\partial_L(V_\ins,V_\out) \subset \SO(\cE)$.
\end{lemma}

\begin{proof}[Proof of Lemma~\ref{lem:family-of-level-3-approx}]
Observe first that given $U \subset \Z^d$ and $\xi \subset \partial U$, there exists a minimal $\xi$-odd set and a minimal $\xi$-even set containing $U$, namely, $U \cup (\Odd \cap \intB^\xi U^c)$ and $U \cup (\Even \cap \intB^\xi U^c)$.

Let $L,M \ge 0$ be integers and let $A$ be a level-$2$-approximation.
Denote $\xi := \partial(A_1,A_2)$. Let $B_{13}$ and $B_{23}$ be the minimal $\xi$-odd sets containing $A_1 \cup A_3$ and $A_2 \cup A_3$, respectively, and let $B_{10}$ and $B_{20}$ be the minimal $\xi$-even sets containing $A_1 \cup A_0$ and $A_2 \cup A_0$, respectively.
Since $\FS^{2}(A) \cap \FS_{L,M} \neq \emptyset$, \refLBAsize{} implies that
\begin{equation}\label{eq:upgrading-weak-FA-1}
|(B_{13} \cup B_{20})^c|, |(B_{23} \cup B_{10})^c| \le |(A_0 \cup A_1 \cup A_2 \cup A_3)^c| \le C(L/d + M) \sqrt{d \log d} .
\end{equation}
Let $\cE_1$ and $\cE_2$ be the families of $\sqrt{d}$-tight semi-odd pairs obtained by applying Lemma~\ref{lem:approx-level2} to $(B_{13},B_{20})$ and $(B_{23},B_{10})$, respectively.
Finally, define
\[ \cA := \big\{ (A_0,A_1,A_2,A_3,A_{10},A_{23},A_{20},A_{13}) ~:~ (A_{13},A_{20}) \in \cE_1, ~ (A_{23},A_{10}) \in \cE_2 \big\} .\]
Then, by the bounds on $|\cE_1|$ and $|\cE_2|$ given in Lemma~\ref{lem:approx-level2} and by~\eqref{eq:upgrading-weak-FA-1}, we have
\[ |\cA| \le \exp\left(C \log d \cdot \left(\tfrac{C(L/d + M) \sqrt{d \log d}}{d} + \tfrac{L}{d^{3/2}}\right)\right) \le \exp\left( C (L/d + M) \tfrac{\log^{3/2} d}{\sqrt{d}} \right) .\]

We now show that $\cA$ satisfies the conclusion of the lemma.
Let $K \in \FS^{2}(A) \cap \FS_{L,M}$.
Then, by~\refLBAsubset{} and~\refLBAanomaly{}, $\xi = \partial(K_1,K_2)$. Thus, as $A_i \cup A_j \subset K_{ij}$, the minimality of $B_{ij}$ implies that $B_{ij} \subset K_{ij}$, for $i \in \{1,2\}$ and $j \in \{0,3\}$.
Therefore, Lemma~\ref{lem:K-ij-is-xi-odd} implies that $(B_{13},K_{13},B_{20})$ and $(B_{23},K_{23},B_{10})$ are semi-odd triples.
Since $(B_{13},B_{20})$ is a refinement of $(A_1 \cup A_3,A_2 \cup A_0)$, \refLBAanomaly{} implies that
\[ |\partial K_{13} \setminus \partial(B_{13},B_{20})| \le |\partial K \setminus (\partial(K_1,K_2) \cup \partial(K_0,K_3))| = |\partial K_{12}| = L .\]
Hence, we conclude that $K_{13} \in \SO^\partial_L(B_{13},B_{20})$.
Therefore, since $\SO^\partial_L(B_{13},B_{20}) \subset \SO(\cE_1)$, there exists $(A_{13},A_{20}) \in \cE_1$ such that $(A_{13},K_{13},A_{20})$ is semi-odd.
Similarly, $|\partial K_{23} \setminus \partial(B_{23},B_{10})| \le L$ so that $K_{23} \in \SO^\partial_L(B_{23},B_{10})$, and so there exists $(A_{23},A_{10}) \in \cE_2$ such that $(A_{23},K_{23},A_{10})$ is semi-odd.
Thus, we have shown that $K \in \FS^3(\cA)$, as required.
\end{proof}

The rest of the section is devoted to the proof of Lemma~\ref{lem:approx-level2}.
For a semi-odd pair $(U_\ins,U_\out)$ and an integer $n \ge 0$, we define
\[
\SO_n(U_\ins,U_\out) := \big\{ U \in \SO(U_\ins,U_\out) ~:~ |\Odd \cap U_* \cap U| + |\Even \cap U_* \cap U^c| \le n \big\}.
\]
We have
\begin{equation}\label{cl:tightness-semi-odd-triples}
\SO^\partial_L(U_\ins,U_\out) \subset \SO_{\lfloor L/(2d-t) \rfloor}(U_\ins,U_\out) .
\end{equation}
To see this, observe that since $\partial(U_\ins,U_\out) \cap \partial U_* = \emptyset$, it suffices to show that, for $U \in \SO(U_\ins,U_\out)$,
\[ |\Odd \cap U_* \cap U| + |\Even \cap U_* \cap U^c| \le \tfrac{1}{2d-t} |\partial U \cap \partial U_*| .\]
Indeed, this follows since, by~\eqref{eq:odd-t-approx} and~\eqref{eq:semi-odd-disjoint-boundary},
\begin{align*}
|\Odd \cap U_* \cap U| \cdot (2d-t) &\le |\partial(\Odd \cap U_* \cap U, \Even \setminus U_*)| = |\partial(\Odd \cap U_* \cap U, U_\out)| ,\\
|\Even \cap U_* \cap U^c| \cdot (2d-t) &\le |\partial(\Even \cap U_* \cap U^c, \Odd \setminus U_*)| = |\partial(\Even \cap U_* \cap U^c, U_\ins)| ,
\end{align*}
and $\partial(\Odd \cap U_* \cap U, U_\out)$ and $\partial(\Even \cap U_* \setminus U, U_\ins)$ are disjoint subsets of $\partial U \cap \partial U_*$.

In the following lemma, we show that for any semi-odd pair $(V_\ins,V_\out)$ there exists a small family of $t$-odd-tight semi-odd refinements which covers $\SO_n(V_\ins,V_\out)$. We then obtain in Corollary~\ref{cor:approx-algorithm} a small family of $t$-tight semi-odd refinements which covers the same collection. Finally, using this and~\eqref{cl:tightness-semi-odd-triples}, we prove Lemma~\ref{lem:approx-level2}.

For a family $\cE$ of semi-odd pairs, we define $\SO_n(\cE) := \cup_{(U_\ins,U_\out) \in \cE} \SO_n(U_\ins,U_\out)$.
A pair $(U_\ins,U_\out)$ of disjoint sets is called \emph{$t$-odd-tight} if $\Odd \cap U_* \cap N_t(U_*) = \emptyset$.

\begin{lemma}\label{lem:approx-algorithm}
Let $(V_\ins,V_\out)$ be a semi-odd pair.
For any integer $n \ge 0$ and any $1 \le t \le 2d$ there exists a family $\cE$ of $t$-odd-tight semi-odd refinements of $(V_\ins,V_\out)$ of size $|\cE| \leq \exp(C \log d \cdot n/t)$
such that $\SO_n(V_\ins,V_\out) \subset \SO_n(\cE)$.
\end{lemma}

\begin{proof}
Define $\overline{W} := \Odd \cap V_* \cap N(V_*)$. For $W \subset \overline{W}$, denote
\begin{align*}
W_\out &:= V_* \cap N(W) ,\\
W_\ins &:= \Odd \cap V_* \cap N_t(V_* \setminus W_\out) .
\end{align*}
Let $\mathcal{W}:=\{W\subset \overline{W} : |W|\le n/t\}$ and define
\[ \cE := \left\{ (V_\ins \cup W_\ins, V_\out \cup W_\out) ~:~ W \in \mathcal{W} \right\} .\]
Let us show that every pair in $\cE$ is a $t$-odd-tight semi-odd pair.
Let $W \in \mathcal{W}$ and denote $U_\ins := V_\ins \cup W_\ins$, $U_\out := V_\out \cup W_\out$ and $\xi := \partial(V_\ins,V_\out)$. Note that $\xi \subset \partial(U_\ins,U_\out)$, since $W_\ins,W_\out$ are disjoint from $V_\ins,V_\out$.
Thus, since $V_\ins$ is $\xi$-odd and $W_\ins \subset \Odd$, we see that $U_\ins$ is $\xi$-odd.
Similarly, since $V_\out$ is $\xi$-even and $W_\out \subset \Even$, we have that $U_\out$ is $\xi$-even.
Hence, $(U_\ins,U_\out)$ is semi-odd.
Moreover, by the definition of $W_\ins$,
\[ \Odd \cap U_* \cap N_t(U_*) \subset \Odd \cap (V_* \setminus W_\ins) \cap N_t(V_* \setminus W_\out) = \emptyset ,\]
so that $(U_\ins,U_\out)$ is $t$-odd-tight.

We now bound the size of $\cE$.
If $\SO_n(V_\ins,V_\out)=\emptyset$ then the lemma follows trivially by taking $\cE:=\emptyset$.
Otherwise, there exists $U \in \SO_n(V_\ins,V_\out)$. Since $(V_\ins,U,V_\out)$ is semi-odd, \eqref{eq:semi-odd-triple} yields
\[ \overline{W} \subset (\Odd \cap V_* \cap U) \cup N(\Even \cap V_* \cap U^c) ,\]
so that $|\overline{W}| \le 2dn$, and hence,
\[ |\cE| \le |\mathcal{W}| \le \sum_{k=0}^{\lfloor n/t \rfloor} \binom{2dn}{k} \le \exp\big(C (n/t) \log d \big) , \]
where the right-most inequality follows by a well-known Chernoff bound.

It remains to show that $\SO_n(V_\ins,V_\out) \subset \SO_n(\cE)$.
Let $U \in \SO_n(V_\ins,V_\out)$ and let $W$ be a maximal subset of $\overline{W} \setminus U$ among those satisfying $|W_\out| \ge t|W|$. Define $U_\ins := V_\ins \cup W_\ins$ and $U_\out := V_\out \cup W_\out$.
We shall show that $(U_\ins,U_\out) \in \cE$ and that $(U_\ins,U,U_\out)$ is a semi-odd triple.

We begin by showing that $W_\out \subset U^c$ and $W_\ins\subset U$.
Indeed, since $(V_\ins,U^c)$ is semi-odd, using~\eqref{eq:semi-odd-disjoint-boundary} we obtain
\[ W_\out \cap U \subset N(U^c) \cap U \setminus V_\ins \subset \Odd .\]
Thus, as $W_\out \subset \Even$, we have $W_\out \subset U^c$. In particular, $|W_\out| \le |\Even \cap V_* \cap U^c| \le n$. Hence, we obtain that $|W| \le n/t$ so that $W \in \mathcal{W}$ and $(U_\ins,U_\out) \in \cE$.
To see that $W_\ins \subset U$, observe that, by the definition of $W_\ins$, for any $v \in W_\ins \subset \overline{W}$, we have
\[ |(W \cup\{v\})_\out| = |N(W \cup\{v\})\cap V_*|\ge |W_\out|+|N(v)\cap V_* \setminus W_\out|\ge t(|W|+1) .\]
Thus, by the maximality of $W$, $v \in U$.
Therefore, $U_\ins \subset U$ and $U_\out \subset U^c$.
It remains to show that $(U_\ins,U,U_\out)$ is a semi-odd triple. Since $(V_\ins,U^c)$ and $(U,V_\out)$ are semi-odd, this follows from the same argument as for $(V_\ins \cup W_\ins,V_\out \cup W_\out)$ above.
\end{proof}

We say that a pair $(U_\ins,U_\out)$ of disjoint sets is \emph{$t$-even-tight} if $\Even \cap U_* \cap N_t(U_*) = \emptyset$.

\begin{cor}\label{cor:approx-algorithm-exterior}
Let $(V_\ins,V_\out)$ be a semi-odd pair.
For any integer $n \ge 0$ and any $1 \le t \le 2d$ there exists a family $\cE$ of $t$-even-tight semi-odd refinements of $(V_\ins,V_\out)$ of size $|\cE| \leq \exp(C \log d \cdot n/t)$
such that $\SO_n(V_\ins,V_\out) \subset \SO_n(\cE)$.
\end{cor}

\begin{proof}
For a semi-even pair $(U_\ins,U_\out)$ we define $\SO^{\text{even}}_n(U_\ins,U_\out)$ as the analogue of $\SO_n(U_\ins,U_\out)$ obtained by exchanging the roles of even and odd. Observe that Lemma~\ref{lem:approx-algorithm} is then applicable to semi-even pairs as well, with ``odd'' replaced by ``even''.
Thus, since $(V_\out,V_\ins)$ is semi-even, we may apply the analogue of Lemma~\ref{lem:approx-algorithm} for semi-even pairs to obtain a family $\cE'$ of $t$-even-tight semi-even refinements of $(V_\out,V_\ins)$ such that $|\cE'| \le \exp(C \log d \cdot n/t)$ and $\SO^{\text{even}}_n(V_\out,V_\ins) \subset \SO^{\text{even}}_n(\cE')$.
Noting that $U \in \SO_n(U_\ins,U_\out)$ if and only if $U^c \in \SO^{\text{even}}_n(U_\out,U_\ins)$, we see that $\cE := \{ (U_\ins,U_\out) : (U_\out,U_\ins) \in \cE' \}$ is a family of $t$-even-tight semi-odd refinements of $(V_\ins,V_\out)$ such that $\SO_n(V_\ins,V_\out) \subset \SO_n(\cE)$.
\end{proof}

\begin{cor}\label{cor:approx-algorithm}
Let $(V_\ins,V_\out)$ be a semi-odd pair.
For any integer $n \ge 0$ and any $1 \le t \le 2d$ there exists a family $\cE$ of $t$-tight semi-odd refinements of $(V_\ins,V_\out)$ of size $|\cE| \leq \exp(C \log d \cdot n/t)$
such that $\SO_n(V_\ins,V_\out) \subset \SO_n(\cE)$.
\end{cor}

\begin{proof}
Applying Corollary~\ref{cor:approx-algorithm-exterior} to $(V_\ins,V_\out)$, we obtain a family $\cE'$ of $t$-even-tight semi-odd refinements of $(V_\ins,V_\out)$ such that $|\cE'| \le \exp(C \log d \cdot n/t)$ and $\SO_n(V_\ins,V_\out) \subset \SO_n(\cE')$.
Applying Lemma~\ref{lem:approx-algorithm} to each $(U_\ins,U_\out) \in \cE'$ and noting that a $t$-odd-tight refinement of a $t$-even-tight pair is $t$-tight, we obtain a collection of families of $t$-tight semi-odd refinements of $(V_\ins,V_\out)$.
Taking the union over this collection, we obtain a family $\cE$ of $t$-tight semi-odd refinements of $(V_\ins,V_\out)$ such that $|\cE| \le |\cE'| \exp(C \log d \cdot n/t)$ and $\SO_n(V_\ins,V_\out) \subset \SO_n(\cE)$.
\end{proof}

\begin{proof}[Proof of Lemma~\ref{lem:approx-level2}]
Denote $m := |V_*|$.
Applying Corollary~\ref{cor:approx-algorithm} with $n=m$ and $t=d$ to $(V_\ins,V_\out)$, we obtain a family $\cE'$ of $d$-tight semi-odd refinements of $(V_\ins,V_\out)$ such that $|\cE'| \le \exp(C \log d \cdot m/d)$ and $\SO_m(V_\ins,V_\out) \subset \SO_m(\cE')$.
Applying Corollary~\ref{cor:approx-algorithm} with $n=\lfloor L/d \rfloor$ and $t=\sqrt{d}$ to each $(U_\ins,U_\out) \in \cE'$, we obtain a collection of families of $\sqrt{d}$-tight semi-odd refinements of $(V_\ins,V_\out)$.
Taking the union over this collection, we obtain a family $\cE$ of $\sqrt{d}$-tight semi-odd refinements of $(V_\ins,V_\out)$ such that
\[ |\cE| \le |\cE'| \exp(C \log d \cdot L/d^{3/2}) \le \exp(C \log d \cdot (m/d + L/d^{3/2})) \]
and $\SO_{\lfloor L/d \rfloor}(\cE') \subset \SO(\cE)$.

It remains to show that $\SO^\partial_L(V_\ins,V_\out) \subset \SO(\cE)$.
To this end, let $U \in \SO^\partial_L(V_\ins,V_\out)$.
Clearly, $U \in \SO_m(V_\ins,V_\out)$.
Therefore, as $\SO_m(V_\ins,V_\out) \subset \SO_m(\cE')$, there exists $(U_\ins,U_\out) \in \cE'$ such that $U \in \SO_m(U_\ins,U_\out)$. Since $(U_\ins,U_\out)$ is a refinement of $(V_\ins,V_\out)$, we have $|\partial U \setminus \partial(U_\ins,U_\out)| \le |\partial U \setminus \partial(V_\ins,V_\out)| \le L$. Thus, by~\eqref{cl:tightness-semi-odd-triples}, $U \in \SO^\partial_L(U_\ins,U_\out) \subset \SO_{\lfloor L/d \rfloor}(U_\ins,U_\out)$.
Finally, since $\SO_{\lfloor L/d \rfloor}(\cE') \subset \SO(\cE)$, we have $U \in \SO(\cE)$, as required.
\end{proof}

\subsection{Level-4-approximation.}

In this section, we construct level-4-approximations using only the information contained within the level-$3$-approximations, and then put together the previous steps to obtain the existence of a small family of level-4-approximations which approximates every four-section in $\FS_{L,M,\rho}$.

We call a family $A=(A_I)_{I \subset \{0,1,2,3\}}$ of subsets of $\Z^d$ an {\em information system} ({\em system} for short).
We say that a system $A'$ is a {\em refinement} of a system $A$ if $A_I \subset A'_I$ for all $I \subset \{0,1,2,3\}$.
We say that a system $A$ {\em respects} a four-section $K$ if $A_I \subset K_I := \cup_{l \in I} K_l$ for all $I \subset \{0,1,2,3\}$.
A system $A$ is said to be {\em exhausted} if
\begin{equation}\label{it:def-strong-FA-intersect}
A_I \cap A_J = A_{I \cap J} \quad \text{ for all } I,J \subset \{0,1,2,3\} .
\end{equation}
To lighten the notation, we shorten $A_{\{l\}}$ to $A_l$, $A_{\{l,m\}}$ to $A_{lm}$, and so forth.
Information systems are a convenient way to represent information about four-sections.
Recall our notation $\bar{l} = 3-l$.

\begin{definition}[level-$4$-approximation]
	\label{def:strong-FA}
	An exhausted system $A$ is called a {\em level-$4$-approximation} of a four-section $K$ if the following holds:
	\begin{enumerate}[\qquad(a)]
		\item \label{it:def-strong-FA-extends-weak-FA} $(A_0,A_1,A_2,A_3,A_{10},A_{23},A_{20},A_{13})$ is a level-$3$-approximation of $K$.
		\item \label{it:def-strong-FA-subset} $A$ respects $K$.
		\item \label{it:def-strong-FA-odd} $A_0$ and $A_{012}$ are even and $A_3$ and $A_{123}$ are odd.
		\item \label{it:def-strong-FA-anomaly03} $N(A_j) \subset A_{\bar{j}} \cup A_{12j}$ for every $j \in \{0,3\}$.
		\item \label{it:def-strong-FA-anomaly12} $N(A_i) \subset A_{\bar{i}} \cup A_{i03}$ for every $i \in \{1,2\}$.
	\end{enumerate}
\end{definition}

Thus, part~\eqref{it:def-strong-FA-odd} reflects the fact that $K_0$ must be even and $K_3$ must be odd, while parts~\eqref{it:def-strong-FA-anomaly03} and~\eqref{it:def-strong-FA-anomaly12} reflect the fact that singularities are guaranteed to be known.

\begin{lemma}\label{lem:algorithm-for-level-4-approx}
For any level-$3$-approximation $A'$ there exists a level-$4$-approximation $A$ such that $\FS^3(A') \subset \FS^4(A)$.
\end{lemma}
\begin{proof}
Let $A'$ be a level-$3$-approximation.
Define $A$ by $A_I := \cap_{K \in \FS^3(A')} K_I$ for $I \subset \{0,1,2,3\}$.
It is straightforward to check that $A$ is an exhausted system and that $\FS^3(A') \subset \FS^4(A)$.
\end{proof}

We now put together the previous steps to obtain the following corollary.
\begin{cor}\label{lem:family-of-level-4-approx}
	For any integers $L,M \ge 0$, there exists a family $\cA$ of level-$4$-approximations of size
	\[ |\cA| \le \exp\left( C L d^{-3/2} \log^{3/2} d + CM \log d \right) \]
	such that
	$\FS_{L,M,\rho} \subset \FS^4(\cA)$.
\end{cor}

\begin{proof}
Applying Lemma~\ref{lem:family-of-level-1-approx}, we obtain a family $\cA_1$ of level-$1$-approximations such that $\FS_{L,M,\rho} \subset \FS^1(\cA_1)$ and
\[ |\cA_1| \le \exp\left( C L d^{-3/2} \log^{3/2} d + CM \log d \right) .\]
Applying Lemma~\ref{lem:family-of-level-2-approx} to each level-$1$-approximation in $\cA_1$, we obtain a collection of families of level-$2$-approximations.
Taking the union over this collection, we obtain a family $\cA_2$ of level-$2$-approximations such that $\FS_{L,M,\rho} \subset \FS^2(\cA_2)$ and
\[ |\cA_2| \le |\cA_1| \cdot \exp\left( C L d^{-3/2} \log^{1/2} d + CM \right) .\]
Similarly, using Lemma~\ref{lem:family-of-level-3-approx}, we obtain a family $\cA_3$ of level-$3$-approximations such that $\FS_{L,M,\rho} \subset \FS^3(\cA_3)$ and
\[ |\cA_3| \le |\cA_2| \cdot \exp\left( C (L/d + M) \tfrac{\log^{3/2} d}{\sqrt{d}} \right) .\]
Once again, using Lemma~\ref{lem:algorithm-for-level-4-approx}, we obtain a family $\cA_4$ of level-$4$-approximations such that $\FS_{L,M,\rho} \subset \FS^4(\cA_4)$ and $|\cA_4| \le |\cA_3|$. Putting together the above bounds, we obtain
	\[ |\cA_4| \le \exp\left( C L d^{-3/2} \log^{3/2} d + CM \log d \right) . \]
Thus, $\cA:=\cA_4$ satisfies the conclusion of the corollary.
\end{proof}

As we have already remarked, level-4-approximations do not contain more information than level-3-approximations per se, but rather they allow us to extract and utilize this information more conveniently.
For instance, this allows us to ``break'' the regular boundary $\partial K_{12}$ into the four types of boundaries $\partial(K_i,K_j)$.
The relevant properties for this are summarized in the following lemma.

Define $\Odd_0 := \Odd$, $\Even_0 := \Even$, $\Odd_3 := \Even$ and $\Even_3 := \Odd$.
These notations will be useful in light of the parity difference between $K_0$ and $K_3$ in~\refFSodd{}.

\begin{lemma}\label{lem:level-4-properties}
	Let $A$ be a level-$4$-approximation. Then, for $i \in \{1,2\}$ and $j \in \{0,3\}$, we have
	\begin{enumerate}[\qquad(a)]
		\item \label{it:def-FA-dilemma-of-pairs} $\Even_j \subset A_{1j} \cup A_{2j} \cup A_{12 \bar{j}}$.
		\item \label{it:def-FA-boundary-neighbor} $\Even_j \subset A_{ij} \cup A_{\bar{i}\bar{j}} \cup N(A_{\bar{i}\bar{j}})$.
		\item \label{it:def-FA-unknown-expands} $\Even_j \cap N(A_{ij}) \subset A_{\bar{i}} \cup A_{ij}$.
		\item \label{it:def-FA-no-triples} $\Even_j \cap N(A_{12j}) \subset A_{1j} \cup A_{2j}$.
		\item \label{it:def-FA-no-mixing} $A_{ij} \cap N(A_{\bar{i} \bar{j}}) \subset A_i \cup A_j$.
	\end{enumerate}
\end{lemma}

\begin{proof}
Let $K \in \FS^4(A)$.
In the following proofs, we repeatedly use the facts that $A$ respects $K$ and that $A$ is an exhausted system, invoking $A_I \subset \cup_{l \in I} K_l$ for $I \subset \{0,1,2,3\}$ and~\eqref{it:def-strong-FA-intersect} without an explicit reference. Moreover, we use properties of level-$3$-approximations and level-$2$-approximations for $A$ without referring to~\refLCAextendsLBA{} and~\refLDAextendsLCA{}.
We also use the following.
By~\refLCAsqrtdTight{}, together with~\eqref{eq:semi-odd-disjoint-boundary} and~\eqref{eq:odd-t-approx},
\begin{equation}\label{eq:sqrt-d-approx-single}
\begin{aligned}
 	\Even_j \setminus (A_{ij} \cup A_{\bar{i}\bar{j}}) \subset N_{2d-\sqrt{d}}(A_{\bar{i}\bar{j}}) .
\end{aligned}
\end{equation}
In particular, applying~\eqref{eq:sqrt-d-approx-single} twice, with $i=1$ and with $i=2$, we obtain
\begin{equation}\label{eq:sqrt-d-approx-double}
\begin{aligned}
 	\Even_j \setminus (A_{10} \cup A_{20} \cup A_{13} \cup A_{23}) \subset N_{2d-2\sqrt{d}}(A_{\bar{j}}) .
\end{aligned}
\end{equation}

\smallskip
\noindent\textbf{Part~(\ref{it:def-FA-dilemma-of-pairs}).}
It suffices to show that if $v \in \Even_j$ does not belong to $A_{10} \cup A_{20} \cup A_{13} \cup A_{23}$, then $v \in A_{12 \bar{j}}$. Indeed, this follows from~\eqref{eq:sqrt-d-approx-double} and~\refLDAanomalyZeroThree{}.

\smallskip
\noindent\textbf{Part~(\ref{it:def-FA-boundary-neighbor}).}
Follows immediately from~\eqref{eq:sqrt-d-approx-single}.

\smallskip
\noindent\textbf{Part~(\ref{it:def-FA-unknown-expands}).}
Let $v \in \Even_j$ and let $u \in A_{ij}$ be adjacent to $v$.
By~\refLCAsqrtdTight{} and~\eqref{eq:semi-odd-disjoint-boundary}, we have $v \in A_{\bar{i}\bar{j}} \cup A_{ij}$.
If $v \in A_{\bar{i} \bar{j}}$ then, since $u \in A_{12j}$, \refLDAodd{} implies that $v \in A_{12j} \cap A_{\bar{i} \bar{j}} = A_{\bar{i}}$.

\smallskip
\noindent\textbf{Part~(\ref{it:def-FA-no-triples}).}
We first show that $A_{12} \cap N(A_{12}) \subset A_1 \cup A_2$.
Let $u,v \in A_{12}$ be adjacent.
By~\refLDAanomalyOneTwo{}, it suffices to show that either $u \in A_1 \cup A_2$ or $v \in A_1 \cup A_2$.
Thus, we may restrict ourselves to the case $u,v \notin A_{10} \cup A_{20} \cup A_{13} \cup A_{23}$.
Assume without loss of generality that $u$ is odd.
Then, by~\eqref{eq:sqrt-d-approx-double}, $u \in N_{2d - 2\sqrt{d}}(A_0)$ and $v \in N_{2d-2\sqrt{d}}(A_3)$.
Therefore, $u,v \in N_{2d-4\sqrt{d}}(K^{\anomaly})$, which implies that $u,v \in A_1 \cup A_2$ by~\refLBAmanyAnomalies{}.

Now, towards showing part~\eqref{it:def-FA-no-triples}, let $v \in \Even_j$ and let $u \in A_{12j}$ be adjacent to $v$. By part~\eqref{it:def-FA-dilemma-of-pairs}, $v \in A_{1j} \cup A_{2j} \cup A_{12\bar{j}}$. If $v \in A_{12\bar{j}}$ then, by~\refLDAodd{}, $v,u \in A_{12}$, and hence, by the above claim, $v \in A_1 \cup A_2$.

\smallskip
\noindent\textbf{Part~(\ref{it:def-FA-no-mixing}).}
We first show that
\begin{equation}\label{eq:strong-FA-sees-many-j}
N_d(\intB K_j) \subset N(A_j) .
\end{equation}
Let $v \in N_d(\intB K_j)$.
If $v \in N(K_j \cap K^{\anomaly})$ then $v \in N(A_j)$ by~\refLBAanomaly{}, and thus, we may assume that $v \in N_d(\intB K_j \setminus K^{\anomaly})$.
If $v \in N_{d/2}(K_j \cap K^{\rev})$ then, by~\refLBAseesManyi{}, $v \in N(A_j)$.
Otherwise, $v \in N_{d/2}(\intB K_j \setminus (K^{\rev} \cup K^{\anomaly}))$ so that, by~\eqref{eq:sqrt-d-approx-single},
\[ \intB K_j \setminus (K^{\rev} \cup K^{\anomaly}) \subset \intB K_j \setminus N_{2d-\sqrt{d}}(A_{\bar{i}\bar{j}}) \subset (\Even_j \setminus A_{\bar{i}\bar{j}}) \setminus N_{2d-\sqrt{d}}(A_{\bar{i}\bar{j}}) \subset A_{ij} .\]
Thus, $v \in N_{d/2}(A_{1j} \cap A_{2j}) = N_{d/2}(A_j) \subset N(A_j)$.

We now show part~\eqref{it:def-FA-no-mixing}.
Let $v \in A_{ij}$ and let $u \in A_{\bar{i}\bar{j}}$ be adjacent to $v$.
Assume towards a contradiction that $v \notin A_i \cup A_j$.
Then, by~\refLDAanomalyZeroThree{} and~\refLDAanomalyOneTwo{}, $u \notin A_{\bar{i}} \cup A_{\bar{j}}$.
Since $v \in A_{ij} \setminus (A_i \cup A_j)$ and $u \in A_{\bar{i}\bar{j}} \setminus (A_{\bar{i}} \cup A_{\bar{j}})$, we must have $v,u \notin A_{i\bar{j}} \cup A_{\bar{i}j}$.
Thus, observing that, by~\refLDAodd{}, $v \in \Even_j$ and $u \in \Odd_j$, \eqref{eq:sqrt-d-approx-single} implies that $v \in N_{2d-\sqrt{d}}(A_{i\bar{j}})$ and $u \in N_{2d-\sqrt{d}}(A_{\bar{i}j})$.

Let $F$ denote the set of pairs $(w,z)$ such that $(v,u,w,z)$ is a four-cycle, $w \in A_{\bar{i}j}$ and $z \in A_{i\bar{j}}$, and note that $|F| \ge 2d-2\sqrt{d}$.
We partition $F$ into $F^* \cup F'$ according to whether the edge $\{w,z\}$ is on the singular or regular boundary, respectively.
If $|F^*| \ge d-2\sqrt{d}$ then, by~\refLBAmanyAnomalies{}, $v \in A_i \cup A_j$ and $u \in A_{\bar{i}} \cup A_{\bar{j}}$, which leads to a contradiction. Thus, we may assume that $|F'| \ge d$. Using that $v,u \notin K^{\anomaly}$ by~\refLBAanomaly{}, one may check that $F'$ consists either entirely of edges in $\partial(K_i,K_j)$ or entirely of edges in $\partial(K_{\bar{i}},K_{\bar{j}})$.
In the former case, $u \in N_d(\intB K_j) \subset N(A_j)$ by~\eqref{eq:strong-FA-sees-many-j}, in which case, $u \in A_{12j} \cap A_{\bar{i}\bar{j}} = A_{\bar{i}}$ by~\refLDAanomalyZeroThree{}.
Similarly, in the latter case, $v \in N_d(\intB K_{\bar{j}}) \subset N(A_{\bar{j}})$, in which case, $v \in A_{12\bar{j}} \cap A_{ij} = A_i$.
\end{proof}

\subsection{Four-approximations.}

Here, we use level-4-approximations to construct four-approximations explicitly.
Lemma~\ref{lem:family-of-FA} is an immediate consequence of Corollary~\ref{lem:family-of-level-4-approx} and the following lemma.

\begin{lemma}\label{lem:level-4-approx-gives-FA}
For any level-$4$-approximation $A$ there exists a four-approximation $A'$ satisfying that $\FS^4(A) \subset \FS(A')$.
\end{lemma}
\begin{proof}
We show that $A':=\big(A_0,A_1,A_2,A_3,(\DilemmaOdd ij,\DilemmaEven ij)_{i,j}\big)$ satisfies the conclusion of the lemma, where
\begin{align*}
\Unknown ij &:= (A_{i \bar{j}} \cup A_{\bar{i} j})^c ,\\
\DilemmaOdd ij &:= \Odd_j \cap \Unknown ij \cap A_{12j} ,\\
\DilemmaEven ij &:= \Even_j \cap \Unknown ij \cap A_{ij}.
\end{align*}
To this end, let $K \in \FS^4(A)$. We must show that $A'$ is a four-approximation of $K$. As before, we shall repeatedly use the facts that $A$ respects $K$ and that $A$ is an exhausted system, without an explicit reference.

\smallskip\noindent{\bf\refFA{1}.}
Clearly, $\{ A_0,A_1,A_2,A_3,\cup_{i,j} (\DilemmaOdd ij \cup \DilemmaEven ij) \}$ are pairwise disjoint, and moreover, it is straightforward to check (using exhaustion and parity) that $\{ A_i,A_j,\DilemmaOdd 1j \cup \DilemmaOdd 2j, \DilemmaEven ij \}_{i,j}$ are pairwise disjoint.
By Lemma~\ref{lem:level-4-properties}\ref{it:def-FA-dilemma-of-pairs}, $\{ \DilemmaOdd {\bar{i}}{\bar{j}}, \DilemmaEven ij \}$ is a partition of $\Even_j \cap \Unknown ij$. Thus, if $v \in \Even_j \setminus \cup_{i,j} (\DilemmaOdd ij \cup \DilemmaEven ij)$ then $v \notin \Unknown ij \cup \Unknown {\bar{i}}j$ so that $v \in (A_{i \bar{j}} \cup A_{\bar{i} j}) \cap (A_{\bar{i} \bar{j}} \cup A_{i j}) = A_0 \cup A_1 \cup A_2 \cup A_3$.

\noindent{\bf\refFA{2}.} Immediate from the definition.

\noindent{\bf\refFA{13}.} This precisely property~\refLCAsqrtdTight{}.

\noindent{\bf\refFA{9}.} Straightforward.

\noindent{\bf\refFA{11}.}
	Let $v \in \DilemmaOdd ij \cap K_i$ and let $u \in N(v) \setminus K_j$. We must show that $u \in \DilemmaEven ij$.
	Note that $u \in K_i$, since $u \in K_{\bar{i}}$ would imply that $v \in A_i$ by~\refLBAanomaly{}, and $u \in K_{\bar{j}}$ would imply that $v \in K_{\bar{j}}$ by~\refFSodd{}.
	Since $u \notin A_{\bar{i}j}$ and $v \notin A_{\bar{i}} \cup A_{i\bar{j}}$, Lemma~\ref{lem:level-4-properties}\ref{it:def-FA-unknown-expands} implies that $u \in \Unknown ij$.
	Finally, by Lemma~\ref{lem:level-4-properties}\ref{it:def-FA-no-triples}, $u \in (A_{1j} \cup A_{2j}) \cap K_i \subset A_{ij}$.

Next, let $u \in \DilemmaEven ij \cap K_j$ and let $v \in N(u) \setminus K_i$. We must show that $v \in \DilemmaOdd ij$.
	By~\refLBAanomaly{}, $v \notin K_{\bar{j}}$.
	Since $v \notin A_{i\bar{j}}$ and $u \notin A_i \cup A_{\bar{i}j}$, Lemma~\ref{lem:level-4-properties}\ref{it:def-FA-unknown-expands} implies that $v \in \Unknown ij$.
	It remains to show that $v \in A_{12j}$.
	Indeed, as $v \in \Even_{\bar{j}} \setminus A_{i\bar{j}}$ and $u \in A_{ij}$, this follows from Lemma~\ref{lem:level-4-properties}\ref{it:def-FA-dilemma-of-pairs} and Lemma~\ref{lem:level-4-properties}\ref{it:def-FA-no-mixing}.

\noindent{\bf\refFA{12}.}
By~\refFA{11} and~\refFSisolated{}, it suffices to show that
	\[ \DilemmaEven ij \cap N(\DilemmaOdd ij \setminus K_i) \subset K_j \quad\text{ and }\quad \DilemmaOdd ij \cap N(\DilemmaEven ij \setminus K_j) \subset K_i .\]

	Let $u \in \DilemmaEven ij$ and let $v \in \DilemmaOdd ij \setminus K_i$ be adjacent to $u$. We must show that $u \in K_j$.
	Indeed, since $u \in A_{ij} \setminus A_i$ and $v \in A_{12j} \setminus K_i \subset K_{\bar{i}j}$, this follows from~\refLBAanomaly{} and~\refFSodd{}.

	Next, let $v \in \DilemmaOdd ij$ and let $u \in \DilemmaEven ij \setminus K_j$ be adjacent to $v$. We must show that $v \in K_i$.
	Indeed, since $u \in (A_{ij} \setminus A_i) \setminus K_j \subset K_i \setminus A_i$ and $v \in A_{12j}$, this follows from~\refLBAanomaly{} and~\refFSodd{}.
\end{proof}

\section{Pattern violations and infinite-volume Gibbs measures}\label{sec:pattern-violations+gibbs}

Recall that $B(f,\rho)$ is the connected component of $\rho$ in $T(f)^+$, where 
$T(f)$ is the set of vertices which violate the even-$0$ pattern.
We would like to use our results about the unlikeliness of large breakups to derive that $B(f,v)$ is typically not large. Unfortunately, the breakup $K(f,v)$ does not necessarily capture the deviation described by $B(f,v)$ (consider for example a coloring $f$ which equals $0$ on a large ball around $v$). To overcome this, we extend the breakup. The procedure by which this is done resembles the original procedure used to define the breakup, an important difference however being that we do not co-connect the $0$-phase. This definition will also be important in the proof of Theorem~\ref{thm:convergence-of-finite-volume-measure}.

Let $(\Lambda,\tau)$ be even-0 boundary conditions and let $f \in \cC_\Lambda^\tau$.
Let us now define $K(f,V)$, the \emph{breakup of $f$ around a set $V \subset T(f)$}. The restriction on $V$ is in order to avoid dealing with vertices whose breakup is trivial.
Define $\kappa$ as in~\eqref{eq:def-K-coloring}.
Let $K'_0$ be the complement of the union of connected components of $\Z^d \setminus \kappa^{-1}(0)$ which intersect $V$ (this is an analogue of the co-connected closure, taken with respect to a set).
We now proceed precisely as in the definition of the breakup.
Let $K'_3$ be the co-connected closure of $\kappa^{-1}(3) \setminus K'_0$ with respect to infinity, let $K'_2$ be the co-connected closure of $\kappa^{-1}(2) \setminus (K'_0 \cup K'_3)$ with respect to infinity and let $K'_1$ be the co-connected closure of $\kappa^{-1}(1) \setminus (K'_0 \cup K'_3 \cup K'_2)$ with respect to infinity.
Finally, define $K(f,V) := (K_0,K_1,K_2,K_3)$, where $K_1 := K'_1$, $K_2 := K'_2 \setminus K'_1$, $K_3 := K'_3 \setminus (K'_1 \cup K'_2)$ and $K_0 := K'_0 \setminus (K'_1 \cup K'_2 \cup K'_3)$.

The following lemma may be seen as an extension of Lemma~\ref{lem:K-is-a-four-section}. The proof is similar to that of Lemma~\ref{lem:K-is-a-four-section} and we omit the details.
For a finite set $V \subset \Z^d$, denote
\[ \FS_V := \left\{ K\text{ four-section} ~:~ K_{123} \cap V \neq \emptyset ~\text{ and }~ \substack{\text{every connected component $U$ of $K_{123}^+$ satisfies:}\\\text{$U$ is finite, $U$ intersects $V$, $U \cap \intB K$ is connected}} \right\}.\]
Observe that, for a singleton $\{v\}$, this definition coincides with that of $\FS_v$ from Section~\ref{sec:breakup+high-level-proof}.

\begin{lemma}\label{lem:join-of-breakups}
	Let $(\Lambda,\tau)$ be even-0 boundary conditions,	let $f \in \cC_\Lambda^\tau$ and let $V \subset T(f)$. Then $K(f,V) \in \FS(f) \cap \FS_V$ and $V \subset K_{123}(f,V)$.
\end{lemma}

Before proving Theorem~\ref{thm:main2}, we require the following lemma. We extend the definition of diameter to disconnected sets $U$ by defining $\diam U := \sum_{i=1}^k \diam U_i$, where $\{U_i\}_{i=1}^k$ are the connected components of $U$. We also define $\diam^+ U := \diam U + 3k$. Observe that $\diam^+ U > \diam U^+$.

\begin{lemma}\label{lem:minimal-breakup-by-diameter}
	Let $d$ be sufficiently large, let $L,M \ge 0$ be integers, let $K \in \FS_{L,M}$ and denote $r := \diam^+ K_{123}$. Then either $L \ge cd^2r$ or $M \ge cdr$.
\end{lemma}
\begin{proof}
	Denote by $s$ the number of isolated vertices in $K_{123}$ and note that $M \ge 2ds$ by~\refFSisolated{}. Assume that $M \le dr/3$ so that $s \le r/6$. Let $A_1,\dots,A_m$ denote the connected components of $K_{123}$ which are not singletons. Then $r = 3s+3m+ \sum_{i=1}^m \diam A_i \le r/2 + 4 \sum_{i=1}^m \diam A_i$. Thus, since $A_1,\dots,A_m$ are odd by~\refFSodd{}, Lemma~\ref{lem:boundary-size-via-diameter} implies that
	\[ L+M \ge |\partial K_{123}| \ge \sum_{i=1}^m |\partial A_i| \ge (d-1)^2 \sum_{i=1}^m \diam A_i \ge \tfrac{1}{8} (d-1)^2 r .\]
	Since $M \le dr/3$, we have that $L \ge cd^2 r$.
\end{proof}

\subsection{Proof of Theorem~\ref{thm:main2}}
We assume throughout the proof that $d$ is large enough for our arguments to hold. 
Let $(\Lambda,\tau)$ be even-0 boundary conditions and let $f \sim \mu^\tau_{\Lambda,\beta}$.
Define $V(f) := T(f) \cap B(f,\rho)$ and $\bar{K}(f) := K(f,V(f))$.
Since $B(f,\rho)$ has size at least $d$ when it is non-empty, it suffices to prove the theorem for $\rho$ odd.
Thus, by~\eqref{eq:origin-non-trivial-four-section}, $\rho \in V(f)$ whenever $B(f,\rho) \neq \emptyset$. Hence, since $V(f)^+=B(f,\rho)$, Lemma~\ref{lem:join-of-breakups} implies that
\begin{equation}\label{eq:extended-breakup-is-adapted-four-section}
B(f,\rho) \neq \emptyset \quad\implies\quad \bar{K}(f) \in \FS(f) \cap \FS_\rho .
\end{equation}
Moreover, by~\eqref{eq:K-0-does-not-contain-rho} and~\refFSodd{}, we have
\begin{equation}\label{eq:B-subset-of-extended-breakup}
B(f,\rho) \subset \bar{K}_{123}(f) \cup \bar{K}^{\anomaly}(f) .
\end{equation}

We first bound the probability that $B(f,\rho)$ has size at least $r$. Since $B(f,\rho)$ has size at least $d$ when it is non-empty, we may assume that $r \ge d$.
By Lemma~\ref{lem:isoperimetry}, \eqref{eq:prob-of-breakup-with-LM-or-more}, \eqref{eq:extended-breakup-is-adapted-four-section} and~\eqref{eq:B-subset-of-extended-breakup},
\begin{align*}
	\Pr(|B(f,\rho)| \ge r)
	  &\le \Pr(|\bar{K}_{123}(f)| \ge r/2 ~\text{ or }~ |\bar{K}^{\anomaly}(f)| \ge r/2) \\
	  &\le \Pr(|\partial \bar{K}_{123}(f)| \ge 2d \cdot (r/2)^{1-1/d} ~\text{ or }~ |\partial \bar{K}(f) \setminus \partial \bar{K}_{12}(f)| \ge r/4) \\
	  &\le \Pr(\bar{K}(f) \in \FS_{\ge d (r/2)^{1-1/d}, \ge 0,\rho}) + \Pr(\bar{K}(f) \in \FS_{\ge 0, \ge (1/4) r^{1-1/d},\rho}) \\
	  &\le Cd \cdot e^{-c'r^{1-1/d}} + Cd \cdot e^{-(\beta'/4) r^{1-1/d}} \le e^{-cr^{1-1/d}} .
\end{align*}

Next, we bound the probability that the diameter of $B(f,\rho)$ is at least $r$.
By~\eqref{eq:B-subset-of-extended-breakup}, we have $B(f,\rho) \subset \bar{K}_{123}(f)^+$, and since $B(f,\rho)$ is connected, we have $\diam B(f,\rho) \le \diam \bar{K}_{123}(f)^+$.
Therefore, by Lemma~\ref{lem:minimal-breakup-by-diameter}, \eqref{eq:prob-of-breakup-with-LM-or-more} and~\eqref{eq:extended-breakup-is-adapted-four-section},
\[ \Pr(\diam B(f,\rho) \ge r) \le \Pr(\bar{K}(f) \in \FS_{\ge cd^2 r, \ge 0,\rho}) + \Pr(\bar{K}(f) \in \FS_{\ge 0, \ge cdr,\rho}) \le e^{-cdr} . \qedhere \]

\subsection{Infinite-volume Gibbs measures}
\label{sec:gibbs-proof}

In this section, we prove Theorem~\ref{thm:convergence-of-finite-volume-measure}. We first state two technical lemmas (whose proofs are postponed) which we require for the proof of the theorem.

For a set $U \subset \Z^d$ and an integer $r \ge 0$, denote
\[ U^{+r} := \{ u \in \Z^d : \dist(u,U) \le r \} .\]
For a distribution $\mu$ on $\cC_{\Z^d}$, we denote by $\mu|_U$ the marginal distribution of $\mu$ on $U$. Given two discrete distributions $\mu$ and $\lambda$ on a common space, we denote the total-variation distance between $\mu$ and $\lambda$ by $\distTV(\mu,\lambda) := \max_{A} |\mu(A)-\lambda(A)|$.

\begin{lemma}\label{lem:convergence}
	Let $d$ be sufficiently large and let $\beta \ge C \log d$.
	Let $(\Lambda,\tau)$ and $(\Lambda',\tau')$ be two even-0 boundary conditions.
	Let $r \ge 1$ and let $U$ be an odd domain such that $U^{+r} \subset \Lambda \cap \Lambda'$.
	Then
	\[ \distTV\big(\mu_{\Lambda,\beta}^\tau|_U, \mu_{\Lambda',\beta}^{\tau'}|_U\big) \le |U| \cdot e^{-cdr} .\]
\end{lemma}

We say that two random variables $X$ and $Y$ are $\epsilon$-almost independent if the covariance between any events of the form $\{X \in A\}$ and $\{Y \in B\}$ is at most $\epsilon$.

\begin{lemma}\label{lem:almost-independence-of-colorings}
	Let $d$ be sufficiently large and let $\beta \ge C \log d$.
	Let $(\Lambda,\tau)$ be even-0 boundary conditions and let $f \sim \mu^\tau_{\Lambda,\beta}$.
	Let $V \subset \Lambda$ be a domain, let $r \ge 1$ and let $U \subset \Lambda$ be such that $U^{+2r} \subset V$.
Then $f|_U$ and $f|_{V^c}$ are $\epsilon$-almost independent, where $\epsilon := |U| \cdot e^{-cdr}$.
\end{lemma}

\begin{proof}[Proof of Theorem~\ref{thm:convergence-of-finite-volume-measure}]
{\bf Convergence.}
	Let $U$	be an odd domain. We must show that the measures $\mu_{\Lambda_n,\beta}^{\tau_n}|_U$ converge as $n \to \infty$. Indeed, since $\Lambda_n$ increases to $\Z^d$, we have $\dist(U,\Lambda_n^c) \to \infty$ as $n \to \infty$. Thus, by Lemma~\ref{lem:convergence}, the sequence of measures $(\mu_{\Lambda_n,\beta}^{\tau_n}|_U)_{n=1}^{\infty}$ is a Cauchy sequence with respect to the total-variation metric, and therefore, converges.
	Since this holds for any odd domain $U$, we have established the convergence of $\mu_{\Lambda_n,\beta}^{\tau_n}$ as $n \to \infty$ towards an infinite-volume measure $\Pr = \mu_{\Z^d,\beta}^{0,0}$.

\smallskip
\noindent{\bf Strong mixing.}
Recall that the result on convergence implies that the limiting measure $\Pr$ is invariant with respect to parity-preserving automorphisms.
Lemma~\ref{lem:almost-independence-of-colorings} immediately implies that $\Pr$ is strongly mixing with respect to any parity-preserving translation $T$. Indeed, since $\Pr$ is $T$-invariant, \eqref{eq:mix} holds for any two cylinder events $A$ and $B$ (with an exponential rate of convergence).

\smallskip
\noindent{\bf Extremality.}
It is well-known (see, e.g., \cite[Proposition~7.7]{georgii2011gibbs}) that extremality (within the set of all Gibbs measures) is equivalent to the fact that the tail $\sigma$-algebra is trivial in the sense that any tail event has probability zero or one.
This in turn follows for $\Pr$ directly from Lemma~\ref{lem:almost-independence-of-colorings} (see, e.g., \cite[Proposition~7.9]{georgii2011gibbs}).

For completeness, we include a short proof that the tail $\sigma$-algebra of $\Pr$ is indeed trivial.
Let $A$ be a tail event for $\Pr$. The fact that $A$ is measurable implies that it can be approximated by cylinder events. Let $E$ be a cylinder event depending only on the values of $f$ on some finite set $U \subset \Z^d$ such that $\Pr(A \Delta E) \le \epsilon$.
Let $V$ be a domain containing $U^{+r}$.
Since $A$ is a tail event, it is measurable with respect to the $\sigma$-algebra generated by the values of $f$ on $\Z^d \setminus V$. Thus, there exists a domain $W$ and another cylinder event $F$ depending only on $f|_{W \setminus V}$ such that $\Pr(A \Delta F) \le \epsilon$. Taking $r$ large enough and applying Lemma~\ref{lem:almost-independence-of-colorings} (using the definition of $\Pr$ as a limit), we get that $f|_U$ and $f|_{W \setminus V}$ are $\epsilon$-almost independent. In particular, $|\text{Cov}(E,F)| \le \epsilon$.
Since $\Pr(A \Delta E) \le \epsilon$ and $\Pr(A \Delta F) \le \epsilon$, it follows that $|\Pr(A)-\Pr(E \cap F)|\le 2\epsilon$ and $|\Pr(A)^2-\Pr(E)\Pr(F)|\le 2\epsilon$.
Thus,
\[ |\Pr(A)-\Pr(A)^2| \le |\Pr(A)-\Pr(E \cap F)|+|\Pr(E \cap F) - \Pr(E)\Pr(F)|+|\Pr(E)\Pr(F)-\Pr(A)^2| \le 5 \epsilon .\]
Since this holds for any $\epsilon>0$, we see that $\Pr(A)=\Pr(A)^2$, which implies that $\Pr(A) \in \{0,1\}$.
\end{proof}

The proofs of Lemma~\ref{lem:convergence} and Lemma~\ref{lem:almost-independence-of-colorings} make use of the following fact which exploits the domain Markov property of the model.
We say that a collection $\cS$ of proper subsets of $\Z^d$ is a \emph{boundary semi-lattice} if for any $S_1,S_2 \in \cS$ there exists $S \in \cS$ such that $S_1 \cup S_2 \subset S$ and $\partial S \subset \partial S_1 \cup \partial S_2$.
Two boundary semi-lattices which we require are $\cS(U,V) := \{ S \subsetneq \Z^d : U \subset S \subset V \text{ and $S$ is odd} \}$ and $\cS(f) := \{ S \subsetneq \Z^d : f|_{\extB S} \equiv 0 \}$.
The latter has the property that if $\cS$ is any boundary semi-lattice, then $\cS \cap \cS(f)$ is also a boundary semi-lattice.

Recall that a domain is a finite, non-empty, connected and co-connected subset of $\Z^d$.

\begin{lemma}\label{lem:marginal-distribution-given-agreement}
	Let $\Lambda,\Lambda' \subset \Z^d$ be finite and let $U \subset V \subset \Lambda \cap \Lambda'$ be non-empty.
	Let $\tau$ and $\tau'$ be two arbitrary colorings, let $\beta>0$ and let $f \sim \mu_{\Lambda,\beta}^\tau$ and $f' \sim \mu_{\Lambda',\beta}^{\tau'}$ be independent.
	\begin{enumerate}[\qquad(a)]
		\item \label{it:marginal-distribution-given-agreement-pair} $\distTV(\mu_{\Lambda,\beta}^\tau|_U,\mu_{\Lambda',\beta}^{\tau'}|_U) \le \Pr(\cS(U,V) \cap \cS(f) \cap \cS(f')=\emptyset)$.
		\item \label{it:marginal-distribution-given-agreement-single} Assume that $U$ is connected, $V$ is co-connected and $\Pr(\cS(U,V) \cap \cS(f) \neq \emptyset)>0$. Then, conditioned on $\{\cS(U,V) \cap \cS(f) \neq \emptyset \}$, the distribution of $f|_U$ is a convex combination of the measures $\{ \mu_{S,\beta}^0|_U \}_{S \in \cS^{\text{dom}}(U,V)}$, where $\cS^{\text{dom}}(U,V)$ is the collection of domains in $\cS(U,V)$.
	\end{enumerate}
\end{lemma}

\begin{proof}
	We shall prove both items together.
	To this end, let $f''$ be either $f$ or $f'$, and denote $\cS := \cS(U,V) \cap \cS(f) \cap \cS(f'')$.
	Since $\cS$ is a finite boundary semi-lattice, it has a unique maximal element $\sf S$ (if  $\cS = \emptyset$ we set $\sf S = \emptyset$).
	Let $S \neq \emptyset$ be such that $\Pr({\sf S}=S)>0$.
	Observe that the event $\{{\sf S}=S\}$ is determined by $f|_{S^c}$ and $f''|_{S^c}$.
	Therefore, by the domain Markov property, conditioned on $\{{\sf S}=S\}$, $f|_S$ and $f''|_S$ are distributed according to $\mu_{S,\beta}^0|_S$. In particular, conditioned on $\{\cS \neq \emptyset \}$, the distribution of both $f|_U$ and $f''|_U$ is $\sum_S \Pr({\sf S}=S \mid \cS \neq \emptyset) \mu_{S,\beta}^0|_U$, from which the first item follows. Moreover, if $U$ is connected and $V$ is co-connected, then $\sf S$ is always a domain, since Lemma~\ref{lem:co-connect-properties}\ref{it:co-connect-reduces-boundary} and Lemma~\ref{lem:co-connect-properties}\ref{it:co-connect-kills-components} imply that the co-connected closure of $S$ (with respect to infinity) belongs to $\cS$ for any $S \in \cS$.
\end{proof}

In order to apply the above lemma for our purposes, we need to extend our result on the unlikeliness of pattern violations to pairs of colorings.
Given two colorings $f$ and $f'$, we define $B(f,f',\rho)$ to be the connected component of $\rho$ in $(\Even \cap T(f))^+ \cup (\Even \cap T(f'))^+$.

\begin{lemma}\label{lem:prob-of-joint-breakup-core}
	Let $d$ be sufficiently large and let $\beta \ge C \log d$.
	Let $(\Lambda,\tau)$ and $(\Lambda',\tau')$ be two even-0 boundary conditions and let $f \sim \mu^{\tau}_{\Lambda,\beta}$ and $f' \sim \mu^{\tau'}_{\Lambda',\beta}$ be independent.
	Then, for any $r \ge 1$ and any vertex $u \in \Z^d$,
	\[ \Pr\big(\diam B(f,f',u) \ge r \big) \le e^{-cdr} .\]
\end{lemma}

Before proving Lemma~\ref{lem:prob-of-joint-breakup-core}, we give the proofs of Lemma~\ref{lem:convergence} and Lemma~\ref{lem:almost-independence-of-colorings}.

\begin{proof}[Proof of Lemma~\ref{lem:convergence}]
	Let $f \sim \mu^{\tau}_{\Lambda,\beta}$ and $f' \sim \mu^{\tau'}_{\Lambda',\beta}$ be independent.
	Let $\mathcal E$ be the event that $B(f,f',u)$ intersects $\Z^d \setminus U^{+r}$ for some $u \in U$.
	Denote $S := \cup_{u \in U} B(f,f',u) \cup U$. Observe that, by definition, $S$ is odd and $f|_{\extB S} \equiv f'|_{\extB S} \equiv 0$. Moreover, on the complement of $\mathcal E$, $S \subset U^{+r}$.
	Thus, by Lemma~\ref{lem:marginal-distribution-given-agreement}\ref{it:marginal-distribution-given-agreement-pair} and Lemma~\ref{lem:prob-of-joint-breakup-core},
	\[ \distTV\big(\mu_{\Lambda,\beta}^\tau|_U, \mu_{\Lambda',\beta}^{\tau'}|_U\big) \le \Pr(\mathcal E) \le \sum_{u \in U} \Pr\big(\diam B(f,f',u) \ge r\big) \le |U| \cdot e^{-cdr} . \qedhere \]
\end{proof}

\begin{proof}[Proof of Lemma~\ref{lem:almost-independence-of-colorings}]
We begin with a simple observation about $\epsilon$-almost independence.
Let $X$ and $Y$ be discrete random variables. Let $\mu_X$ denote the distribution of $X$ and let $\mu_{X|Y}$ denote the conditional distribution of $X$ given $Y$. Assume that $\E[\distTV(\mu_{X|Y},\mu_X)] \le \epsilon$.
Then $X$ and $Y$ are $\epsilon$-almost independent. Indeed, this follows immediately from the fact that
\[ \text{Cov}(X \in A, Y \in B) = \E\big[(\mu_{X|Y}(A) - \mu_X(A)) \1_{\{Y \in B\}} \big] .\]

Let $\mu := \mu^f_{V,\beta}$ be the conditional distribution of $f$ given $f|_{V^c}$. Let $\cE'$ be the event that there exists an odd set $S$ such that $U^{+r} \subset S \subset V$ and $f|_{\extB S} \equiv 0$.
By Lemma~\ref{lem:marginal-distribution-given-agreement}\ref{it:marginal-distribution-given-agreement-single}, conditioned on $\cE'$, $\mu|_U$ is a convex combination of measures $\mu^0_{S,\beta}|_U$, where $S$ is an odd domain containing $U^{+r}$.
For any such $S$, by Lemma~\ref{lem:convergence}, we have \[ \distTV(\mu^0_{S,\beta}|_U,\mu^\tau_{\Lambda,\beta}|_U) \le |U| \cdot e^{-cdr} .\]
Let $\cE$ be the event that $B(f,u)$ intersects $V^c$ for some $u \in U^{+r}$, and observe that $\cE^c \subset \cE'$.
Hence,
\[ \E[\distTV(\mu|_U,\mu^\tau_{\Lambda,\beta}|_U)] \le |U| \cdot e^{-cdr} + \E[\mu(\cE)] = |U| \cdot e^{-cdr} + \Pr(\cE) .\]
By Theorem~\ref{thm:main2},
\[ \Pr(\cE) \le |U^{+r}| \cdot e^{-c'dr} \le |U| \cdot (Cd)^r \cdot e^{-c'dr} \le |U| \cdot e^{-cdr} .\]
Thus, $\E[\distTV(\mu|_U,\mu^\tau_{\Lambda,\beta}|_U)] \le |U| \cdot e^{-cdr}$, and the lemma follows from the above observation.
\end{proof}

The rest of this section is devoted to the proof of Lemma~\ref{lem:prob-of-joint-breakup-core}.
Let us first prove the following generalization of Lemma~\ref{lem:family-of-level-1-approx}.
For $L,M \ge 0$ and a finite $V \subset \Z^d$, denote $\FS_{L,M,V} := \FS_{L,M} \cap \FS_V$.

\begin{lemma}\label{lem:family-of-level-1-approx-many-points}
	For any integers $L,M \ge 0$ and any finite set $V \subset \Z^d$, there exists a family $\cA$ of level-$1$-approximations of size
	\[ |\cA| \le 2^{|V|} \cdot \exp\left( C L d^{-3/2} \log^{3/2} d + CM \log d \right) \]
	such that $\FS_{L,M,V} \subset \FS^1(\cA)$.
\end{lemma}
\begin{proof}
	Observe that for every $K \in \FS_V$ there exists $S \subset V$ and $(K^v)_{v \in S} \in \prod_{v \in S} \FS_v$ such that $\{ (K^v_{123})^+ \}_{v \in S}$ are pairwise disjoint, $K_0 = \cap_{v \in S} K_0^v$ and $K_i = \cup_{v \in S} K_i^v$ for $i \in \{1,2,3\}$. Observe also that for such a collection of four-sections, $\intB K = \uplus_{v \in S} \intB K^v$, and thus, if $\{(U_v,W_v)\}_{v \in S}$ are level-1-approximations of $\{K^v\}_{v \in S}$, then $(\cup_{v \in S} U_v, \cup_{v \in S} W_v)$ is a level-1-approximation of $K$.
	For integers $\ell,m \ge 0$ and $v \in V$, let $\cA_{\ell,m,v}$ be a family of level-1-approximations obtained from Lemma~\ref{lem:family-of-level-1-approx} such that $\FS_{\ell,m,v} \subset \FS^1(\cA_{\ell,m,v})$. Let $s \le |V|$ and denote
	\[ \cA := \bigcup_{\substack{S \subset V\\|S| \le s}} \bigcup_{\substack{\sum_{v \in S} \ell_v \le L\\\sum_{v \in S} m_v \le M}} \left\{ \left(\bigcup_{v \in S} U_v, \bigcup_{v \in S} W_v\right) : (U_v,W_v)_{v \in S} \in \prod_{v \in S} \cA_{\ell_v,m_v,v} \right\} .\]	
	Then $\FS_{L,M,V} \subset \FS^1(\cA)$ for $s=|V|$. In fact, applying~\eqref{eq:L+M-is-large} to each pair $(\ell_v,m_v)$, we see that this also holds for $s = \lfloor 2L/d^2 + 2M/d \rfloor$. Therefore, since
	\[ |\cA| \le \binom{|V|}{\le s} \binom{L+s}{s} \binom{M+s}{s} \exp\left( C L d^{-3/2} \log^{3/2} d + CM \log d \right) ,\]
	the stated bound on $|\cA|$ easily follows.
\end{proof}

\begin{lemma}\label{lem:prob-of-breakup-many-points}
	Let $d$ be sufficiently large and let $\beta \ge C \log d$.
	Let $(\Lambda,\tau)$ be even-0 boundary conditions and let $f \sim \mu^\tau_{\Lambda,\beta}$.
	Then, for any integers $L,M \ge 0$ and any finite $V \subset \Z^d$, we have
	\[ \Pr\big( \FS(f) \cap \FS_{L,M,V} \neq \emptyset \big) \le 2^{|V|} \cdot \exp(- c L/d - \beta M/2) .\]
	Moreover, for any $r \ge 1$, we have
	\[ \Pr\big(\FS(f) \cap \FS_V \cap \FS_{\diam \ge r} \neq \emptyset \big) \le 2^{|V|} \cdot e^{-cdr} ,\]
	where $\FS_{\diam \ge r}$ is the collection of four-sections $K$ such that $\diam^+ K_{123} \ge r$.
\end{lemma}
\begin{proof}
	Repeating the proof of Corollary~\ref{lem:family-of-level-4-approx}, using Lemma~\ref{lem:family-of-level-1-approx-many-points} in place of Lemma~\ref{lem:family-of-level-1-approx}, and using Lemma~\ref{lem:level-4-approx-gives-FA}, we obtain a family $\cA$ of four-approximations satisfying $\FS_{L,M,V} \subset \cup_{A \in \cA} \FS_{L,M}(A)$ and of size $|\cA| \le 2^{|V|} \cdot \exp(CLd^{-3/2} \log^{3/2}d + CM \log d)$. Thus, the first part of the lemma follows by applying Lemma~\ref{lem:prob-of-approx-enhanced} and repeating a computation similar to that in~\eqref{eq:prob-of-breakup-with-LM}.
	
	The second part follows by summing over $L$ and $M$ and using Lemma~\ref{lem:minimal-breakup-by-diameter}.
\end{proof}

\begin{lemma}\label{lem:diam-witness}
Let $U,V \subset \Z^d$ be finite and assume that $U \cup V$ is connected. Then for any $u,v \in U \cup V$ there exists a path $p$ from $u$ to $v$ of length at most $\diam^+ U_p + \diam^+ V_p$, where $U_p$ and $V_p$ are the union of connected components of $U$ and $V$ which intersect $p$.
\end{lemma}
\begin{proof}
Let $\mathcal{W}$ be the collection of connected components of $U$ and $V$.
Consider the graph $G$ on $\mathcal{W}$ with $W,W' \in \mathcal{W}$ adjacent if and only if $\dist(W,W') \le 1$.
Note that $G$ is connected.
Consider a simple path $q=(W_1,\dots,W_k)$ in $G$, where $u \in W_1$ and $v \in W_k$.
For each $1 \le i \le k-1$, let $u_i \in W_i$ and $v_i \in W_{i+1}$ be such that $\dist(u_i,v_i) \le 1$.
Let $p$ be a path from $u$ to $v$ constructed by connecting $v_{i-1}$ to $u_i$ by a shortest-path for every $1 \le i \le k$ (where we set $v_0 := u$ and $u_k := v$).
Then the length of $p$ is at most $\sum_{i=1}^k (\diam W_i + 1)$. On the other hand, $\diam^+ U_p + \diam^+ V_p \ge \sum_{i=1}^k (\diam W_i + 3)$, and the lemma follows.
\end{proof}

For a four-section $K$ and a set $V \subset \Z^d$, define $K|_V=K'$ by $K'_0 := K_0 \cup V^c$ and $K'_i := K_i \cap V$ for $i \in \{1,2,3\}$. Note that if every connected component of $K_{123}$ is either contained in $V$ or disjoint from $V$, then $K|_V$ is a four-section. Moreover, if $K$ is adapted to $f$, then so is $K|_V$.

\begin{proof}[Proof of Lemma~\ref{lem:prob-of-joint-breakup-core}]
Denote $B := B(f,f',u)$ and assume that $\diam B \ge r$.
Observe that there exists a vertex $v \in B$ such that $\dist(u,v) \ge r/2$.
Denote $U(f) := K_{123}(f,B \cap T(f))$ and note that, since $B$ and $U(f)$ are odd, Lemma~\ref{lem:join-of-breakups} implies that $B \subset U(f) \cup U(f')$ and that $U(f) \cup U(f')$ is connected.
Thus, by Lemma~\ref{lem:diam-witness}, there exists a path $p$ from $u$ to $v$ of length $s \le \diam^+ V(f) + \diam^+ V(f')$, where $V(f)$ is the union of connected components of $U(f)$ which intersect $p$.
In particular, either $\diam^+ V(f)$ or $\diam^+ V(f')$ is at least $s/2$. Assume without loss of generality that $\diam^+ V(f) \ge s/2$.
Denote $K:=K(f,B \cap T(f))|_{V(f)}$ and observe that $K_{123}=V(f)$ and that $K \in \FS(f) \cap \FS_p$ by Lemma~\ref{lem:join-of-breakups}.
Thus, by a union bound on the choices of $p$ and by Lemma~\ref{lem:prob-of-breakup-many-points},
\[ \Pr\big(\diam B(f,f',u) \ge r \big) \le \sum_{s=\lceil r/2 \rceil}^\infty 2 (2d)^s 2^{s+1} e^{-cds/2} \le e^{-c'dr} . \qedhere \]
\end{proof}

\bibliographystyle{amsplain}
\bibliography{biblio}

\end{document}